\documentclass[12pt]{article}
\usepackage[margin=1.25in]{geometry}
\usepackage{mathtools, amssymb, mathrsfs, amsthm}
\usepackage[colorlinks = true, linkcolor = blue, citecolor = blue, urlcolor = blue]{hyperref}
\usepackage{indentfirst, bbm, accents, centernot, multirow, resizegather, enumitem, cleveref}
\usepackage{graphicx}
\graphicspath{{images/}}
\usepackage{caption}
\usepackage[dvipsnames]{xcolor}
\usepackage{titlesec}
\usepackage[utf8]{inputenc}
\usepackage[english]{babel}
\usepackage{natbib}
\usepackage{setspace}
\usepackage{algorithm2e}

\usepackage{booktabs}
\usepackage{threeparttable}
\usepackage{array}

\newtheorem{theorem}{Theorem}
\newtheorem{proposition}{Proposition} 
\newtheorem{corollary}{Corollary}
\newtheorem{lemma}{Lemma} 

\theoremstyle{definition}

\newtheorem{assumption}{Assumption}

\allowdisplaybreaks

\newcommand{\Var}{\textup{Var}}
\newcommand{\E}{E}
\newcommand{\R}{\mathbb{R}}

\DeclareMathOperator*{\argmin}{argmin}

\newcommand{\ubar}[1]{\text{\b{$#1$}}} %From http://tex.stackexchange.com/a/175862

\newtheorem{manualassumptioninner}{Assumption}
\newenvironment{manualassumption}[1]{%
  \renewcommand\themanualassumptioninner{#1}%
  \manualassumptioninner
}{\endmanualassumptioninner}

\crefname{assumption}{Assumption}{Assumptions}
\Crefname{assumption}{Assumption}{Assumptions}
\crefname{conjecture}{Conjecture}{Conjectures}
\Crefname{conjecture}{Conjecture}{Conjectures}

\crefname{manualassumptioninner}{Assumption}{Assumptions}
\Crefname{manualassumptioninner}{Assumption}{Assumptions}
\crefname{manualtheoreminner}{Theorem}{Theorems}
\Crefname{manualtheoreminner}{Theorem}{Theorems}

\title{How Much Weak Overlap Can Doubly Robust T-Statistics Handle?}

\date{September 27, 2026}

\author{Jacob Dorn\footnote{
I am grateful for suggestions from Kirill Borusyak, Colleen Carey, Xiaohong Chen, Rebecca Dorn, Kevin Guo, Andreas Hagemann, Edward Kennedy, Samir Khan, Michal Kolesár, Lihua Lei, Xinwei Ma, Francesca Molinari, José Luis Montiel Olea, Ulrich Müller, Zhuan Pei, Jörg Stoye, Yuya Sasaki, Stefan Wager, Yulong Wang, Larry Wasserman, and Dan Zeltzer, and participants at seminars at the American Health Econometrics Workshop, The University of California Irvine, The University of California San Diego, Cornell University, the University of Pennsylvania, and Stanford University. Artificial intelligence was used to suggest changes and identify potential errors. This material is based upon work supported by the National Science Foundation Graduate Research Fellowship Program under Grant No. DGE-2039656 and by grant T32 HS026116 from the Agency for Healthcare Research and Quality. Any opinions, findings, and conclusions or recommendations expressed in this material are those of the author and do not necessarily reflect the views of the National Science Foundation or the Agency for Healthcare Research and Quality. 
}}

\begin{document}
\bibliographystyle{apalike}

\maketitle

\onehalfspacing
%\doublespacing 

\begin{abstract}
    Propensity scores near zero or one can make standard confidence intervals undercover. Econometric theory has targeted the Average Treatment Effect (ATE) under such weak overlap using nonstandard estimators or intervals. Practice has instead favored fixed thresholding rules that restore asymptotic normality by discarding extreme-propensity units for whom treatment may matter most. I show that standard t-tests for the ATE remain valid even when overlap is so weak that every regular estimator has infinite asymptotic variance. The key is to threshold the usual doubly robust estimator aggressively enough to control bias among unthresholded observations, but mildly enough to control bias from thresholding. I provide sufficient conditions for such thresholds to exist, and propose an adaptive rule of thumb to achieve coverage without knowing the weak overlap degree. In a canonical application, fixed trimming excludes 18\% of the sample; the adaptive procedure estimates the resulting bias at five percent of the effect. 
\end{abstract}

\section{Introduction}

I study the problem of inference on treatment effects when the treatment of interest is highly predictable in certain covariate regions, so that the distributions of covariates among treated and control units only overlap weakly. Under such weak overlap, regression strategies require extreme extrapolation, reweighting strategies become volatile, and confidence intervals that ignore the issue can overreject.

Econometric theory has responded with either an unusual confidence interval, an unusual point estimator, or both \citep{RotheLimitedOverlap, XinweiMaRobustIPW, SasakiUra2022}.  Empirical practice has instead favored fixed thresholding rules that ensure asymptotic normality of standard point estimators, at the cost of potentially changing the target population substantially: in my empirical application, the standard \cite{CrumpEtAlOptimizePrecision} rule of thumb discards 18\% of the sample. Changing the target population is harmless when treatment effects are constant or few observations are dropped, but when treatment effects vary with the propensity score, the standard practice introduces persistent bias that leads confidence intervals to fail. Reporting the trimmed-population estimate as a proxy for the Average Treatment Effect (ATE) therefore maintains an implicit assumption that discarded units' treatment effects are not too different from the rest, which has so far been untestable without implementing a nonstandard procedure.

I show that standard Wald  (Estimate $\pm$ $1.96$ Standard Error) confidence intervals for the standard ATE can attain nominal coverage asymptotically. The point estimator is the doubly robust Augmented Inverse Propensity Weighted (AIPW) estimator with inverse propensity weights trimmed (discarded) or clipped (Winsorized) at a threshold tending to zero. The result holds even when every regular estimator's asymptotic variance is infinite: thresholding bounds the finite-sample variance, so that t-statistics with estimated standard errors can still attain standard asymptotics. The procedure only requires modifying the AIPW pipeline to impose a data-adaptive threshold, with an optional critical value adjustment available to improve finite-sample performance, allowing researchers to test the implicit assumption behind fixed trimming with tools already in use.

The intuition for this result is that AIPW is a debiasing estimator. AIPW can be derived by augmenting an Inverse Propensity Weighting (IPW) estimator using an outcome regression estimate to reduce the bias from errors in estimating inverse propensity weights. Thresholding systematically biases extreme weights in the direction of regularity. It is known that thresholding IPW aggressively can restore asymptotic normality at the cost of introducing first-order bias from thresholding \citep{XinweiMaRobustIPW,ChaudhuriHill}. AIPW's augmentation reduces this thresholding bias without affecting the asymptotic normality argument, and as a result confidence intervals based on AIPW's t-statistics can achieve nominal coverage. Nominal coverage requires the threshold to be aggressive enough to restore asymptotic normality and control the bias introduced by inverse propensity weight estimation in the unthresholded region, but the threshold must be mild enough to control the bias remaining after outcome regression in the thresholded region.

Applying these ideas in practice requires choosing a threshold, which determines how required accuracy is divided between the outcome and propensity models. I propose a rule that sets a bias bound equal to an estimated standard error, computed under homoskedasticity so that the rule only depends on estimated propensity scores and treatment indicators and can be fixed before outcomes are observed. The rule can encode beliefs about the relative accuracy of the two estimated nuisances, the propensity score and outcome regression functions, including the degree to which weak overlap makes a given outcome rate more difficult to achieve. The default asks the same consistency rate of each nuisance, yielding a provably unique threshold that is quick to compute. The default rule also returns the usual requirement that each nuisance converge at an $n^{-1/4}$ rate when the semiparametric bound is finite, while yielding a smaller threshold and demanding faster rates when overlap is too weak for regular estimators to attain a finite asymptotic variance.

%The resulting thresholded AIPW procedure is close to standard practice. Start by estimating an outcome regression and propensity score as in standard AIPW practice. Instead of applying a fixed trimming procedure at $0.1$, choose a threshold through the algorithm described in \Cref{subsec:ChoiceOfThreshold}: either use an equal-assumed-rate default, appropriate with parametric models or whenever the outcome and propensity rates should be comparable, or supply a more sophisticated guarantee, which is necessary when weak overlap substantially degrades the outcome consistency rate. Clip at the threshold to target the full-population ATE, computing the standard error and t-statistic as usual. Compare the t-statistic to a \cite{SatterthwaiteAdjustment}  critical value that adjusts for any excess kurtosis. Under the conditions of \Cref{thm:SecondOrderNuisances}, this procedure delivers asymptotically valid inference on the ATE. Under the much weaker conditions for thresholded AIPW to be consistent, the difference between the fixed-trimming and new point estimates delivers a consistent estimate of the fixed trimming bias. 

The resulting thresholded AIPW procedure is close to standard practice, but introduces a data-adaptive thresholding rule and an overlap-adaptive critical value. Start by estimating an outcome regression and propensity score as in standard AIPW practice. Instead of applying a fixed trimming procedure at $0.1$, choose a threshold through the algorithm described in \Cref{subsec:ChoiceOfThreshold}. Clip, rather than discard, at that threshold to target the full-population ATE, computing the standard error and t-statistic as usual. Compare the t-statistic to a \cite{SatterthwaiteAdjustment}  critical value that adjusts for any excess kurtosis. Under the conditions of \Cref{thm:SecondOrderNuisances}, this procedure delivers asymptotically valid inference on the ATE. Under the much weaker conditions for thresholded AIPW to be consistent, the difference between the fixed-trimming and new point estimates delivers a consistent estimate of the fixed trimming bias.

I establish sufficient conditions for the procedure to achieve nominal coverage as a function of the degree of overlap weakness. I quantify minimal overlap with \cite{XinweiMaRobustIPW}'s tail parameter $\gamma_0 > 1$: a guarantee on the number of finite inverse propensity weight moments, corresponding loosely to a density of propensity scores $\pi$ that behaves like $\pi^{\gamma_0-2}$ near zero.  Larger values of $\gamma_0$ correspond to stronger overlap guarantees, with strict overlap corresponding to infinite $\gamma_0$. A phase transition occurs as $\gamma_0$ crosses two, which corresponds to a uniform distribution. Under the regime I call \emph{somewhat weak overlap} ($\gamma_0 > 2$), the inverse propensity weights have finite variance, the propensity density vanishes at zero, and the semiparametric bound is finite and achievable.  Under the regime I call \emph{very weak overlap} ($\gamma_0 < 2$), inverse propensity weights have infinite variance, the propensity density is unbounded at zero, and the semiparametric bound is infinite. No estimator can achieve standard $\sqrt{n}$ asymptotics uniformly under very weak overlap \citep{KhanAndTamer}.

The results are familiar under somewhat weak overlap. If the usual product rate guarantee holds for $L^{p}$-norm errors with $p > 2 + 2/(\gamma_0-2)$ rather than the usual $p = 2$ norm, then some sequence of thresholds delivers nominal coverage, and AIPW is $\sqrt{n}$-consistent and semiparametrically efficient. The sufficient norm relaxes to $L^2$ with strict overlap, and strengthens to sup-norm rates at the regime boundary. If the propensity score is estimated accurately enough, then with probability tending to one, no treated unit is thresholded and t-tests based on the usual unthresholded AIPW estimator will also be valid.

%The rate conditions and results are familiar under somewhat weak overlap, and nonparametric rates remain sufficient under very weak overlap. Under somewhat weak overlap, if the usual product rate guarantee holds for sup-norm rather than $L^2$-norm errors, then some sequence of thresholds delivers nominal coverage, and AIPW is $\sqrt{n}$-consistent and semiparametrically efficient. %The mechanism for the product rate is different from the mechanism under strict overlap: a faster propensity rate reduces the bias in the unthresholded region, allowing for choosing a smaller threshold that reduces the bias in the thresholded region. 
%If the propensity score is estimated accurately enough, then with probability tending to one, no treated unit is thresholded and the unthresholded and thresholded AIPW estimators coincide. 

Under very weak overlap, t-statistics remain asymptotically normal even though AIPW's asymptotic variance is unbounded, because the estimated standard error captures the additional variance from weak overlap. Validity does not require parametric consistency rates: phrased in terms of equal rates, the sufficient sup-norm rate of $n^{-1/(2 + \gamma_0)}$ recovers the familiar $n^{-1/4}$ rate at the regime boundary and only tightens to $n^{-1/3}$ at even the worst level of overlap. These sufficient rates are also necessary: any weaker rates would allow Wald confidence intervals to exclude the true ATE asymptotically. Still, the estimator converges more slowly than $n^{-1/2}$, as \citeauthor{KhanAndTamer} show it must, and potentially as slowly as $n^{-1/4}$.

These conditions are stated in terms of black-box consistency rates, but any degree of weak overlap makes a given outcome consistency rate more difficult to obtain: there must be regions with few treated observations and a slower regression consistency rate. Quantifying the added difficulty in terms of Hölder continuity, weak overlap of order $\gamma_0$ has a similar effect to degrading the level of Hölder continuity by a factor of $1 - 1 / \gamma_0$.  Combining the reduced outcome consistency rate with any reduced benefit from propensity rates yields a closed-form characterization of how much weak overlap doubly robust t-statistics can handle. With equal Hölder levels $d / 2 < \beta \leq d / \sqrt{2}$, t-statistics only have guarantees within the somewhat weak overlap regime. When $\beta > d / \sqrt{2}$, t-statistics are compatible with very weak overlap, although $\gamma_0$ near one requires unboundedly large $\beta$. Under Lipschitz continuity in one dimension, one only needs $\gamma_0 > 5/3$ moments: t-statistics remain valid well into the very weak overlap regime in which $\sqrt{n}$ consistency is impossible, without even requiring derivatives to exist.

In simulations with very weak overlap and nonparametric nuisance estimates, thresholded AIPW's t-tests achieve near-nominal coverage with only modest precision loss. Fixed trimming mechanically introduces a false negative bias and thresholded IPW displays its known first-order bias, while clipped AIPW exhibits only the predicted second-order bias: 95\% confidence intervals cover the true parameter in 93.5\% of simulations at $n=10{,}000$, and cover in 95.3\% of simulations at $n=100{,}000$. (Trimmed AIPW attains nominal coverage more slowly.) Coverage comes at a modest precision cost relative to fixed trimming, inflating the sampling standard deviation by at most 10\% at these sample sizes. The precision loss is mainly limited by the slowly diverging asymptotic variance, and is further narrowed by the fixed trimming inclusion criterion's dependence on the estimated propensity score.

I apply the clipped AIPW estimator to  \cite{ConnorsEtAl1996}'s data on right heart catheterization. The setting is the empirical application for \cite{CrumpEtAlOptimizePrecision}'s 10\% rule, which discards 18\% of the sample and disproportionately drops low-propensity patients. My proposed adaptive procedure chooses a threshold of $0.068$ for low-propensity patients, but clips rather than discards and estimates the ATE in the full sample. The robustness check is low-cost and delivers reassuring news: the estimated mortality harm rises from 5.8 to 6.1 percentage points, with the standard error only increasing from 1.7 to 1.8. The point estimates imply an average harm of 7.4 percentage points among discarded patients, consistent with some selection (but statistically indistinguishable from noise) that increases the ATE by about five percent. The exercise shows the adaptive procedure offers an easy robustness check on fixed trimming. The comparison also illustrates that fixed trimming bias is the product of the share dropped and the gap in treatment effects between retained and discarded units, so either a small effect gap or a small discarded share pushes towards robustness.

Weak overlap is a common and serious problem for statistical inference \citep{KhanAndTamer, DAmourEtAlOverlap, MaSasakiWang}. The standard inverse propensity estimator in the presence of weak overlap involves trimming: dropping samples with small propensity estimates to target a better-behaved subpopulation \citep{imbens_unconfoundedness_review, BaileyGoodmanBacon, GalianiEtAl}, typically following the 10\% rule of thumb from \cite{CrumpEtAlOptimizePrecision}, or clipping strategies that Winsorize weights \citep{IonidesClipping, LeeEtAlWeightTrimming}.\footnote{Awkwardly, the epidemiological literature sometimes refers to the Winsorization strategy as ``trimming.'' My results hold for either dropping or Winsorizing extreme propensities.} These and other proposals that reweight towards higher-precision populations \citep{YangAndDing, LiEtAlOverlapWeights, ContaminationBiasInLinearRegressions} allow standard Wald confidence intervals to cover some treatment effect, but respond to a local weak overlap failure with a global change to a nonstandard and often less interpretable estimand.\footnote{Practitioners who are interested in a weighted estimand from first principles, for example to capture individuals at the margin of participation in an underlying logistic model, can avoid many of the statistical issues generated by weak overlap.} An important theoretical literature has proposed novel point and confidence interval estimators with desirable properties for statistical inference on the standard average treatment effect under weak overlap \citep{RotheLimitedOverlap, ArmstrongKolesarBandwidthSnooping, ArmstrongKolesarFiniteSampleOptimal,JuEtAlTMLEPositivity, XinweiMaRobustIPW, HirshbergWager, HeilerKazak, SasakiUra2022, ma2023doubly, ChaudhuriHill, clivio2026deconfoundingscoresrepresentationlearning}. However, these proposals introduce new and less familiar estimators, and there has been little take-up by practitioners. I instead show that a standard estimator can provide valid confidence intervals on the standard ATE.

My analysis builds on \cite{XinweiMaRobustIPW}'s analysis of trimmed IPW with known propensity scores, adding black-box-rate analysis and quantifying the added difficulty of nonparametric outcome regression under weak overlap.  \cite{KhanUganderDoublyRobustTrimming} and \cite{ChaudhuriHill} also show that  thresholded IPW and AIPW can remain asymptotically normal with first-order bias. \cite{HeilerKazak} and \cite{HirshbergWager} show that under somewhat weak overlap, AIPW-type estimators can achieve oracle equivalence.   My regression analysis relates to the literature on regression with degenerate designs \citep{Stone1982, HallEtAlLocalLinear, GaiffasPointwise, mou2023kernelbased}, and most closely to \cite{mou2023kernelbased}, whose singularity rate for a first-order Sobolev space matches my minimax rate for the implied Hölder class.

The plan of the paper is as follows. \Cref{sec:Setting} presents the setting, clipped estimator, and a bias decomposition at the core of the results. \Cref{sec:MainTheoreticalResults} presents the main theory in terms of black-box rates (\Cref{subsec:BlackBoxRates}) and smoothness conditions (\Cref{sec:RegressionRates}). \Cref{sec:ParametricResultsAndThresholds} offers lessons for practice, \Cref{sec:Sims} presents the simulation, \Cref{sec:EmpiricalApplication} presents the empirical application, and \Cref{sec:Conclusion} concludes. Theorems are proved in the main text, as are any results stated in the main text that a theorem's proof relies on; I defer the remaining proofs, including proofs of some main-text claims, to the appendix.

\textbf{Notation}. $E_{P}[\cdot]$ and $E[\cdot]$ denote expectations under the maintained  distribution $P$ over data $Z$, and $\sup_{P \in A} B$ the supremum of $B$ over distributions $P \in A$ satisfying any maintained restrictions on the distribution and nuisance functions. I abuse notation and write $\psi$ for the target estimand under $P$, which in the technical analysis is the average potential outcome, and $P(E_n)$ for the probability of an event $E_n$ under the distribution $P$, leaving the number of draws $n$ implicit where possible. For nonnegative sequences $c_n, d_n$: $c_n \ll d_n$ and $d_n \gg c_n$ if for all $\epsilon > 0$, $c_n \leq \epsilon d_n$ for $n$ large enough; $c_n \lesssim d_n$ and $d_n \gtrsim c_n$ if $c_n \leq k d_n$ for some $k > 0$ and all $n$ large enough; and $c_n = \Theta( d_n )$ if there are $0 < k_1 \leq k_2 < \infty$ such that $c_n / d_n \in [ k_1, k_2 ]$ for all $n$ large enough. For a sequence of distributions $P_n$, I write $c_n = O_{P_n}( d_n )$ if for all $\epsilon > 0$ there is a $\delta > 0$ with $\limsup_{n \to \infty} P_n( | c_n | / d_n > \delta ) < \epsilon$ and $c_n = o_{P_n}( d_n )$ if the same conclusion holds for all $\epsilon, \delta > 0$. I use $o(d_n)$ and $O(d_n)$ in the usual sense; bounds on quantities that depend on estimated nuisances hold almost surely under \Cref{assum:NuisanceRates} below.   I use $\log$ to refer to the natural logarithm and $a \vee b$ and $a \wedge b$ to indicate $\max\{ a, b \}$ and $\min\{a, b\}$, respectively. I write $c_n = O\big( \operatorname{polylog}(d_n) \big)$ if there is a $k > 0$ such that $c_n = O\big( \log(d_n)^k \big)$, and similar for $O_{P_n}( \cdot )$.  I use $\sigma_n$, $\tilde{\sigma}_n$, and $\hat{\sigma}_n$ for the population, sample oracle, and sample estimated standard deviation of the estimator, $V_n = n \sigma_n^2$ for the oracle asymptotic variance at the current threshold, and $\hat{\phi}_n = \phi( Z \mid b_n, \hat{\eta} )$ and $\phi_n = \phi( Z \mid b_n, \eta )$ for the feasible and oracle scores. I write $\| f \|_{L^p(P)} = E_{P}[ |f|^{p} ]^{1/p}$ for finite $p$ and $\| f \|_{L^\infty(P)} = \lim_{p \to \infty} \| f \|_{L^p(P)}$. Following a multivariate version of \cite{TsybakovBook2009}, $f \in \Sigma(\beta, L)$ if its order-$\ell = \max \{ z \in \mathbb{Z} \cap [0, \beta) \}$ derivatives are well-defined and $D^{\alpha} f = \frac{\partial^{|\alpha|} f}{\partial x_{1}^{\alpha_1} \partial x_{2}^{\alpha_2} \cdots}$ satisfies $| D^{\alpha} f(x) - D^{\alpha} f(x') | \leq L \| x - x' \|^{\beta - \ell}$ for all $x,x'$ and every multiindex $\alpha$ with $| \alpha | = \ell$.

\section{Setting, Framework, and Intuition}\label{sec:Setting}

This section presents the setting and weak overlap tail parameter that I work under, as well as the clipped estimator that I study.

\subsection{Setting}

I follow the standard semiparametric model for a binary treatment. I assume the econometrician has access to $n$ i.i.d. samples $Z$ of data $(X, D, Y)$, where $X \in \R^d$ are covariates, $D \in \{0, 1\}$ is a binary treatment, and $Y \in \R$ is an observed outcome.

For simplicity, I focus the theoretical analysis on estimating the average potential outcome (APO) $\psi(P) = E_{P}[E[Y \mid X, D=1]]$. The average treatment effect  $E_{P}\big[ E_{P}[Y \mid X, D=1] - E_{P}[Y \mid X, D=0] \big]$ follows as a corollary, and  the average treatment effect on the treated $( E[ Y - E_{P}[E[Y \mid X, D=0]] ] ) / E_{P}[ D ]$ follows without too much difficulty.

I assume that the propensity score $e(X) = P(D = 1 \mid X)$ is in $(0, 1)$ almost surely, so that the APO can be identified as $\psi(P) = E_{P}[D Y / e(X)]$.  I refer to regions of the covariate space in which the propensity can be arbitrarily close to zero as singularities. Strict overlap rules out singularities.  However, I do not bound the propensity score away from zero or one. Settings in which singularities are allowed are sometimes called ``limited overlap'' \citep{KhanAndTamer, ChaudhuriHill}.

I derive uniform convergence rates under bounds on overlap weakness. I follow \cite{XinweiMaRobustIPW}, who provide important building blocks in my analysis, and parameterize overlap weakness through a tail parameter $\gamma_0$. I extend their results to a model family $\mathscr{P}$ over distributions whose overlap is at least as strong as that tail parameter, in addition to some regularity conditions. 
\begin{assumption}[Distribution moments and tail bound] \label{def:AllowedDistributions}
Let $\mathscr{P}$ be a nonempty family of distributions, and write $e(X) = P(D = 1 \mid X)$ and $\mu(X) = E_{P}[Y \mid X, D=1]$. Then every $P \in \mathscr{P}$ is a distribution over $(X, D, Y) \in \R^d \times \{ 0, 1 \} \times \R$ satisfying the following conditions for some $q > 3, M, \sigma_{\min}, C > 0 $, and $\gamma_0 > 1$:
\begin{enumerate}[label=(\alph*), itemsep=-0.5ex, topsep=-0.5ex]
    \item \emph{Conditional moments}. $\E[| Y |^q \mid X, D=1] \leq M^q < \infty$ almost surely. \label{def:ConditionalMoments}  
    \item \emph{Unconditional moments}. $Var( \mu(X) ) \leq M$. \label{subsef:BoundedVarMu}
    \item \emph{Residuals}. $\Var(Y \mid X, D) \geq \sigma_{\min}^2$. \label{subdef:Residuals}
    \item \emph{Propensity tail}.  $P(e(X) \leq \pi) \leq C \pi^{\gamma_0 - 1}$ for all $\pi \in [0, 1]$.\label{item:PropensityTail}
    %\item \emph{Nuisance rates}.  $\E[\| \hat{\mu}^{(-k)}_n - \mu \|_{L^{\infty}(P)}] \leq \delta_n r_{\mu,n}$ and $\E[\| \hat{e}_n^{(-k)} - e \|_{L^{\infty}(P)}] \leq \delta_n r_{e,n}$ where $r_{\mu,n}, r_{e,n} \downarrow 0$.
\end{enumerate}
\end{assumption}

\Cref{def:AllowedDistributions} is a uniform generalization of \cite{XinweiMaRobustIPW}'s slowly varying tails assumption. Assumptions \ref{def:AllowedDistributions}\ref{def:ConditionalMoments} through \ref{def:AllowedDistributions}\ref{subdef:Residuals} are regularity conditions that rule out cases like perfectly predictable outcomes.

\Cref{def:AllowedDistributions}\ref{item:PropensityTail} provides the substantial restriction on $\mathscr{P}$: overlap may be weak in the sense that $\gamma_0$ is finite, but there is some minimal $\gamma_0$ and $C$ that provides an upper bound on the inverse propensity's tail behavior. Intuitively, the tail parameter $\gamma_0$ represents the smallest order at which a moment of $D / e(X)$ can be infinite.  $\gamma_0 \leq 1$ would allow a mass point of $e(X)$ at zero.

I distinguish between three cases of the degree of overlap weakness. I follow \cite{HeilerKazak} and use \emph{strict overlap} to refer to the case in which  the propensity score is bounded away from zero almost surely. Under strict overlap,  \Cref{def:AllowedDistributions}\ref{item:PropensityTail} holds for any finite $\gamma_0 > 1$, and most results here hold after replacing $\gamma_0$ with infinity.

I use \emph{somewhat weak overlap} to refer to the case that the inverse propensity distribution may be heavy-tailed, but the inverse propensity weight still has strictly more than two moments, which corresponds to $\gamma_0 > 2$.\footnote{Distributions satisfying  somewhat weak overlap are sometimes said to satisfy ``strict overlap'' or  ``overlap'' \citep{HeilerKazak, brunssmith2024augmentedbalancingweightslinear}, and distributions that exhibit very weak overlap are sometimes called ``heavy tailed'' \citep{ChaudhuriHill}.}  Heuristically, somewhat weak overlap is the case in which $P( e(X) \leq \pi )$ is positive near zero, but goes to zero at an $o(\pi)$ rate. Under somewhat weak overlap, the semiparametric efficiency bound is finite and known to be achievable \citep{ChenEtAlSemiparametricEfficiency, HirshbergWager, HeilerKazak}.

I use \emph{very} weak overlap to refer to the case in which $\gamma_0 < 2$ (and no other $C', \gamma_0' \geq 2$ would apply to any $P \in \mathscr{P}$). Under very weak overlap, the semiparametric bound is infinite, so that no estimator can achieve $\sqrt{n}$ consistency uniformly \citep{KhanAndTamer}, and IPW with even known propensity scores may fail to be asymptotically normal \citep{XinweiMaRobustIPW}.  A phase transition occurs at $\gamma_0 = 2$, in which case logarithmic terms can emerge.

\begin{figure}
    \centering
    \includegraphics[width=0.8\textwidth]{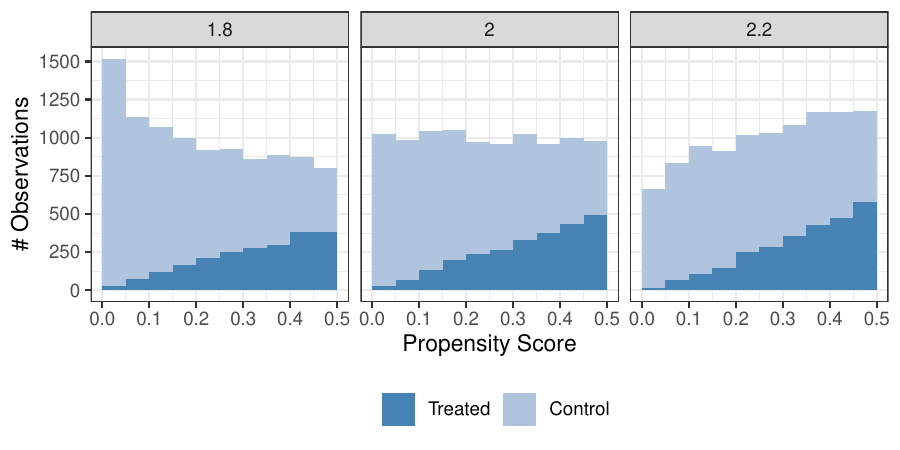}
    \caption{Simulations of 10,000 observations of $e(X)$ with $P(e(X) \leq \pi) = \pi^{\gamma_0 - 1}$ for increasing values of $\gamma_0$.}
    \label{fig:ExampleValuesOfGamma}
\end{figure}

\Cref{fig:ExampleValuesOfGamma} illustrates behavior for simulated data with various values of $\gamma_0$. When the propensity score $e(X)$ has a well-defined density, $\gamma_0 = 2$ corresponds to a roughly uniform distribution of propensity scores \citep{XinweiMaRobustIPW}. When $\gamma_0$ is above two, the density of propensity scores tends to zero at zero; when $\gamma_0$ is below two, the density of propensity scores can tend to infinity at zero. Heuristically, there are never too many treated observations with very small propensity scores, but $\gamma_0$ governs the degree to which there can be many untreated observations with very small propensity scores.

In many closed-form results, I assume that $\gamma_0$ is both an inverse propensity tail upper bound and lower bound. 

\begin{manualassumption}{\ref*{def:AllowedDistributions}.LB}[Tail lower bound]\label{assum:TailLowerBound}
    There exists a $C' > 0$ such that for every $P \in \mathscr{P}$ and every $\pi \in [0, 1]$, $P(e(X) \leq \pi) \geq C' \pi^{\gamma_0 - 1}$. 
\end{manualassumption}
The combination of \Cref{def:AllowedDistributions} and \Cref{assum:TailLowerBound} is related to \cite{XinweiMaRobustIPW}'s slowly varying tails assumption, but rules out slowly-varying factors so that guarantees hold uniformly over $\mathscr{P}$. \Cref{assum:TailLowerBound} could be dispensed with in most results under somewhat weak overlap. However, under very weak overlap, \Cref{assum:TailLowerBound} prevents nature from placing a sequence of point masses that blow up estimation bias while having a limited effect on oracle variance. I present high-level coverage conditions without \Cref{assum:TailLowerBound}, and present specific rate requirements with the assumption.

\subsection{Clipped AIPW Estimator}

I study a cross-fit AIPW estimator. Cross-fitting is a common strategy for simplifying the analysis of Neyman-orthogonal estimators like AIPW \citep{doubleML}.  

The clipped AIPW estimator of the APO $\psi$ is:
\begin{align}
    \hat{\psi}_{clip}^{AIPW}(b) & = \frac{1}{n} \sum_{k=1}^K \sum_{i \in \mathcal{F}^k} \hat{\mu}^{(-k)}(X_i) + \frac{D_i (Y_i - \hat{\mu}^{(-k)}(X_i))}{\max\{\hat{e}^{(-k)}(X_i), b\}},\label{eq:DefinePhi}
\end{align}
where $\mathcal{F}^k$ is the set of observations $i$ randomly partitioned into fold $k$ and $\hat{\eta}^{(-k)} = \left( \hat{\mu}^{(-k)}, \hat{e}^{(-k)} \right)$ are the nuisance function estimates constructed only on observations in folds other than $k$. (For evaluating $\hat{\psi}_{clip}^{AIPW}(0)$, I take the convention $D / \max\{0, 0\} = 0$.) I shorthand $\phi(Z \mid b, \bar{\eta}) = \bar{\mu}(X) + \frac{D (Y - \bar{\mu}(X))}{\max\{\bar{e}(X), b\}}$, and write $\mathcal{G}_k$ for $\{X_i, D_i\}$ and $\{ Y_i \}_{i \notin \mathcal{F}^k}$. The unthresholded AIPW estimator is the special case of $b_n = 0$.  I consider the pooled variance $\hat{\sigma}_n^2(b) = \frac{1}{n} \bigg( \frac{1}{n} \sum_{k=1}^{K} \sum_{i \in \mathcal{F}^k} \phi\big( Z_i \mid b, \hat{\eta}^{(-k)} \big)^2 - \big( \hat{\psi}_{clip}^{AIPW}(b) \big)^2 \bigg)$, which I shorthand $\hat{\sigma}_n$; the Bessel-corrected $\hat{\sigma}_n^2 \times n / (n-1)$ is asymptotically equivalent. The AIPW estimator is sometimes called a debiased, double machine learning, or doubly robust estimator.

The theoretical results here would also apply to the trimmed AIPW estimator $ \hat{\mu}(X) + \frac{D (Y - \hat{\mu}(X))}{ \hat{e}(X) } \mathbf{1}\{ \hat{e}(X) \geq b_n \} $ with only slight modification. However, the trimmed estimator has larger bias and smaller oracle variance, both of which reduce finite-sample coverage. Further, a trimmed estimator requires notation to track whether trimming occurs over the inverse propensity weights $D / \hat{e}(X)$, in which case t-statistic validity for the ATE follows, or trimming occurs over the target population, in which case t-statistic validity may fail. The clipped estimator therefore allows for simpler terminology and better practical performance.

\subsection{Intuition for Main Result}\label{subsec:IntuitionForResult}

The intuition for why clipped AIPW t-tests can achieve nominal coverage is as follows. Intuitively, clipping at an aggressive rate $b_n$ that tends to zero slowly ought to be close enough to fixed clipping to achieve asymptotic normality. In fact, \cite{XinweiMaRobustIPW} show that an oracle clipped IPW estimator $\tilde{\psi}_{(Oracle)}^{IPW}(b_n) = \frac{1}{n} \sum_{i=1}^n D_i Y_i / \max\{ e(X_i), b_n \}$ with aggressive clipping can achieve asymptotic normality, at the cost of introducing first-order bias from shifting extreme inverse propensity weights towards regularity.\footnote{\citeauthor{XinweiMaRobustIPW} formally study $\frac{1}{n} \sum_{i=1}^n \mathbf{1}\{ e(X_i) \geq b_n \} D_i Y_i / e(X_i)$, but their results carry over for clipping immediately.} That argument also implies that oracle AIPW $\tilde{\psi}_{(Oracle)}^{AIPW}(b_n) = \frac{1}{n} \sum_{i=1}^n \mu(X_i) + D_i (Y_i - \mu(X_i)) / \max\{ e(X_i), b_n \}$ should achieve asymptotic normality. By inspection, the infeasible $\tilde{\psi}_{(Oracle)}^{AIPW}(b_n)$ would have zero finite-sample bias, so that there is hope for t-tests constructed based on the feasible $\hat{\psi}_{clip}^{AIPW}(b_n)$ to achieve nominal coverage.

There is a standard recipe for establishing t-test validity for a Neyman-orthogonal estimator like AIPW under strict overlap. First, one establishes that the infeasible $\tilde{\psi}_{(Oracle)}^{AIPW}$ is consistent and asymptotically normal in the sense that $\sqrt{n} (\tilde{\psi}_{(Oracle)}^{AIPW} - \psi(P)) \to^d \mathcal{N}(0, V)$ for some asymptotic variance $V$. Second, one establishes that the feasible $\hat{\psi}_{clip}^{AIPW}$ is oracle equivalent in the sense that $\hat{\psi}_{clip}^{AIPW} - \tilde{\psi}_{(Oracle)}^{AIPW} = o_{P}( \sigma_n )$, where $\sigma_n = n^{-1/2} \sqrt{Var_{P}(  \phi(Z \mid b_n, \eta) )}$ is the oracle estimator's standard deviation. Typically, oracle equivalence follows by showing second-order bias, $E\big[ \hat{\psi}_{clip}^{AIPW} - \tilde{\psi}_{(Oracle)}^{AIPW} \big] =  o_{P}( \sigma_n )$ and then appealing to some regularity conditions that control $Var( \hat{\psi}_{clip}^{AIPW} -  \tilde{\psi}_{(Oracle)}^{AIPW} )$. Finally, one establishes that there is a feasible variance estimator $\hat{V}_n$ such that $\hat{V}_n / V \to^P 1$, so that $\sqrt{n} (\hat{\psi}_{clip}^{AIPW} - \psi(P)) / \sqrt{\hat{V}_n} \to^d \mathcal{N}(0, 1)$ and t-tests achieve asymptotically nominal coverage.

Under weak overlap, each step of the standard recipe becomes more stringent. \cite{XinweiMaRobustIPW} show that oracle asymptotic normality is not immediate: under very weak overlap, $\tilde{\psi}_{(Oracle)}^{AIPW}(0)$ fails to be asymptotically normal. Under feasible conditions on $b_n$, asymptotic normality can be restored. 
\begin{proposition}[Oracle asymptotic normality]\label{prop:OracleAsympNormal}
    Suppose either (i) $\gamma_0 > 2$ or (ii) $b_n$ satisfies $n^{-1/2} \ll b_n \ll 1$. Let $\sigma_n$ be the oracle estimator's population standard deviation. Then the oracle clipped AIPW estimator has uniform convergence to a normal distribution in the sense that $$\limsup_{n \to \infty} \sup_{P \in \mathscr{P}} \sup_{t \in \R} \left| P \left( \sigma_n^{-1} \left( \tilde{\psi}_{(Oracle)}^{AIPW}(b_n) - \psi(P) \right) \leq t \right) - \Phi(t) \right| = 0.$$
\end{proposition}

Asymptotic normality of oracle AIPW is known under a weaker condition $b_n \gg n^{-1/\gamma_0}$ \citep{XinweiMaRobustIPW}. However, I prove  asymptotic normality of feasible AIPW under conditions that imply $r_{e,n} \ll b_n$, so that the analysis here could only cover $b_n \lesssim n^{-1/2}$ with superparametric propensity estimates; if so, then the worst-case consistency rate in \Cref{cor:WorstCaseConsistency} could be degraded, potentially as far as $1 / polylog(n)$ in the $b_n = n^{-1/\gamma_0} polylog(n)^2$, $\gamma_0 \to 1$ limit.

\begin{proof}
    The result follows from a Berry-Esseen argument as follows. Let $P_n$ be a sequence of distributions in $\mathscr{P}$. By Appendix \Cref{lemma:BerryEsseenConditions} and the Berry-Esseen-Katz Theorem, for every $\epsilon \in (0, 1]$, there is a $C_{\epsilon} > 0$ such that the difference between the CDF of oracle clipped AIPW t-statistic $\frac{\tilde{\psi}_{(Oracle)}^{AIPW} - \psi_n}{\sigma_n} = \frac{ \sum \tilde{\phi}(Z \mid b_n, P_n) }{\sqrt{Var(\phi(Z \mid b_n, \eta))} \sqrt{n} }$ and the standard normal CDF $\Phi$ is uniformly bounded above by $\frac{C_\epsilon m(b_n, P_n, 2+\epsilon)}{V_n^{1+\epsilon/2} n^{\epsilon/2}}$. By Appendix \Cref{lemma:BerryEsseenConditions}.\ref{subitem:BEFinalRemainder}, this difference tends to zero. Therefore:
    \begin{align*}
        \limsup_{n \to \infty} \sup_{P \in \mathscr{P}} \sup_{t \in \R} & \left| P \left( \sigma_n^{-1} \left(  \tilde{\psi}_{(Oracle)}^{AIPW}(b_n) - \psi(P) \right) \leq t \right) - \Phi(t) \right| = \limsup_{n \to \infty} o(1) = 0. \qedhere 
    \end{align*} 
\end{proof}

Oracle equivalence becomes substantially more delicate. To achieve oracle equivalence, the finite-sample bias $E[ \hat{\psi}_{clip}^{AIPW} - \tilde{\psi}_{(Oracle)}^{AIPW} ]$ should go to zero faster than the oracle standard deviation. In fact, under very weak overlap with $\gamma_0 < 2$, the oracle asymptotic variance $V$ must be infinite \citep{KhanAndTamer}. In some sense this is a blessing of very weak overlap: $\sigma_n$ must go to zero more slowly than $n^{-1/2}$, so that it is easier to show $\hat{\psi}_{clip}^{AIPW} - \tilde{\psi}_{(Oracle)}^{AIPW} = o_{P}( \sigma_n )$. However, inverse propensity estimation can become extremely delicate, making oracle equivalence harder.

To formalize the added difficulty of controlling bias under weak overlap, write $U = \mathbf{1}\{ e(X) \geq b_n \}$ as an indicator that an observation should not be clipped, and decompose 
\begin{align}
\begin{split}
    \hat{\psi}_{clip}^{AIPW} - \tilde{\psi}_{(Oracle)}^{AIPW} & = 
        \frac{1}{n} \sum \left( \frac{D}{\max\{ \hat{e}(X), b_n \}} - \frac{D}{\max\{ e(X), b_n \}} \right) (Y - \mu(X))  \\
        & \ \ + \ ( \hat{\mu}(X) - \mu(X) ) \frac{e(X) - D}{\max\{ \hat{e}(X), b_n \}}
    \end{split}\label{eq:BiasDecompNice} \\
    &  + \frac{1}{n} \sum \phantom{(1-}U\phantom{)} ( \hat{\mu}(X) - \mu(X) ) \frac{\max\{ \hat{e}(X), b_n \} - e(X)}{\max\{ \hat{e}(X), b_n \}}  \label{eq:BiasDecompUnclipped} \\
    & + \frac{1}{n} \sum (1-U) ( \hat{\mu}(X) - \mu(X) ) \left( 1 - \frac{e(X)}{\max\{ \hat{e}(X), b_n \}} \right). \label{eq:BiasDecompClipped}
\end{align}
Line (\ref{eq:BiasDecompNice}) is mean-zero under cross-fitting and is $o_P( \sigma_n )$ under minor conditions.

Line (\ref{eq:BiasDecompUnclipped}) is similar to a classic product of errors and essentially captures the bias in the unclipped region. The term requires the regression error $\hat{\mu}(X) - \mu(X)$ to be small enough to control the bias from estimating inverse propensity weights as large as $b_n^{-1}$. This product term can be controlled by making $\hat{\mu}(X) - \mu(X)$ small, making $\hat{e}(X) - e(X)$ small, and by choosing a large threshold $b_n$ that makes this product term close to a strict overlap product of errors. If $b_n$ is large enough, and the nuisance estimates $\hat{\mu}$ and $\hat{e}$ are accurate enough, then line (\ref{eq:BiasDecompUnclipped}) will be $o_P( \sigma_n )$ under plausible conditions.

Line (\ref{eq:BiasDecompClipped}) is new and essentially captures the bias in the clipped region. The term is not orthogonalized by a product of errors: $E[1 - e / (\hat{e} \vee b_n) \mid X]$ can be anywhere between zero and one. Instead, the term is orthogonalized by a novel strategy: the probability of a given observation being thresholded tends to zero with $b_n$. If $b_n$ is small enough, and the regression estimate $\hat{\mu}$ is accurate enough, then line (\ref{eq:BiasDecompClipped}) will be $o_P( \sigma_n )$ under plausible conditions.

Taken together, there is a simple heuristic for why AIPW t-tests can cover even under very weak overlap. If $b_n$ is large enough, then line (\ref{eq:BiasDecompUnclipped}) will be second-order. If $b_n$ is small enough, then line (\ref{eq:BiasDecompClipped}) will be second-order. If the outcome regression errors $\hat{\mu}(X) - \mu(X)$ and the propensity errors $\hat{e}(X) - e(X)$ are small enough, then there will be a Goldilocks region of thresholds $b_n$ in which both terms are second-order. Under regularity conditions, the plug-in standard error estimate $\hat{\sigma}_n$ will satisfy $\frac{\hat{\sigma}_n}{\sigma_n} \to^P 1$, and feasible t-tests will achieve nominal coverage asymptotically.

\section{Main Theoretical Results}\label{sec:MainTheoreticalResults}

This section derives and interprets the rate requirements for Wald confidence intervals to cover asymptotically. \Cref{subsec:BlackBoxRates} derives results in terms of black-box nonparametric rates, as well as some associated corollaries on efficiency under somewhat weak overlap and degraded consistency rates under very weak overlap. \Cref{sec:RegressionRates} proves the achievability of these black-box regression rates, characterizing the degradation of outcome consistency rates under weak overlap.

I summarize the results in \Cref{tab:SummaryOfBehavior}. Under somewhat weak overlap, thresholded and unthresholded AIPW with known nuisance functions remain semiparametrically efficient and $\sqrt{n}$-consistent, and the traditional product of errors nuisance condition is sufficient with a modification to an $L^\infty$ on errors. Under very weak overlap, clipped AIPW achieves a slower consistency rate and the required black-box nuisance rates are more stringent even with $L^\infty$-norm guarantees. Both cases make outcome regression more difficult, but never so difficult as to require parametric assumptions.

\begin{table}[!ht]
    \caption{Summary of degradation of asymptotic behavior and requirements as overlap is permitted to be increasingly weak under either $L^2$-norm or $L^\infty$-norm consistency guarantees. Black box requirements are conditions for a sequence of valid $b_n$ to exist. Weak overlap entries marked (Known) are either established in the prior literature that is cited in-text or can be derived as simple corollaries. (That note generally applies to the strict overlap results, where it is omitted for brevity.) 
    }
    \label{tab:SummaryOfBehavior}
    \centering
    \resizebox{\textwidth}{!}{\begin{tabular}{@{\extracolsep{5pt}}llll}
        \\[-1.8ex]\hline
        \hline \\[-1.8ex]
        \textbf{Overlap Phase} & \textbf{Strict $(\inf e(x) > 0)$} & \textbf{Somewhat Weak $(\gamma_0 > 2)$} & \textbf{Very Weak $(\gamma_0 < 2)$} \\ \cline{1-4} \\
        \textbf{Double Robustness} &  \\
        \multicolumn{1}{l}{P\ref{prop:Consistency}: Consistency Conditions ($L^2$)} & $r_{\mu,n} r_{e,n} \big/ b_n \to 0$ & $r_{\mu,n} r_{e,n} \big/ b_n \to 0$ & $r_{\mu,n} r_{e,n} \big/ b_n \to 0$  \\
        \multicolumn{1}{l}{P\ref{prop:Consistency}: Consistency Conditions ($L^\infty$)} & $r_{e,n} \to 0$ or $\frac{r_{\mu,n}}{b_n} \to 0$ & $\frac{r_{\mu,n} r_{e,n}}{b_n} r_{e,n}^{\gamma_0-1} \to 0$ & $\frac{r_{\mu,n} r_{e,n}}{b_n} \max\{ r_{e,n}, b_n \}^{\gamma_0-1} \to 0$  \\[-1.8ex] \\ \cline{1-4} \\
        \textbf{Oracle AIPW Asymptotics} \\
        \multicolumn{1}{l}{P\ref{prop:OracleAsympNormal}: Unthresholded Distribution} & Normal & Normal (Known) & Nonnormal (Known) \\
        \multicolumn{1}{l}{P\ref{prop:OracleAsympNormal}: Thresholded Distribution} & Normal & Normal (Known) & Normal \\
        \multicolumn{1}{l}{C\ref{cor:TtestsUnderSomewhatWeakOverlap}: Thresholded Efficient?} & Yes & Yes (Known) & No, Impossible (Known) \\ 
        \multicolumn{1}{l}{C\ref{cor:WorstCaseConsistency}: Thresholded Convergence Rate} & $n^{-1/2}$ & $n^{-1/2}$ (Known) & $n^{-1/2} b_n^{(\gamma_0-2) / 2}$ \\ \\[-1.8ex]  \cline{1-4} \\
        \textbf{Nuisance Requirements} \\
        \multicolumn{1}{l}{P\ref{prop:SufficientRates}, T\ref{thm:SecondOrderNuisances}: Black Box ($L^2$)} & $n^{1/2} r_{\mu,n} r_{e,n} \to 0$ & $n^{1/2} r_{\mu,n} r_{e,n}^{1 - 2 / (\gamma_0+1)} \to 0$ & $n^{1/2} r_{\mu,n} r_{e,n}^{1/(\gamma_0+1)} \to 0$ \\
        \multicolumn{1}{l}{P\ref{prop:SufficientRates}, T\ref{thm:SecondOrderNuisances}: Black Box ($L^\infty$)} & $n^{1/2} r_{\mu,n} r_{e,n} \to 0$ & $n^{1/2} r_{\mu,n} r_{e,n} \to 0$ & $n^{1/2} r_{\mu,n} r_{e,n}^{\gamma_0/2} \to 0$ \\
        \multicolumn{1}{l}{C\ref{cor:MinimalSmoothnessConditions}: Smoothness}  & $\frac{\beta_{\mu}}{2 \beta_{\mu} + d} + \frac{\beta_e}{2 \beta_{e} + d} > \frac{1}{2}$ & $\frac{\beta_{\mu} (1 - 1 / \gamma_0) }{2 \beta_{\mu} (1 - 1 / \gamma_0) + d} + \frac{\beta_e}{2 \beta_{e} + d} > \frac{1}{2}$ & $\frac{\beta_{\mu} (1 - 1 / \gamma_0) }{2 \beta_{\mu} (1 - 1 / \gamma_0) + d} + \frac{\beta_e \gamma_0 / 2}{2 \beta_{e} + d} > \frac{1}{2}$ \\ \\[-1.8ex]  \cline{1-4} \\
        \textbf{Regression Rates (Up to Polylog)} \\
        \multicolumn{1}{l}{T\ref{prop:GlobalRateAchievable}: Pointwise and uniform} & $n^{\frac{-\beta_{\mu}}{2 \beta_{\mu} + d}}$ & $n^{\frac{-\beta_{\mu} \left(1 - 1/\gamma_0 \right)}{2 \beta_{\mu} \left(1 - 1/\gamma_0 \right) + d}}$ & $n^{\frac{-\beta_{\mu} \left(1 - 1/\gamma_0 \right)}{2 \beta_{\mu} \left(1 - 1/\gamma_0 \right) + d}}$ \\ \\[-1.8ex]
        \hline  \hline \\[-1.8ex]
    \end{tabular}}
\end{table}

\subsection{Sufficient Black-Box Rates}\label{subsec:BlackBoxRates}

This subsection provides AIPW performance guarantees under black-box nuisance consistency rate guarantees. I prove that clipped AIPW can be oracle-equivalent and achieve asymptotically nominal t-test coverage. I show that stronger nuisance rate conditions can be necessary to achieve this result. I show that when these conditions hold, clipped (and potentially unadjusted) AIPW can achieve the semiparametric efficiency bound under somewhat weak overlap, and I characterize the degraded consistency rate under very weak overlap.

I require certain consistency rates on the estimates of the nuisance functions $\mu(X)$ and $e(X)$. I write the worst-case rates as $r_{\mu,n}$ and $r_{e,n}$. 
\begin{assumption}[Cross-fitting and rates]\label{assum:NuisanceRates}
    The nuisances $\hat{\mu}$ and $\hat{e}$ are estimated with cross-fitting with a fixed number of folds $K$: the observations $Z_i$ are partitioned into $K$ folds at random such that the distribution of the fold-$k$ nuisance function estimates $\{ \hat{\mu}_i, \hat{e}_i \}_{k_i = k }$ for each $k$ is degenerate conditional on the other-fold data $\{ Z_i \}_{k_i \neq k}$. If $n_k$ is the number of observations per fold, then $\inf_k n_k / \sup_k n_k \to 1$. Further, there exist constants $p_{\mu}, p_{e} \geq 2$ and $B_{\mu} > 0$ and sequences $r_{\mu,n}, r_{e,n} = O(1)$ such that for every $P \in \mathscr{P}$ and every fold $k$,  $\| \hat{\mu}^{(-k)}_n - \mu \|_{L^{p_{\mu}}(P)} \leq r_{\mu,n}, \| \hat{\mu}^{(-k)}_n \|_{L^{\infty}(P)} \leq B_{\mu}, \text{ and } \| \hat{e}_n^{(-k)} - e \|_{L^{p_{e}}(P)} \leq r_{e,n}$  almost surely. 
\end{assumption}
Often, such rate guarantees are only available on a high-probability event of training data. In this case, statements below about convergence in probability and distribution hold immediately, and statements below about moments hold conditional on this event. Under strict overlap, it is standard to assume that $\frac{1}{p_{\mu}} + \frac{1}{p_e} \leq 1$. The version here is slightly stronger, and is much weaker than the sup-norm assumption ($p_e = p_{\mu} = \infty$) that is standard when studying semiparametric estimators under irregular identification \citep{semenova2024aggregatedintersectionboundsaggregated}.

Under these rates, I bound the estimator bias in terms of some shorthand. I write $B_U$ and $B_C$ for the expectations of (\ref{eq:BiasDecompUnclipped}) and (\ref{eq:BiasDecompClipped}), respectively, conditional on other-fold data. I write $\zeta$ for the Hölder conjugate solving $\frac{1}{p_{\mu}} + \frac{1}{p_e} + \frac{1}{\zeta} = 1$. (If $p_{\mu} = p_{e} = 2$, then take $\zeta = \infty$.) The variable arises from applying Hölder's inequality to obtain bounds of the form $E\big[ (\hat{\mu} - \mu) (\hat{e} - e) z \big] \leq r_{\mu,n} r_{e,n} \| z \|_{L^{\zeta}(P)}$. I also write $\rho_n = P\big( |\hat{e}(X) - e(X)| \geq ( e(X) - b_n ) / 2 \big)$ as a convenient upper bound on the probability of an observation being clipped at the estimated propensity function. I write $w_n = \| 1 / (e \vee b_n) \|_{L^{\zeta}(P)}$ for the moment on inverse propensity weights $D / (e \vee b_n)^2$ after applying Hölder's inequality to translate the bias into a rate product. When $r_{e,n}$ tends to zero more quickly than $b_n^{1 + (\gamma_0 - 1) / p_e}$ and $\gamma_0 < 1 + \zeta$, then Appendix Lemmas \ref{lemma:BoundClipping} and \ref{lemma:InvPropMoment} show that $\rho_n^{1/\zeta} / b_n$ and $w_n$, respectively, can behave like $b_n^{(\gamma_0-1) / \zeta - 1}$: larger with weaker overlap (smaller $\gamma_0$), weaker nuisance norms (larger $\zeta$), or smaller thresholds ($b_n$).   
\begin{lemma}[Bias bound]\label{lemma:BiasBound}
    Suppose \Cref{def:AllowedDistributions} and \ref{assum:NuisanceRates} hold. Then, almost surely,  
    \begin{align*}
        \left| B_U \right| & \leq \min\left\{ r_{\mu,n}, (B_{\mu} + M) \rho_n^{1/p_{\mu}} \right\} r_{e,n} \rho_n^{1/\zeta} / b_n + 2 r_{\mu,n} r_{e,n} w_n \\
        \left| B_C \right| & \leq \min\left\{ (B_{\mu} + M) P_n( e(X) \leq b_n ), r_{\mu,n} P_n( e(X) \leq b_n )^{1 - 1/p_{\mu}} \right\}.   
    \end{align*}
\end{lemma}

\begin{proof}[Proof of \Cref{lemma:BiasBound}]
    Fix one fold. Write $d_{e}(X) = |\hat{e}(X) - e(X)|$ for the absolute value of the propensity error. Then:
    \begin{align*}
        |B_U| & = \left| E_{P_n}\left[ \mathbf{1}\{ e(X) \geq b_n \} \left( 1 - e \left/ \max\{ \hat{e}, b_n \} \right.  \right) (\hat{\mu}(X) - \mu(X)) \right] \right| \\
        & \leq b_n^{-1}  E_{P_n}\left[ \mathbf{1}\{ e(X) \geq b_n, d_e \geq (e(X) - b_n)/2 \} d_e  \left| \hat{\mu}(X) - \mu(X) \right| \right]  \\ 
        & \ \  + 2  E_{P_n}\left[ \mathbf{1}\{ e(X) \geq b_n, d_e \leq (e(X) - b_n)/2 \} d_e \left|  \hat{\mu}(X) - \mu(X) \right| / e \right]  \\ 
        & \leq \min\left\{ r_{\mu,n}, (B_{\mu} + M) \rho_n^{1/p_{\mu}} \right\} r_{e,n} \rho_n^{1/\zeta} / b_n + 2 r_{\mu,n} r_{e,n} w_n  \tag{L\ref{lemma:ConditionalBiasBound}, H\"older} \\
        |B_C| & = \left| E_{P_n}\left[ \mathbf{1}\{ e(X) < b_n \} \left( 1 - e \left/ \max\{ \hat{e}, b_n \} \right. \right)  (\hat{\mu}(X) - \mu(X)) \right] \right| \\
        & \leq E_{P_n}\left[ \mathbf{1}\{ e(X) < b_n \} \left| \hat{\mu}(X) - \mu(X) \right| \right]  \leq r_{\mu,n} P_n( e(X) \leq b_n )^{1 - 1/p_{\mu}}.
    \end{align*}
    
    The remaining claim holds after averaging across folds. 
\end{proof}

A standard result for the unthresholded AIPW estimator is double robustness: when $e(X)$ is bounded away from zero, unthresholded AIPW is consistent for $\psi$ if either $r_{e,n}$ or $r_{\mu,n}$ tends to zero. This result implicitly relies on $\hat{e}(X)$ not being too small: even under strict overlap, a fraction $( r_{e,n} / \inf e(X) )^{p_e}$ of data can have arbitrarily small estimated propensity scores. If $r_{e,n}$ goes to zero too slowly to offset the extreme weights ($w_n$) or clipping bias ($\rho_n^{1/\zeta} / b_n$), then consistent propensity estimates may not suffice for consistency. Explicitly acknowledging the existence of weak overlap introduces a subtlety to double robustness.

\begin{proposition}[Consistency] \label{prop:Consistency}
    Suppose $b_n$ satisfies $n^{-1/2} \ll b_n \ll 1$, the conditions of \Cref{assum:NuisanceRates} hold, $r_{\mu,n} r_{e,n} w_n \to 0$, and $\min\left\{ r_{\mu,n}, (B_{\mu} + M) \rho_n^{1/p_{\mu}} \right\} r_{e,n} \rho_n^{1/\zeta} / b_n  \to 0$. Then for all $\epsilon > 0$, $$\sup_{P \in \mathscr{P}} P\left( \left| \hat{\psi}_{clip}^{AIPW}(b_n) - \psi(P) \right| > \epsilon \right) \to 0.$$ 
\end{proposition}
When \Cref{assum:TailLowerBound} holds and polylog boundary terms do not emerge, Appendix Lemmas \ref{lemma:BoundClipping} and \ref{lemma:InvPropMoment} show that the fraction of clipped data can look like $\big( b_n + r_{e,n}^{p_e / (p_e + \gamma_0 - 1)} \big)^{\gamma_0-1}$, so that a natural sufficient condition is $r_{\mu,n} r_{e,n} \max\big\{ b_n, r_{e,n}^{p_e / (p_e + \gamma_0 - 1)} \big\}^{(\gamma_0-1)/\zeta} / b_n \to 0$.  Under $L^2$ guarantees,  $\zeta$ is infinite  and this condition collapses to $r_{\mu,n} r_{e,n} / b_n \to 0$ regardless of $\gamma_0$. Under $L^\infty$ guarantees, $\zeta = 1$, so that this condition becomes $r_{\mu,n} r_{e,n} \max\big\{ r_{e,n}, b_n \big\}^{\gamma_0-1} / b_n \to 0$. Either version implies the usual double robustness consistency condition $r_{\mu,n} r_{e,n} \to 0$, ensuring $r_{\mu,n} r_{e,n} w_n \to 0$ whether or not $w_n$ diverges.  Further observing $r_{\mu,n} r_{e,n} b_n^{\gamma_0-2} \to 0$ under somewhat weak overlap yields the formula in \Cref{tab:SummaryOfBehavior}.

\begin{proof}[Proof of \Cref{prop:Consistency}]
    Write the fold-$k$ estimate of $\psi$ as
    \begin{align*}
        \hat{\psi}_{clip}^{AIPW,(k)}(b_n) & \equiv \frac{\sum_{i=1}^n \mathbf{1}\{ k_i = k \} \phi\big( Z_i \mid b_n, \hat{\eta}^{(-k)} \big)}{\sum_{i=1}^n \mathbf{1}\{ k_i = k \}}.
    \end{align*}
    By \Cref{assum:NuisanceRates}, it is sufficient to show that for any fixed fold $k$, any $\epsilon > 0$, and any sequence of $P_n \in \mathscr{P}$, $ P_n\left( \left| \hat{\psi}_{clip}^{AIPW,(k)}(b_n) - \psi(P_n) \right| > \epsilon \right) \to 0.$

    By assumption, $|B_U| = o(1)$. Because $b_n \to 0$ and $r_{\mu,n}$ is bounded, $|B_C| \leq r_{\mu,n} P_n( e(X) \leq b_n )^{1-1/p_{\mu}} = o(1)$. By \Cref{lemma:BiasBound}, the bias $|B_C + B_U| = o(1)$.  Then,
    \begin{align*}
        \Var_{P_n}\left( \hat{\psi}_{clip}^{AIPW,(k)}(b_n) - \psi(P_n) \right) & = \frac{Var_{P_n}\left(  \phi\big( Z_i \mid b_n, \hat{\eta}^{(-k)} \big) \right)}{\sum_{i=1}^n \mathbf{1}\{ k_i = k \}} = O\left( \frac{1}{n b_n^2} \right) = o(1). 
    \end{align*}
     Therefore $\big| \hat{\psi}_{clip}^{AIPW,(k)}(b_n) - \psi(P_n) \big| = o_{P_n}(1)$, so that the desired result holds by averaging across folds. 
\end{proof}

My results for asymptotic normality and statistical inference will proceed under certain rate requirements. To understand the rate requirement, it is useful to write the square root of the second moment of the difference between the feasible and oracle scores as $m_n = \sqrt{E_{P}\left[ (\hat{\phi}_n - \phi_n)^2 \right]}$. T-statistic validity follows from proving second-order bias in the sense that $|B_U + B_C| \ll \sigma_n$, and from proving consistent standard errors where by a Cauchy-Schwarz argument, $m_n^2 \ll V_n$ suffices.  %Further, under any degree of weak overlap, the fixed-threshold oracle estimator's asymptotic variance $V_n = n \sigma_n^2$ is driven by $E_{P}[ D / \max\{ e(X), b_n \}^2 ]$.

All told, the sufficient conditions for t-statistics on the ATE to be valid are as follows:
\begin{assumption}[Sufficient rates for t-statistic validity]\label{assumption:ConsistencyRatesSufficient}
    Assumptions \ref{def:AllowedDistributions} and
    \ref{assum:NuisanceRates} hold. For any sequence of $P_n \in \mathscr{P}$, the following conditions hold for the regression error $r_{\mu,n}$, the propensity error $r_{e,n}$, and the threshold $b_n$:  
    \begin{enumerate}[label=(\alph*), itemsep=-0.5ex, topsep=-0.5ex]
        \item \emph{Usual conditions}. $r_{\mu,n}, r_{e,n} \to 0$ and $n^{1/2} r_{\mu,n} r_{e,n} \to 0$. \label{assum:StrictOverlapRates}
        \item \emph{Control of $B_C$}.  $\sigma_n^{-1} r_{\mu,n} P_n\big( e(X) \leq b_n \big)^{1 - 1/p_{\mu}} \to 0$. \label{assum:ControlOfClippedBias}
        \item \emph{Control of $B_U$}.  $\sigma_n^{-1} \big( r_{\mu,n}  r_{e,n}  w_n +  \min\big\{ r_{\mu,n}, (B_{\mu} + M) \rho_n^{1/p_{\mu}} \big\} r_{e,n} \rho_n^{1/\zeta} / b_n \big)  \to 0$. \label{assum:ControlOfUnclippedBias}
        \item \emph{Control of $m_n$}. $b_n^{-1 - 1/p_e} V_n^{-1/p_e} r_{e,n} \to 0$, and either $\| \hat{\mu} - \mu \|_{L^\infty(P_n)} = o(1)$ or $\big( b_n V_n \big)^{-1/p_{\mu}} r_{\mu,n} \to 0$. \label{assum:ControlOfSquaredDiff} %TODO: Consider weakening to V_n^{-1} r_{\mu,n}^2 \| e / (e \vee b_n )^2 \|_{L^{1 + 2 / (p_{\mu} - 2}(P)} \to 0, but would break a bunch of proofs and require rewriting 
    \end{enumerate}
\end{assumption}

Condition \ref{assum:StrictOverlapRates} is the standard rate condition under strict overlap. Condition \ref{assum:ControlOfClippedBias} is reasonably clear: the exponent ${1 - 1/p_{\mu}}$ emerges from applying Hölder's inequality to $| \hat{\mu} - \mu |$, and $\sigma_n$ captures any benefit to attainability from a slower consistency rate under very weak overlap. Condition \ref{assum:ControlOfUnclippedBias} separately bounds the bias from correctly unclipped observations (the $w_n$ term) and incorrectly clipped observations (the $\rho_n$ term). The core Goldilocks condition for a threshold $b_n$ to be valid is that $b_n$ should be small enough to ensure $|B_C| \ll \sigma_n$, but large enough to ensure $|B_U| \ll \sigma_n$.

Condition \ref{assum:ControlOfSquaredDiff} includes regularity conditions that prevent squared errors exploding due to extreme weights. The main condition can be written as $b_n^{-1} V_n^{-1} (r_{e,n} / b_n)^{p_e} \to 0$: even if the largest possible mass of observations are incorrectly clipped ($(r_{e,n} / b_n)^{p_e}$ by Markov's inequality), the second moment of those weights (loosely $E[D / \max\{e(X) , b_n\}^2 \mid X] \leq b_n^{-1}$) must be second-order relative to the oracle variance.  This condition can further constrain $b_n$ beyond condition \ref{assum:ControlOfUnclippedBias}. Under \Cref{assum:TailLowerBound}, a sufficient version of the condition is $b_n^{-1 - \min\{ \gamma_0 - 1, 1 \} / p_e} r_{e,n} \to 0$. Without condition \ref{assum:ControlOfSquaredDiff}, it is possible for each individual fold to produce valid t-statistics but the correlation across folds to break asymptotic normality. An interesting direction for future work is whether \cite{LuetkevdL}'s train-test-type approach can be adapted to derive weaker conditions.

These assumptions will ensure validity of t-statistics under any degree of weak overlap $\gamma_0 > 1$. However, weak overlap will make $P_n\big( e(X) \leq b_n \big)$, $w_n$, and $\rho_n$ potentially more extreme, making the conditions more difficult to achieve. The following result shows the net effect. 
\begin{proposition}[Existence of valid thresholds]\label{prop:SufficientRates}
    Suppose Assumptions \ref{def:AllowedDistributions}, \ref{assum:TailLowerBound}, and \ref{assum:NuisanceRates} hold for some $p_{\mu} = p_{e} = p$ and some $\gamma_0 \neq 1 + \zeta$. Suppose $n^{-1/2} \lesssim r_{\mu,n}, r_{e,n} \ll 1$ are such that $n^{1/2} r_{\mu,n}^{p} \to 0$ and $r_{\mu,n} r_{e,n}^{\min\big\{ \frac{(p-1)(\gamma_0-1) + p \max\{ 2 - \gamma_0, 0 \} / 2 }{p + \gamma_0 - 1}, 1 \big\}} = o\big( n^{-1/2} \big)$. Then there exists a sequence of $b_n$ satisfying $n^{-1/2} \ll b_n \ll 1$ such that \Cref{assumption:ConsistencyRatesSufficient} holds. %, so that Wald confidence intervals attain nominal coverage asymptotically. 
\end{proposition}
Under strict overlap, the usual required product rate condition is $n^{1/2} r_{\mu,n} r_{e,n} \to 0$. \Cref{prop:SufficientRates} shows that this rate is sufficient provided $\gamma_0 > 2$ and the nuisance guarantees hold in $L^p$-norm for $p > 2 + 2 / (\gamma_0 - 2)$ and $r_{\mu,n} \ll n^{-1/(2 p)}$. Otherwise, the propensity rate enters more slowly. (The case $p = 2 + 2 / (\gamma_0 - 2)$ corresponds to $\gamma_0 = 1 + \zeta$ and introduces a logarithm term that I omit for brevity.) The propensity rate enters more slowly because a faster propensity rate $r_{e,n}$ reduces the bias $B_U$, allowing for a smaller threshold $b_n$ that reduces $B_C$ and reduces the required outcome consistency rate $r_{\mu,n}$. However, the act of reducing $b_n$ also increases $\sigma_n$ and $\rho_n^{1/\zeta} / b_n$; on net, the $\rho_n$ effect outweighs the $\sigma_n$ effect, limiting the benefit available from faster propensity rates.  

%The key results are that by Appendix \Cref{lemma:InvPropMoment}, $w_n = \Theta\big( b_n^{\min\{\gamma_0-1-\zeta, 0\}/\zeta} \big)$, and one can choose a carefully selected $b_n \gg r_{e,n}^{p_e / (p_e + \gamma_0 - 1)}$ so that by Appendix \Cref{lemma:BoundClipping}, $\rho_n = O\big( b_n^{\gamma_0-1}  \big)$ without introducing first-order bias.  The case $\gamma_0 - 1 = \zeta$  

\Cref{prop:SufficientRates} traces out different behavior with different nuisance norms. With the most stringent sup-norm guarantees ($p_{\mu} = p_{e} = \infty$), the full sufficient condition $n^{1/2} r_{\mu,n} r_{e,n}^{\min\{ \gamma_0, 2 \} / 2} \log(1/r_{e,n})^{\mathbf{1}\{ \gamma_0 = 2\}} \to 0$ relaxes to the usual condition under somewhat weak overlap. Phrased in terms of equal rates, the sufficient sup-norm rate of $n^{-1/\min\{ 2 + \gamma_0, 4 \}}$ degrades from $n^{-1/4}$ at the regime boundary to $n^{-1/3}$ under the weakest allowed overlap. With $L^2$ rates, the usual product condition is insufficient for this condition for any finite $\gamma_0$; the stronger requirement $n^{1/2} r_{\mu,n} r_{e,n}^{\max\{ \gamma_0 - 1, 1 \} / (\gamma_0 + 1)} \to 0$ yields a shared rate of $n^{-1/4}$ as $\gamma_0 \to \infty$, of $n^{-3/8}$ at $\gamma_0 = 2$, and a more favorable shared rate of $n^{-1/3}$ as $\gamma_0 \to^+ 1$. The more favorable $L^2$ rate for weaker overlap within very weak overlap arises because increasingly weak overlap places an increasingly large mass of propensity scores below $b_n$, so that a fixed $L^2$ outcome guarantee corresponds to an increasingly strong uniform rate in the clipped region. Nevertheless, as I show later, this blessing of very weak overlap is likely outweighed by the added difficulty of outcome regression.

I now return to the lower-level rate conditions in \Cref{assumption:ConsistencyRatesSufficient} to provide the main theoretical result: if the nuisance functions are estimated sufficiently quickly and the threshold is chosen appropriately, then t-tests will be asymptotically well-calibrated. 
\begin{theorem}[(Slow) Asymptotic Normality] \label{thm:SecondOrderNuisances}
    Suppose $b_n$ satisfies $n^{-1/2} \ll b_n \ll 1$ and \Cref{assumption:ConsistencyRatesSufficient} holds.  Then the clipped AIPW estimator is oracle-equivalent in the sense that $\sigma_n^{-1} \big( \hat{\psi}_{clip}^{AIPW}(b_n) - \tilde{\psi}_{(Oracle)}^{AIPW}(b_n) \big) = o_{P_n}(1)$ uniformly in $\mathscr{P}$. 
    Further, clipped AIPW t-statistics are uniformly asymptotically standard normal: $$\lim_{n \to \infty} \sup_{P \in \mathscr{P}} \sup_{t \in \R} \left| P \left( \hat{\sigma}_n^{-1} \left(  \hat{\psi}_{clip}^{AIPW}(b_n) - \psi(P) \right) \leq t \right) - \Phi(t) \right| = 0.$$
    As a result,  for any fixed $\alpha \in (0, 1)$, the Wald confidence interval $\hat{\mathcal{C}}_n(\alpha) = \big[ \hat{\psi}_{clip}^{AIPW}(b_n) + z_{\alpha/2} \hat{\sigma}_n, \hat{\psi}_{clip}^{AIPW}(b_n) + z_{1-\alpha/2} \hat{\sigma}_n \big]$ has asymptotic coverage equal to $1-\alpha$. 
\end{theorem}

\Cref{thm:SecondOrderNuisances} is the core theoretical claim of this paper. The first result shows that thresholded AIPW is first-order equivalent to an oracle estimator with known nuisances: the effect of nuisance estimation error on the treatment effect estimate tends to zero faster than the standard deviation of the Oracle estimator. The remaining results leverage this first-order equivalence to show that feasible t-statistics are asymptotically standard normal and Wald confidence intervals constructed using these t-statistics are well-calibrated.

These results are standard for AIPW under strict overlap, but care is required to handle unbounded inverse propensities under weak overlap.  This is nontrivial even with sup-norm convergence ($p = \infty$), because under weak overlap, there is an asymptotically unbounded number of observations with arbitrarily large inverse propensities with even known nuisance functions. Nevertheless, the bias decomposition in \Cref{lemma:BiasBound} and the condition $V_n^{-1} m_n^2 \to 0$ accomplished by Condition \ref{assum:ControlOfSquaredDiff} ensures that the estimator is oracle equivalent and estimated standard errors are consistent, so that t-statistics are asymptotically standard normal and Wald confidence intervals attain nominal coverage asymptotically,  even under the very weak overlap case in which unthresholded AIPW fails to be asymptotically normal and no estimator can achieve uniform root-n consistency.

\begin{proof}[Proof of \Cref{thm:SecondOrderNuisances}]
    First, I show that the estimator is oracle equivalent. Fix one fold, leaving the index implicit, and let $P_n$ be a sequence of distributions in $\mathscr{P}$.

    Write $D(n) = \sigma_n^{-2} E_{P_n}\big[ \big( \hat{\psi}_{clip}^{AIPW}(b_n) - \tilde{\psi}_{(Oracle)}^{AIPW}(b_n) \big)^2 \mid \{ Z_i \}_{i \notin \mathcal{F}^1} \big]$. Then by a bias-variance decomposition, \Cref{lemma:BiasBound}'s bias control, Appendix \Cref{lemma:VarBound}'s variance control, and iterated expectations conditional on other-fold data,
    \begin{align*}
        D(n) & = \sigma_n^{-2} \left( B_C + B_U \right)^2 + \sigma_n^{-2} \Var_{P_n}\left( \hat{\psi}_{clip}^{AIPW}(b_n) - \tilde{\psi}_{(Oracle)}^{AIPW}(b_n) \right) = o(1). 
    \end{align*}
    Therefore $D(n) = o(1)$ and the estimator is oracle equivalent by averaging over folds.

    Second, I show that clipped AIPW t-statistics are asymptotically standard normal. Write $\hat{T}_n = \hat{\sigma}_n^{-1} \big( \hat{\psi}_{clip}^{AIPW}(b_n) - \psi(P_n) \big)$, $T_n^* = \sigma_n^{-1} \big( \tilde{\psi}_{(Oracle)}^{AIPW}(b_n) - \psi(P_n) \big)$, and $\tilde{T}_n = \sigma_n^{-1} \big( \hat{\psi}_{clip}^{AIPW}(b_n) - \psi(P_n) \big)$. Then,
    \begin{align*}
        \hat{T}_n & = T_n^* + \sigma_n^{-1} \left( \hat{\psi}_{clip}^{AIPW}(b_n) -  \tilde{\psi}_{(Oracle)}^{AIPW}(b_n) \right) + \tilde{T}_n \left( \sigma_n / \hat{\sigma}_n - 1 \right) = T_n^* \left( \sigma_n / \hat{\sigma}_n \right) + o_{P_n}(1),
    \end{align*}
    where the final equality uses oracle equivalence.  By \Cref{prop:OracleAsympNormal}, $T_n^* = O_{P_n}(1)$ uniformly, and by Appendix \Cref{lemma:estimated_variance_consistency}, $\sigma_n / \hat{\sigma}_n = 1 + o_{P_n}(1)$ uniformly. Therefore $\hat{T}_n = T_n^* + o_{P_n}(1)$ uniformly. Uniform asymptotic normality of $\hat{T}_n$ then follows from applying \Cref{prop:OracleAsympNormal} to $T_n^*$ and observing $\Phi(t)$ is continuous.

	Finally, I show Wald confidence intervals achieve nominal coverage asymptotically. Write $G_n$ for the CDF of $\hat{T}_n$ under $P_n$. For any fixed $\alpha \in (0, 1)$ and any sequence of $P_n \in \mathscr{P}$, $P_n\big( \psi_n \in \hat{\mathcal{C}}_n(\alpha) \big) = G_n( z_{1 - \alpha / 2} ) - G_n( z_{\alpha / 2} ) = 1 - \alpha + o(1)$. Therefore, the t-statistic-based confidence intervals have asymptotic coverage equal to $1-\alpha$. 
\end{proof}

\Cref{thm:SecondOrderNuisances} applies to a single potential outcome, but average treatment effects immediately follow as a corollary. 
\begin{corollary}[ATE Coverage]\label{cor:ATE}
    Let $\mu_d(X) = E[ Y \mid X, D=d ]$. Suppose the conditions of \Cref{thm:SecondOrderNuisances} hold for $\phi_1(Z \mid b, \bar{\eta}) = \bar{\mu}_1(X) + \frac{D (Y - \bar{\mu}_1(X))}{\max\{\bar{e}(X), b\}}$ with a sequence of thresholds $b_n^{(1)}$, and hold for $\phi_0(Z \mid b, \bar{\eta}) = \bar{\mu}_0(X) + \frac{(1-D)(Y - \bar{\mu}_0(X))}{\max\{1-\bar{e}(X), b\}}$ for a sequence of thresholds $b_n^{(0)}$ under a corresponding tail parameter $\gamma_0^{(0)}$ for $1-e(X)$. Then the clipped AIPW ATE estimator $\hat{\tau}_{clip}\big( b_n^{(1)}, b_n^{(0)} \big) = \frac{1}{n} \sum_{i=1}^n \phi_1(Z_i \mid b_n^{(1)}, \hat{\eta}) - \phi_0(Z_i \mid b_n^{(0)}, \hat{\eta})$ is oracle equivalent, its feasible t-statistic centered at the true ATE $\tau = E_{P}[ \mu_1(X) - \mu_0(X) ]$ converges uniformly to a standard normal distribution, and the Wald confidence interval $\hat{\mathcal{C}}_n^{\tau}(\alpha)$ has asymptotic coverage $1-\alpha$. 
\end{corollary}

Taken together, this subsection yields a remarkable result for practice. The distribution $P$ may place so much propensity mass near the origin that the semiparametric efficiency bound is infinite, the lower bound on the density of propensity mass near the origin can be so weak that identification nearly fails, and the nuisance estimator may be so poorly designed that it pushes all observations' estimated propensities towards the origin at a slower-than-parametric rate. Nevertheless, AIPW's debiasing is sufficiently powerful to ensure the validity of the simple t-test.

The remainder of this subsection establishes the necessity of faster nuisance rates, the feasibility of efficient estimation under somewhat weak overlap, and the lost asymptotic precision under very weak overlap.

\subsubsection{Need for Improved Rates}\label{subsec:InterpretingThese}

This section shows that the faster sup-norm rates from \Cref{prop:SufficientRates} are necessary: an $n^{-1/2}$ product in $L^2$ norm is insufficient under somewhat weak overlap, and no weaker sup-norm product condition than $n^{1/2} r_{\mu,n} r_{e,n}^{\gamma_0 / 2} \to 0$ suffices under very weak overlap.

If overlap is weak enough relative to the maintained nuisance norms, then the usual product of errors condition that $n^{1/2} r_{\mu,n} r_{e,n} \to 0$ will be insufficient for establishing t-test coverage for at least some thresholds. 
\begin{proposition}[Under weak overlap, faster rates can be necessary]\label{prop:NoBetterProduct} 
    Take some $\gamma_0 > 1, p_{\mu}, p_{e} \geq 2$ such that $\gamma_0 < 1 + \zeta$. Then there exists a sequence of $b_n \to 0$ satisfying $b_n n^{1/2} \to \infty$, a family $\mathscr{P}$, and cross-fit nuisance estimators $\hat{\mu}$ and $\hat{e}$ such that:
    \begin{enumerate}
        \item \emph{Overlap is not too weak}. $\mathscr{P}$ satisfies \Cref{def:AllowedDistributions} and \ref{assum:TailLowerBound} for this $\gamma_0$.\label{item:PNBP:Overlap}
        \item \emph{Nuisances satisfy usual rates}. Assumptions \ref{assum:NuisanceRates} and \ref{assumption:ConsistencyRatesSufficient}\ref{assum:StrictOverlapRates},\ref{assum:ControlOfSquaredDiff} hold. \label{item:PNBP:Crossfit}
        \item \emph{Wald inference fails}. For any fixed nominal coverage level $\alpha \in (0, 1)$, the Wald confidence interval $\hat{\mathcal{C}}_n(\alpha)$ based on AIPW with clipping at $b_n$ has the zero-coverage property $\sup_{P \in \mathscr{P}} P( \psi(P) \in \hat{\mathcal{C}}_n(\alpha)) \to 0$. \label{item:PNBP:Coverage}
    \end{enumerate}
\end{proposition}
Under sup-norm rates, $\zeta = 1$ and \Cref{prop:NoBetterProduct} only applies under very weak overlap. Under $L^2$ rates, $\zeta$ is infinite and \Cref{prop:NoBetterProduct} applies to any level of weak overlap. In between, as long as $\gamma_0 < 1 + \zeta$, then a stronger product rate is needed. \Cref{prop:NoBetterProduct} leaves open the minimal rate in general, as well as behavior when  $\gamma_0 > 1 + \zeta$. If $\gamma_0 > 1 + \zeta$ or the propensity error $r_{e,n}$ is assumed to tend to zero quickly, then appendix \Cref{lemma:ConditionalBiasBound} can establish a relatively optimistic, but subtle, bias guarantee based on controlling the probability of $e(X) \leq b_n$ and of $\hat{e}(X) \geq b_n, |\hat{e}(X) - e(X)| \geq (e(X) - b_n) / 2$. Such arguments are outside the scope of this paper.

\Cref{prop:NoBetterProduct} only establishes that some thresholds fail, not that all thresholds fail. The following result shows that at least under very weak overlap, when the sufficient product rate above fails, all thresholds can fail.  
\begin{proposition}[The very weak overlap sup-norm product condition is optimal]\label{prop:OptimalVWORate}
    Suppose $r_{\mu,n}$ and $r_{e,n}$ are sequences such that $1 \gg r_{\mu,n}, r_{e,n} \gtrsim n^{-1/2}$ and $n^{1/2} r_{\mu,n} r_{e,n} \to 0$. Suppose the conditions of Assumption \ref{def:AllowedDistributions} and \ref{assum:TailLowerBound} hold for some $\gamma_0 \in (1, 2)$. Fix $p_{\mu} = p_{e} = \infty$. Then (i) if $n^{1/2} r_{\mu,n} r_{e,n}^{\gamma_0/2} \to 0$ and \Cref{assum:NuisanceRates} hold,  there is a sequence of $b_n \to 0$  such that clipped AIPW t-tests are asymptotically well-calibrated, and (ii)  if $n^{1/2} r_{\mu,n} r_{e,n}^{\gamma_0/2} \to \infty$, then for all sequences of $b_n \to 0$, there are nuisance estimators such that the properties of \Cref{prop:NoBetterProduct} (other than possibly the $b_n$-dependent \Cref{assumption:ConsistencyRatesSufficient}\ref{assum:ControlOfSquaredDiff}) hold, including the zero asymptotic coverage property (\ref{item:PNBP:Coverage}). 
\end{proposition}
\Cref{prop:OptimalVWORate} establishes that the sufficient sup-norm condition $n^{1/2} r_{\mu,n} r_{e,n}^{\gamma_0/2} \to 0$ from \Cref{prop:SufficientRates} cannot be replaced by a larger exponent: whenever the product diverges, any threshold sequence risks having coverage tend to zero. %It is known from the strict overlap literature that the product rate $r_{\mu,n} r_{e,n} \ll n^{-1/2}$ is also optimal. 

\subsubsection{Efficiency and Consistency Rates}\label{subsec:SlowerConsistencyRate}

This section interprets the consistency rate: clipped AIPW (and potentially unthresholded AIPW) is semiparametrically efficient under somewhat weak overlap, but has a reduced consistency rate as slow as $n^{-1/4}$ under very weak overlap.

Under somewhat weak overlap, oracle unthresholded AIPW is semiparametrically efficient. I now show that under moderate conditions, thresholding is also unnecessary in this case.  
\begin{corollary}[Semiparametric efficiency and unthresholded inference under somewhat weak overlap]\label{cor:TtestsUnderSomewhatWeakOverlap}
    Suppose Assumptions \ref{def:AllowedDistributions} and \ref{assum:NuisanceRates} hold for some $\gamma_0 > 2$,  $r_{\mu,n}, r_{e,n} \to 0$, $n^{1/2} r_{\mu,n}^{p_{\mu}} \to 0$, and $$n^{1/2} r_{\mu,n} r_{e,n}^{\min\big\{ \frac{(\gamma_0 - 1) p_e}{p_e+\gamma_0-1} \frac{p_{\mu}-1}{p_{\mu}}, 1 \big\}} \log(1/r_{e,n})^{\mathbf{1}\{ \gamma_0 = 1 + \zeta \}} \to 0.$$ Let $AV(P) = Var_{P}\big( \phi(Z \mid 0, \eta) \big)$ be the semiparametric asymptotic variance bound for a distribution $P \in \mathscr{P}$ \citep{hahn1998Rate}. Then there is a sequence of $b_n \to 0$ such that the thresholded AIPW estimator achieves the semiparametric efficiency bound in the sense that $\hat{\psi}_{clip}^{AIPW}(b_n) = \mathbb{P}_n[ \phi(Z \mid 0, \eta) ] + o_{P_n}(n^{-1/2})$ uniformly.  If in addition $n^{1/\gamma_0} r_{e,n}^{p_e / (p_e + \gamma_0 - 1)} \to 0$, then for all sequences of distributions $P_n \in \mathscr{P}$, $\hat{\psi}_{clip}^{AIPW}(0) = \mathbb{P}_n[ \phi(Z \mid 0, \eta) ] + o_{P_n}(n^{-1/2})$ and the coverage of Wald confidence intervals for a fixed target $\alpha \in (0, 1)$ based on $\hat{\psi}_{clip}^{AIPW}(0)$ and $\hat{\sigma}_n(0)$ tends to $1-\alpha$. 
\end{corollary}
\Cref{cor:TtestsUnderSomewhatWeakOverlap} shows that it is possible for an unthresholded AIPW estimator to be asymptotically equivalent to an efficient estimator under weak overlap. The logic involves first verifying that there is a sequence of $b_n \to 0$ for which the clipped estimator is semiparametrically efficient, and then showing that the probability of there being a treated observation with a sufficiently small estimated propensity score to clip is asymptotically negligible. As the nuisance norms $p_e$ and $p_{\mu}$ become large, these conditions approach the standard $n^{1/2} r_{\mu,n} r_{e,n} \to 0$ condition, coupled with a requirement that $n^{1/\gamma_0} r_{e,n} \to 0$. An important subtlety is that the unthresholded estimator may have an unbounded, or even undefined, finite-sample variance. The key is that even if the unthresholded estimator's variance is unbounded, the probability of drawing a treated observation with unbounded weight must still tend to zero in probability.

Under very weak overlap, no estimator can achieve $\sqrt{n}$ consistency uniformly. I characterize the consistency rate of any oracle-equivalent estimator as follows. 
\begin{proposition}[Degraded consistency rate]\label{prop:ConsistencyRates}
    For any $b \in (0, 1]$, write $W(P, b) = E_{P}[D / \max\{e(X), b\}^2]$. There exist positive constants $c_{\min}$ and $c_{\max}$ such that $c_{\min} n^{-1} W(P, b) \leq \sigma_n^2 \leq  c_{\max} n^{-1} W(P, b)$ for all $b \in (0, 1]$ and $P \in \mathscr{P}$.
\end{proposition}
If the AIPW propensities were trimmed instead of clipped, $W(P, b)$ would be replaced by $\E_{P}\big[  D \mathbf{1}\{ e(X) \geq b \} \big/ e(X)^2 \big]$. Weaker overlap corresponds to larger values of $W(P, b)$ and slower consistency rates. Conditional on $P$, larger values of $b$ correspond to a smaller value of $W(P, b)$, faster oracle consistency, and greater asymptotic power.

\Cref{prop:ConsistencyRates} implies a worst-case consistency rate over distributions in $\mathscr{P}$. I focus on the case of very weak overlap, because \Cref{cor:TtestsUnderSomewhatWeakOverlap} shows that under somewhat weak overlap, clipped and traditional AIPW achieve the usual $\sqrt{n}$ consistency rate.  
\begin{corollary}[Worst-case consistency rate]\label{cor:WorstCaseConsistency} 
    Write $s_n^2(b) = n^{-1} \big( b^{\min\{\gamma_0 - 2, 0\}} + \log( 1/b )^{\mathbf{1}\{\gamma_0 = 2\}} \big)$. Then $\sup_{P \in \mathscr{P}} \sigma_n^2 = O\big( s_n^2(b_n) \big)$. If in addition \Cref{assum:TailLowerBound} holds, then $\sigma_n^2 = \Theta\big( s_n^2(b_n) \big)$ uniformly over $P \in \mathscr{P}$. 
\end{corollary}
The rate $s_n(b)$ is a worst-case consistency rate in $b_n$: every distribution in $\mathscr{P}$ achieves a consistency at least as fast as $n^{-1/2} b_n^{\min\{ \gamma_0 - 2, 0 \} / 2}$ for $\gamma_0 \neq 2$, and it is possible to find a distribution for which the consistency rate is no faster.  The combination of \Cref{cor:WorstCaseConsistency} and \Cref{thm:SecondOrderNuisances} yields a trade-off under very weak overlap: milder (smaller) thresholds $b_n$ yield laxer requirements on $\hat{\mu}$, but lead to larger variance and slower consistency. Without \Cref{assum:TailLowerBound}, it is possible for nature to place propensity mass at a point going to zero slowly, which has a similar effect on worst-case bias as placing mass at or below $b_n$, but has a near-ignorable effect on the oracle consistency rate $\sigma_n$.

\subsection{Sufficient Smoothness}\label{sec:RegressionRates}

Weak overlap makes outcome regression more difficult: there may be few treated observations in precisely the regions in which thresholded AIPW depends most acutely on outcome regression. In this section, I show that even somewhat weak overlap degrades the effective outcome smoothness for local polynomial estimators. The result will be a closed-form answer to how much weak overlap AIPW t-statistics can handle.

I characterize nonparametric regression rates under Hölder continuity. For convenience, I fix the covariates to be uniform over a specific hypercube in $\R^d$ with constant variance. 
\begin{assumption}[Hölder continuity and fixed domain]\label{assum:HolderSmoothnessAssumptions}
    $\mathscr{P}$ satisfies Assumption \ref{def:AllowedDistributions}, and for all $P \in \mathscr{P}$, $X \sim Unif([-1, 1]^d)$ and $Y \mid X, D \sim \mathcal{N}( D \mu_{P}(X) + (1-D) \mu_{P}'(X), \sigma_{\min}^2 )$, with $\mu, \mu'$ in the Hölder class $\Sigma(\beta_{\mu}, L)$ for some fixed $\beta_{\mu}, L > 0$.
\end{assumption}
Most of these assumptions are standard assumptions for studying local polynomial regression under strict overlap \citep{Stone1982}. However, weak overlap allows the density of treated observations to approach zero, in which case outcome regression near the resulting singularity becomes a problem of estimation under a degenerate design \citep{HallEtAlLocalLinear, GaiffasPointwise}. I assume normal outcomes in order to simplify the analysis.  I work within a Hölder class to study the effect of weak overlap, not because the Hölder class should motivate AIPW: many estimators besides AIPW can yield valid t-statistics under \Cref{assum:HolderSmoothnessAssumptions} with sufficient smoothness \citep{HiranoImbensRidder,newey_robins_crossfitting,ArmstrongKolesarFiniteSampleOptimal}, and higher-order estimators can have better statistical properties \citep{robins2008higherorder}. The distinct advantage of AIPW debiasing is in structure-agnostic models  \citep{balakrishnan2025fundamentallimitsstructureagnosticfunctional}.  Hölder continuity instead offers a one-parameter family in which the cost of weak overlap can be characterized exactly. It is known that in this case but under strict overlap, the optimal pointwise rate is $n^{\frac{-\beta_{\mu}}{2 \beta_{\mu} + d}}$, and the optimal uniform rate is $n^{-\beta_{\mu} / (2 \beta_{\mu} + d)} O\big( \operatorname{polylog}(n) \big)$.

I require that treated observations cannot concentrate in small regions. 
\begin{assumption}[No adversarial propensity concentration]\label{assum:NonTrivialConcentration}
    There are parameters $\rho, \nu > 0$ such that for all $h > 0$ small enough and all $P \in \mathscr{P}$ and $x_0 \in [-1, 1]^d$, $P( e(X) \geq \rho \sup_{\| x - x_0 \| \leq h} e(x) \mid  \| X - x_0 \| \leq h ) > \nu$.
\end{assumption}
\Cref{assum:NonTrivialConcentration} rules out nature choosing to add treatment mass on a thin curve that makes the local polynomial regression problem degenerate. This problem can be avoided if $\beta_{\mu} \leq 1$ or $d=1$ and $\beta_{\mu} \leq 2$: for example, \cite{HallEtAlLocalLinear} mainly study this case with $d = 2$. The problem can also be avoided from first principles if the propensity score is sufficiently smooth: $e(X)$ will look like its Taylor expansion near a singularity, and a sufficiently high-order propensity Taylor expansion must be nonzero under the global overlap bound. The companion note \cite{Dorn_Overlap_Minimax} focuses on minimax rates, but also shows that the specific sufficient condition for \Cref{assum:NonTrivialConcentration} is $\beta_e > d / (\gamma_0 - 1)$. In cases where \Cref{assum:NonTrivialConcentration} fails, it is likely some hypothetical degeneracy-robust estimator would attain the rates below, but the construction is outside the scope of this work.

In the worst case, weak overlap of order $\gamma_0 > 1$ plays a role equivalent to scaling the effective outcome smoothness downward by $\left( 1 - 1 / \gamma_0 \right)$. 
\begin{theorem}[Effective outcome smoothness under weak overlap]\label{prop:GlobalRateAchievable}
    Suppose $\gamma_0$ is finite and $\mathscr{P}$ satisfies Assumptions \ref{assum:HolderSmoothnessAssumptions} and \ref{assum:NonTrivialConcentration}. Define  $r_{\mu,n}^* = n^{\frac{-\beta^*}{2 \beta^* + d} }$, where $\beta^* = \beta_{\mu} \left(1 - 1 / \gamma_0 \right)$. Then there exists an estimator $\hat{\mu}(x)$ and a sequence $g_n = O\big( \operatorname{polylog}(n) \big)$ such that  for all $\epsilon > 0$, there is a finite $c(\epsilon)$ with $$\limsup_{n \to \infty} \sup_{P \in \mathscr{P}} P\left( \| \hat{\mu} - \mu \|_{L^{\infty}(P)} > c(\epsilon) g_n r_{\mu,n}^* \right) \leq \epsilon.$$
    Further, the estimator can be computed without knowledge of the overlap bound $\gamma_0$. 
\end{theorem}
The intuition for this result is that as $\gamma_0$ becomes smaller, the global overlap guarantee becomes weaker and nature can choose a point at which there are increasingly few treated observations nearby to use in regression. The mean-squared-error-optimal response to weak overlap is to choose a wider bandwidth than usual that reduces the variance at the cost of increasing the regression bias. The resulting rate penalty occurs even under somewhat weak overlap. The strict overlap optimal pointwise convergence rate $n^{-\beta_{\mu} / (2 \beta_{\mu} + d)}$ is recovered up to a polylog term as the limit as $\gamma_0 \to \infty$, and smaller levels of $\gamma_0$ correspond to slower consistency rates.

\begin{proof}
    For simplicity of notation, I proceed assuming $\mathscr{P} = \{ P \}$ has only one element, but the proof can easily be extended to an arbitrary sequence of distributions. 
    
    For each $n$, construct a grid of width $1/n$ with gridpoints $x_{n,j}$. Write the number of treated observations within $h$ of $x$ as $N_n(h \mid x) = \sum_{i=1}^n \mathbf{1}\{ D_i = 1, \| X_i - x \| \leq h \}$. For each gridpoint $j$ and every $0 \leq k \leq \log_2(n)$, write $h_k = 2 \sqrt{d} 2^{-k}$  and compute $N_{j,k} = N_n\big( h_k \mid x_{n,j} \big)$. Proceed on the high probability event $N_{j,0} = \sum_{i=1}^{n} D_i \gg \log(n)^2$ for all $j$. Write $k(j) = \max k : N_{j,k} \geq \log(n)^2 h_k^{-2 \beta_{\mu}}$. Construct $\hat{\mu}_{n,j}$ as the $\lfloor \beta_{\mu} \rfloor$-order local polynomial estimator of $\mu(x_{n,j})$ with uniform kernel and bandwidth $h_{k(j)}$. For all $x \in [-1, 1]^d$, write $j(x)$ to be the nearest $x_{n,j}$, tiebreaking arbitrarily, and take $\hat{\mu}(x) = \hat{\mu}_{n,j(x)}$. 

    Let the conditional outcome mean of gridpoint $j$ be $\mu_{n,j}$. Note that $\sup_{x \in [-1, 1]^d} \| \mu_{n,j(x)} - \mu(x) \| = O\big( n^{-\min\{1, \beta_{\mu}\}} \big) = o\big( r_{\mu,n}^* \big)$.  Then $\sup_{x \in [-1, 1]^d} \| \hat{\mu}(x) - \mu(x) \| \leq \max_{j} \| \hat{\mu}_{n,j} - \mu_{n,j} \| + o\big( r_{\mu,n}^* \big)$, so it suffices to show that $\max_{j} \| \hat{\mu}_{n,j} - \mu_{n,j} \| = O_{P_n}\big( r_{\mu,n}^* \operatorname{polylog}(n) \big)$. 
    
    By Appendix \Cref{lemma:LPBiasVar}, the $\hat{\mu}(x_{n,j}) \mid \{ X_i, D_i \}$ are distributed $\mathcal{N}\big( \bar{\mu}_{n,j}, s_{n,j}^2 \big)$ with $\big| \bar{\mu}_{n,j} - \mu(x_{n,j}) \big| \leq c_{\mu} h_{n,k(j)}^{\beta_{\mu}}$ and $s_{n,j} \leq c_s h_{n,k(j)}^{\beta_{\mu}} / \log(n)$, where $c_{\mu}$ and $c_{s}$ only depend on the parameters of Assumptions \ref{assum:HolderSmoothnessAssumptions} and \ref{assum:NonTrivialConcentration}. Then by a union bound, for every $t > 0$, 
    \begin{align*}
        P\left( \max_j s_{n,j}^{-1} | \hat{\mu}_{n,j} - \bar{\mu}_{n,j} | \geq t  \right) \leq \sum_{j} P\left( s_{n,j}^{-1} | \hat{\mu}_{n,j} - \bar{\mu}_{n,j} | \geq t \right) \lesssim n^d \exp(-t^2 / 2),
    \end{align*}
    with constants that depend only on $d$. As a result, $\max_j s_{n,j}^{-1} | \hat{\mu}_{n,j} - \bar{\mu}_{n,j} | = O_{P_n}\big( \sqrt{\log(n)} \big)$.  Then
    \begin{align*}
        \max_{j} \| \hat{\mu}_{n,j} - \mu_{n,j} \| & \leq \left( \max_{j} s_{n,j} \right) O_{P_n}\big( \sqrt{\log(n)} \big) + c_{\mu} \max_{j} h_{n,k(j)}^{\beta_{\mu}} = O_{P_n}\left( \max_{j} h_{n,k(j)}^{\beta_{\mu}} \right).
    \end{align*}
    By Appendix \Cref{lemma:LPUniformBounds}(i), $h_{n,k(j)} \lesssim \big( \log(n)^2 / n \big)^{1/(2 \beta_{\mu} + d \gamma_0 / (\gamma_0 - 1))}$. 
    Therefore $\max_{j} \| \hat{\mu}_{n,j} - \mu_{n,j} \|  = O_{P_n}\big( r_{\mu,n}^* \operatorname{polylog}(n) \big)$, completing the claim. 
\end{proof}

\Cref{prop:GlobalRateAchievable} shows that weak overlap can be interpreted as increasing the effective dimension by $d / (\gamma_0-1)$, or equivalently as degrading the effective smoothness rate from $\beta_{\mu}$ by a factor no worse than $(1 - 1 / \gamma_0 )$. By standard arguments, $r_{\mu,n}^*$ is the minimax pointwise rate, so that the rate is optimal up to a possible polylog factor.  The companion note \cite{Dorn_Overlap_Minimax} shows in fact the polylog term can be replaced by a constant for even the sup-norm rate: there can be a finite number of points with the worst-case order of weak overlap, but an infinite number of points would make the global degree of weak overlap even worse. However, the polylog term does not affect the achievability of t-statistic validity, so the minimax results are left to the note in the interest of space.

The result parallels
%first-order Sobolev space with integrability exponent two in one dimension, which the Sobolev Embedding Theorem shows 
\cite{mou2023kernelbased}'s analysis of the effect of weak overlap in a Sobolev space: their setting implies Hölder continuity of order $\beta_{\mu} = 1/2$;  my pointwise rate for $\beta_{\mu} = 1/2$, $d=1$ matches their pointwise rate at the singularity; and my degraded outcome consistency rate  can be interpreted as increasing the effective dimension from $d$ to $d \gamma_0 / (\gamma_0-1)$, which generalizes the \citeauthor{mou2023kernelbased} singularity exponent $\alpha$ to higher dimensions.  For example, the difficulty of estimating a twice continuously differentiable function when $\gamma_0 = 2$ is comparable to the difficulty of estimating a Lipschitz-continuous function under strict overlap.   In the limit in which $\gamma_0$ tends to one, the rate in \Cref{prop:GlobalRateAchievable} can become arbitrarily poor.

\Cref{prop:GlobalRateAchievable} yields sufficient smoothness assumptions for Wald confidence interval validity. 
\begin{corollary}[Sufficient smoothness conditions]\label{cor:MinimalSmoothnessConditions}
    Suppose Assumptions \ref{assum:TailLowerBound}, \ref{assum:HolderSmoothnessAssumptions}, and \ref{assum:NonTrivialConcentration} hold and there is a $\beta_{e} > 0$ such that $e(X) \in \Sigma(\beta_{e}, L)$ with $\frac{\beta_{\mu}}{2 \beta_{\mu} + d} + \frac{\beta_{e}}{2 \beta_{e} + d} > 1 / 2$. Let $\bar{\gamma}$ solve 
    \begin{align}
         \frac{\beta_{\mu}}{2 \beta_{\mu} + d \bar{\gamma} / (\bar{\gamma} - 1)} + \frac{\min\{ \bar{\gamma} / 2, 1\} \beta_e}{2 \beta_e + d} = 1 / 2. \label{eq:SomeClippingThresholdInSmoothnessLipschitz}
    \end{align}
    Then there is a sequence of nuisance estimators and thresholds that depend on $\beta_{\mu}, \beta_{e}$ but not $\bar{\gamma}$ or $\gamma_0$ such that for all $\gamma_0 > \bar{\gamma}$, the associated Wald confidence interval $\hat{\mathcal{C}}_n(\alpha)$ constructed using $\hat{\psi}_{clip}^{AIPW}(b_n)$ satisfies $\limsup_{n \to \infty} \sup_{P \in \mathscr{P}} \big| P( \psi(P) \in \hat{\mathcal{C}}_n(\alpha)) - (1-\alpha) \big| = 0$.
\end{corollary}
In one dimension, thresholded AIPW with Lipschitz-continuous conditional outcome mean can handle weak overlap of order $\gamma_0 > 4 / (4 - 1/\beta_e)$ (if $1/4 < \beta_e \leq 1/2$) or $\gamma_0 > (4 + 1/\beta_e) / 3$ (if $\beta_e > 1/2$). When the propensity function is infinitely-differentiable, thresholded AIPW with a  Lipschitz-continuous conditional outcome mean can handle very weak overlap of order $\gamma_0 > 2 \frac{d + 1}{d + 2}$. More generally, under very weak overlap, if $\beta_{\mu}, \beta_e > \frac{d \left( \sqrt{\gamma_0^2 + 4 \gamma_0 - 4} + 2 - \gamma_0 \right)}{4 (\gamma_0-1)}$, then it is feasible to achieve standard inference with thresholded AIPW.

This concludes the substantive theoretical analysis. In the next section, I use these results to infer lessons for empirical practice.

\section{Lessons for Empirical Practice}\label{sec:ParametricResultsAndThresholds}

This section leverages the theoretical analysis to consider misspecified parametric estimators, propose rules of thumb for empirical use, and discuss a critical value adjustment that modestly improves finite-sample coverage.

\subsection{Parametric Estimators and Misspecification}\label{subsec:AnalysisOfParametricCase}

When both nuisance functions are estimated nonparametrically, then consistency is achievable and AIPW is generally preferable under strict overlap. When both nuisance functions are estimated parametrically, then it is possible for one or both nuisance functions to be inconsistent and the choice of estimator may be ambiguous. I now provide some intuition on the two estimators when nuisance functions are estimated parametrically and through cross-fitting.

In this subsection, I consider IPW and AIPW with the same sequence of thresholds $b_n$ satisfying $1 \gg b_n \gg n^{-1/2}$. I write that a nuisance estimate $\hat{\eta}$ is consistent if it tends to the correct limit $\eta$, and I write that $\hat{\eta}$ is inconsistent otherwise. I assume that parametric nuisance estimators $\hat{\eta}$ achieve an $L^\infty$ error relative to a limiting nuisance function $\bar{\eta}$ that is on the order of $n^{-1/2}$.  For example, consider logit estimation of a propensity model of the form $\bar{e}(X) = \frac{exp(X' \beta)}{1 + \exp(X' \beta)}$ for a pseudo-true parameter $\beta$. If the support of $X$ is bounded, then an $n^{-1/2}$-consistent estimate of $\beta$ is sufficient to achieve $n^{-1/2}$-consistent estimation of $\bar{e}(X)$ everywhere.  However, weak overlap may emerge from unbounded tails, in which case the $L^\infty$ rate may not go to zero. Unbounded covariates are an important case in general. For example, \cite{XinweiMaRobustIPW} motivate weak overlap tails through the distribution of covariates under a logistic propensity model. Nevertheless, a careful treatment of parametric estimation of nuisances with unbounded covariates is outside the scope of this work.

The analysis above is easiest to extend when either both or neither nuisance function is consistent.  If both nuisance estimates are consistent, then the AIPW and IPW estimators will be consistent and will have variance on the same order, but the IPW estimator may have leading-order bias, while the AIPW estimator can permit simple t-statistic asymptotics under even very weak overlap. If both the propensity and outcome regression estimates are inconsistent, then both the IPW and AIPW estimators fail to be consistent, and as in the case of inconsistent nuisance functions with strict overlap, there is no general reason to prefer one or the other.

If the propensity estimates are inconsistent but the outcome regression estimates are consistent, then there is a clear reason to prefer AIPW to IPW. The Berry-Esseen argument for AIPW asymptotic normality applies to any fixed fold under cross-fitting if $b_n \gg n^{-1/2}$, so that both IPW and AIPW may be asymptotically normal. However, the propensity model could create spurious apparent weak overlap, so that AIPW's bias may go to zero more slowly than sampling error. Still, \Cref{prop:Consistency} shows that AIPW will remain consistent, as long as the threshold $b_n$ is aggressive enough. This robustness intuition is useful, because I apply parametric nuisance estimators in the application to right heart catheterization.

When the outcome regression estimate is inconsistent, there is no general reason to prefer IPW or AIPW: both IPW and AIPW can be viewed as instances of AIPW with an inconsistent outcome regression estimate. Under very weak overlap, the thresholding bias is of order $b_n^{\gamma_0 - 1}$, which will generally dominate the oracle $\sigma_n \sim n^{-1/2} b_n^{\gamma_0 / 2 - 1}$ as long as $b_n \gg n^{-1/\gamma_0}$, which \cite{XinweiMaRobustIPW}'s analysis shows is the key condition for the oracle estimators to remain asymptotically normal. Under somewhat weak overlap, the thresholding bias is of order $b_n^{\gamma_0-1}$ but the oracle standard deviation is of order $n^{-1/2}$, so that thresholds $b_n \ll n^{-1 / (2 (\gamma_0 - 1))}$ ($\gamma_0 > 3$ under the equal-rate rule of thumb) may permit some form of t-test with even inconsistent outcome regression estimates, though as in IPW with strict overlap, confidence intervals must still be adjusted to account for propensity estimation. In the parametric application to right heart catheterization, I bootstrap with reestimated nuisance functions to capture propensity estimation, although this paper does not formally establish conditions for doubly robust bootstrap validity.

Before proceeding to apply clipped AIPW, I derive some rules of thumb for choosing a threshold and describe a critical value adjustment that improves finite-sample coverage in practice.

\subsection{Choice of Threshold and Adjustment of Critical Values}\label{subsec:ChoiceOfThreshold}

Applying these ideas in practice requires a specific threshold $b_n$. I propose a simple rule: set the threshold where an estimated bias bound crosses the leading-order term of the standard error. This rule of thumb does not require directly estimating the degree of weak overlap, uses only the treatment indicators and estimated propensity scores, has a unique and easy-to-compute crossing point,  and implicitly balances the bias from thresholding and from weight estimation in a way which recovers the standard $n^{-1/4}$ sufficient rates under somewhat weak overlap. I discuss how a researcher willing to assert more stringent smoothness or consistency rates for one nuisance function can choose a threshold that relaxes the requirement on the other nuisance function.

When the econometrician is unwilling to take a stance on either nuisance function, the proposed rule of thumb is as follows. Define $S_n(b) = Mean\big( D / \max\{ \hat{e}, b \}^2 \big)$ to be the estimated oracle asymptotic variance at threshold $b$ if $\mu(X)$ were constant and $Var(Y \mid X, D) = 1$.\footnote{A rule of thumb that reflects estimated heteroskedasticity is feasible but would require care to interpret the required outcome rate in finite samples and to ensure monotonicity in $b_n$.} Then $b^2 S_n(b)$  and $b \frac{1}{n} \sum \mathbf{1}\{ \hat{e} \leq b \}$ are reasonable bounds for $B_U$ and $B_C$, respectively. Set 
\begin{align}
    \hat{b}_n = \sup\left\{ b > 0 : n^{-1/2} \sqrt{S_n(b)} \geq b^2 S_n(b)  + b \ \frac{1}{n} \sum_{i=1}^{n} \mathbf{1}\{ \hat{e}_i \leq b \} \right\}. \label{eq:RuleOfThumb}
\end{align}
The right-hand side is the bias bound that one would obtain in \Cref{assumption:ConsistencyRatesSufficient}\ref{assum:ControlOfClippedBias},\ref{assum:ControlOfUnclippedBias} by taking $p_{\mu} = p_{e} = \infty$ (so that $\zeta = 1$), replacing each nuisance sup-norm rate with the threshold $b$, and replacing $\rho_n$ with $P( e(X) \leq b )$. The left-hand side is decreasing in $b$, and the right-hand side $\frac{1}{n} \sum D \min\{ b^2 / \hat{e}^2, 1 \} + \frac{1}{n} \sum b \mathbf{1}\{ \hat{e} \leq b \}$ is increasing in $b$, so the set in \Cref{eq:RuleOfThumb} is an interval and $\hat{b}_n$ can be found quickly. The following result shows that $\hat{b}_n$ is well-defined.
\begin{lemma}[Well-defined rule of thumb]\label{lemma:WellDefinedRuleOfThumb}
    Suppose $\hat{e} \in (0, 1]$ and $\sum D / \hat{e} > 0$. Then the set in \Cref{eq:RuleOfThumb} is a bounded and non-empty interval, so $\hat{b}_n \in (0, \infty)$. 
\end{lemma}
A rule of thumb with nonspecified rates seems particularly attractive, since it is an entirely data-driven way to choose the AIPW threshold. However, in practice, a given outcome regression rate is more difficult to achieve than a given propensity rate under weak overlap, so that rules of thumb with specified nuisance rates are more appropriate for nonparametric outcome regression estimates.

The rule in \Cref{eq:RuleOfThumb} has some nice properties. Under very weak overlap (and Assumptions \ref{def:AllowedDistributions} and \ref{assum:TailLowerBound} with sufficiently accurate propensity estimates), both terms on the right-hand side have the same order of growth, so that neither source of bias dominates, and the crossing $\hat{b}_n$ occurs on the order of $n^{-1/(2 + \gamma_0)}$. Under somewhat weak overlap, the product bias term dominates, yielding a threshold $\hat{b}_n$ and required nuisance consistency rates on the order of $n^{-1/4}$. The rule therefore will be valid under the standard sufficient condition that each nuisance converges at an $n^{-1/4}$ rate, albeit in sup-norm and without the usual ability to allow one nuisance rate to be slower if the other nuisance rate is faster.  More generally, this rule adapts to the unknown $\gamma_0$ through the empirical distribution of estimated propensity scores.

A researcher with a stronger rate guarantee for one nuisance function can choose a more targeted threshold that relaxes the requirement on the other nuisance. Under Hölder continuity of order $\beta_e$ and $\beta_{\mu}$, \cite{Stone1982} implies a feasible propensity rate of $\bar{r}_{e,n} = (n / \log(n))^{-\beta_e / (2 \beta_e + d)}$ and \Cref{prop:GlobalRateAchievable} implies a feasible outcome regression rate of $\bar{r}_{\mu,n} = n^{-\beta_\mu / (2 \beta_\mu + d \gamma_0 / (\gamma_0-1))} \operatorname{polylog}(n)$, where $\gamma_0$ can be estimated from the distribution of $\hat{e}$ or more formally as in \cite{KARLSSON2024102912}. In this case, \Cref{eq:RuleOfThumb} becomes 
\begin{align}
    \hat{b}_n & = \sup\big\{ b > 0 : n^{-1/2} \sqrt{S_n(b)} \geq \underbrace{\bar{r}_{\mu,n} \bar{r}_{e,n} S_n(b) }_{\text{Product bias}} + \underbrace{\bar{r}_{\mu,n} \ \frac{1}{n} \sum_{i=1}^{n} \mathbf{1}\{ \hat{e}_i \leq b \}}_{\text{Thresholding bias}} \big\}, \label{eq:RuleOfThumbProper}
\end{align}
where $\bar{r}_{\mu,n}$ and $\bar{r}_{e,n}$ are either the specified rates or $b$. Once rates are supplied, the right-hand side of \Cref{eq:RuleOfThumbProper} may not be monotone in $b$. I continue to take the sup, which favors efficiency.   If no $b$ in this set exists at large samples, then this paper does not support inference on the ATE under the conjectured rates at any threshold.  A simpler alternative exists if the econometrician is willing to specify $\bar{r}_{e,n}$ but not $\bar{r}_{\mu,n}$: taking $\hat{b}_n^{1+1/p_e} V_n^{1/p_e} = \bar{r}_{e,n} \operatorname{polylog}(n)$ to target \Cref{assumption:ConsistencyRatesSufficient}\ref{assum:ControlOfSquaredDiff} will minimize the burden on the outcome regression estimator.\footnote{An econometrician with $p$-norm guarantees instead of sup-norm guarantees will have to further modify the strategy, for example taking $\hat{b}_n = \bar{r}_{e,n}^{p_e / (p_e + \hat{\gamma}_0 - 1)} \operatorname{polylog}(n)$ for $L^{p_e}$-norm guarantees with an estimated $\hat{\gamma}_0$.  A researcher doing so for finite $p_e$ may wish to check robustness to varying $\hat{\gamma}_0$.} Algorithm \ref{alg:errorUpperBound} formalizes this strategy.

\begin{algorithm}[!ht]
  \DontPrintSemicolon
\caption{Rule-of-thumb threshold: solve \Cref{eq:RuleOfThumbProper} with any provided information on nuisance rates.}
\label{alg:errorUpperBound}
\KwIn{Sup-norm rate upper bounds $\bar{r}_{\mu,n}$ and $\bar{r}_{e,n}$ (maybe null), data $\{D_i, \hat{e}_i\}_{i=1}^{n}$}\;  
\KwOut{Rule-of-thumb threshold $\hat{b}_n$}\;
% rule of thumb threshold 
\SetKwFunction{FEB}{calculateRuleofThumb}
\SetKwProg{Fn}{Function}{:}{}
\Fn{\FEB{$\bar{r}_{\mu,n}, \bar{r}_{e,n}, \{D_i, \hat{e}_i\}_{i=1}^{n}$}}{
% rule of thumb threshold 
\uIf{is.null($\bar{r}_{\mu,n}$) \textup{and} !is.null($\bar{r}_{e,n}$)}{
    $\hat{b}_n \leftarrow \bar{r}_{e,n} \operatorname{polylog}(n)$\;
}
\Else{
    $A_\mu \leftarrow ! is.null(\bar{r}_{\mu,n})$\;
    $A_e \leftarrow !  is.null(\bar{r}_{e,n})$\;
    $\hat{b}_n \leftarrow \sup b$ : $0 \geq $\texttt{errorBoundDiff}($b, \bar{r}_{\mu,n}, \bar{r}_{e,n}, \{D_i, \hat{e}_i\}_{i=1}^{n}, A_{\mu}, A_{e}$)\; 
  }
  \KwRet $\hat{b}_n$
}
\SetKwFunction{FEB}{errorBoundDiff}
\SetKwProg{Fn}{Function}{:}{}
\Fn{\FEB{$b, \bar{r}_{\mu,n}, \bar{r}_{e,n}, \{D_i, \hat{e}_i\}_{i=1}^{n}, A_{\mu}, A_{e}$}}{
    $\tilde{r}_{\mu,n} \leftarrow \left( A_{\mu} \bar{r}_{\mu,n} + (1 - A_{\mu}) b \right)$\;
    $\tilde{r}_{e,n} \leftarrow  \left( A_{e} \bar{r}_{e,n} + (1 - A_{e}) b \right)$\; 
    $\texttt{second\_moment} \leftarrow \frac{1}{n} \sum_{i=1}^{n} \frac{D_i}{\max\{\hat{e}_i, b\}^2}$\;
    $\texttt{error\_diff} \leftarrow \tilde{r}_{\mu,n} \tilde{r}_{e,n} \texttt{second\_moment} + \tilde{r}_{\mu,n} \frac{1}{n} \sum \mathbf{1}\{ \hat{e}_i \leq b \}  - n^{-1/2} \sqrt{\texttt{second\_moment}}$\;
    \KwRet $\texttt{error\_diff}$\;
    
}
\end{algorithm}

\begin{proposition}[Coverage with rule-of-thumb threshold] \label{prop:CoverWithEstimatedThresh}
    Suppose Assumptions \ref{def:AllowedDistributions}, \ref{assum:TailLowerBound}, \ref{assum:HolderSmoothnessAssumptions}, and \ref{assum:NonTrivialConcentration} hold for some finite $\gamma_0 > 1$ and with $\beta_{\mu} > d \gamma_0 / \big( (\gamma_0 - 1) \min\{ \gamma_0, 2 \} \big)$ and $\beta_{e} > d / \min\{ \gamma_0, 2 \}$. Let $\hat{\mu}, \hat{e}$ be estimated through the procedure of \Cref{cor:MinimalSmoothnessConditions},  but let $\hat{b}_n$ be chosen adaptively through Algorithm \ref{alg:errorUpperBound} with equal-rate default. Then the associated Wald confidence interval $\hat{\mathcal{C}}_n(\alpha)$ constructed using $\hat{\psi}_{clip}^{AIPW}(\hat{b}_n)$ and $\hat{\sigma}_n(\hat{b}_n)$ satisfies $$\limsup_{n \to \infty} \sup_{P \in \mathscr{P}} \left| P( \psi(P) \in \hat{\mathcal{C}}_n(\alpha)) - (1-\alpha) \right| = 0.$$ 
\end{proposition}
The Hölder continuity restrictions ensure that both nuisance functions have a better rate than $n^{-1 /(2 + \min\{ \gamma_0, 2 \} )}$.

In small samples, clipped AIPW t-statistics can undercover relative to the $1.96$ critical value. I therefore follow \cite{RotheLimitedOverlap} and apply a critical value adjustment based on a Behrens-Fisher problem using the estimated variance and kurtosis. I find that using the \cite{SatterthwaiteAdjustment} effective degrees of freedom works well. The estimated degrees of freedom I implement is $\hat{\nu} = \max\big\{ 1, 2 n \big/ \big( \hat{\kappa} - 1 + 2 /(n-1) \big) \big\}$, where the estimated kurtosis of the scores is $\hat{\kappa} = n \sum_{i=1}^{n} ( \phi(Z_i \mid \hat{b}_n, \hat{\eta}_i) - \hat{\psi}_{clip}^{AIPW}(\hat{b}_n) )^4 \big/ (n \hat{\sigma}_n)^4$.

\section{Simulation Evidence}\label{sec:Sims}

My simulation design is based on the design in \cite{XinweiMaRobustIPW}. I find that simple plug-in nonparametric estimators are able to remove bias and achieve near-nominal coverage in large samples at only a slight precision cost.

As in \citeauthor{XinweiMaRobustIPW}'s simpler specification, I simulate data with $P(e(X) \leq \pi) = \pi^{\gamma_0 - 1}$ and draw treated outcomes as $\mu(X) + \varepsilon$, where $\mu(X) = \kappa (1-e(X)) $ and $\varepsilon \mid X, D \sim \chi_4^2$ is centered and scaled to achieve zero mean and unit variance. I achieve this propensity distribution by taking $X \sim Unif([0, 1])$ i.i.d. and setting $e(X) = X^{1/(\gamma_0-1)}$. However, I take $\gamma_0 = 1.8$ (rather than $1.5$ or $1.9$) to ensure feasible regression rates without getting too close to $\gamma_0 = 2$, set $\kappa = 3.37$ rather than $\kappa = 1$ so that $Var( \mu(X) ) = \Var( Y \mid X, D = 1)$, and study t-statistics with rule-of-thumb critical values rather than subsampling-based confidence intervals. I also recenter $\mu(X)$ so that $E[ \mu(X) ] = 0$.  I present results for 5,000 simulations of increasingly large samples.

I use nuisance function estimates and adjustments that are not particularly adapted to the setting. I estimate $\hat{\mu}(x)$ and $\hat{e}(x)$ using shrinkage cubic splines and REML estimation, as implemented by the \texttt{mgcv} package in \texttt{R}. Both nuisance functions have continuous first derivatives, which the estimators capture, but both nuisances have unbounded second derivatives at the origin, so that the estimators may be poorly behaved precisely where thresholding binds.  I use $K=5$-fold cross-fitting. In this setting, \Cref{prop:GlobalRateAchievable} establishes that kernel regression can achieve a pointwise rate of $n^{-1 / (3 + 1 / (\gamma_0-1))}$. I conjecture that $r_{\mu,n} \ll n^{-1/5}$, which is feasible if $\gamma_0 > 1.5$, and choose the clipping threshold $b_n$ based on Algorithm \ref{alg:errorUpperBound}. I compare clipped AIPW to two estimators that use the same nuisance functions: fixed-trimming AIPW, which targets the subpopulation with $\hat{e}(X) \in [0.1, 0.9]$, and clipped IPW, which thresholds but does not use the AIPW regression adjustment.

\begin{figure}[!ht]
    \centering
    \includegraphics[width=\textwidth]{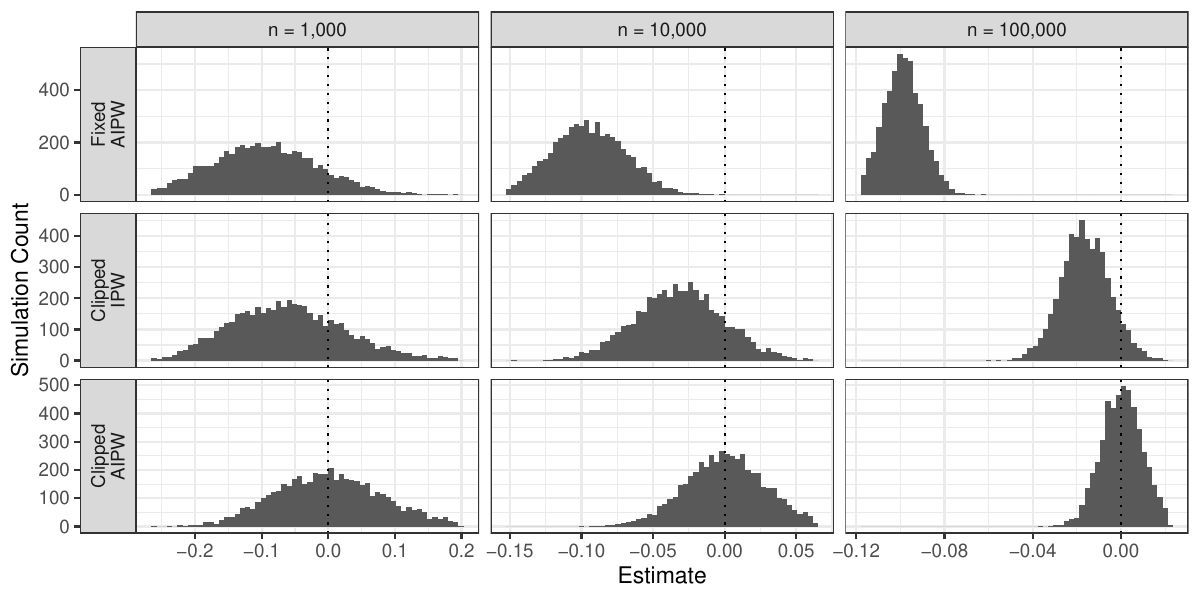}
    \caption{Histograms of point estimates in simulations for the various methods considered in the simulations. Vertical dotted and solid lines indicate true APO and median estimate, respectively.}
    \label{fig:estimates}
\end{figure}

I find that clipped AIPW effectively mitigates bias at minimal precision cost. \Cref{fig:estimates} presents a histogram of point estimates for the three estimators at increasing sample sizes. Fixed-trimming AIPW concentrates at $-0.099$, reflecting the bias introduced by dropping the selected subsample with weak overlap. Clipped IPW is consistent but with first-order bias. Clipped AIPW's bias is only 0.005 with 1,000 observations, and is visually small at all samples. The cost of targeting the full population is a loss in precision; \cite{KhanAndTamer} and \Cref{cor:WorstCaseConsistency} establish that the percent increase in variance must be infinite in the limit. \Cref{fig:estimates} shows that the lost precision is small in practical terms: 0.0039 (4.8\%) with 1,000 observations, 0.0025 (9.3\%) with 10,000 observations, and 0.0007 (7.2\%) with 100,000 observations. In theory, the percentage standard error increase should grow to infinity in the limit because $\gamma_0 < 2$. However, the divergence of standard errors is slow: at the median $n=100{,}000$ adaptive threshold of $0.007$, Appendix \Cref{lemma:VarianceComparison} shows that oracle clipped AIPW standard errors would only be 38\% larger than under fixed trimming. Further, fixed trimming's inclusion criterion depends on the estimated propensity score and is not orthogonalized, so that feasible fixed trimming's standard deviation is larger than the oracle estimator's standard deviation asymptotically, and the realized precision gap need not widen with the sample size.

\begin{figure}[!t]
    \centering
    \includegraphics[width=\textwidth]{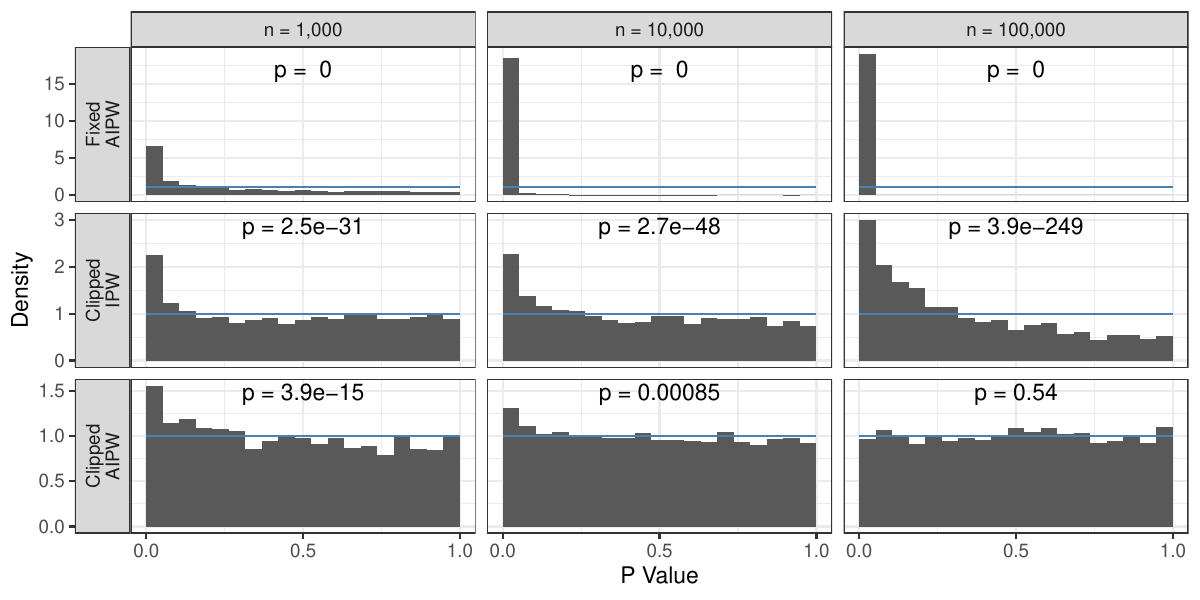}
    \caption{Histograms of simulation p-values on null hypothesis of true APO for various sample sizes. Dotted lines correspond to the target Uniform(0, 1) density. P-values in labels correspond to Kolmogorov-Smirnov tests for the Uniform(0, 1) distribution.}
    \label{fig:pvalues}
\end{figure}

I find that clipped AIPW yields well-calibrated p-values in large samples. \Cref{fig:pvalues} presents a histogram of p-values on the true null hypothesis for the three estimators. Fixed-trimming p-values tend to zero, reflecting that it provides well-calibrated inference on a biased estimand. Clipped IPW overrejects in all samples. Clipped AIPW overrejects in small samples, with 90\% and 95\% confidence intervals covering the null hypothesis in 86.4\% and 92.2\%, respectively, of simulations with 1,000 observations. The undercoverage reflects estimated standard errors that understate sampling variability with heavy-tailed scores: the median $\hat{\sigma}_n$ is 84\% of the simulation standard deviation at $n=1{,}000$, rising to 98\% at $n=100{,}000$ (\Cref{tab:mc_results}). With 10,000 observations, the coverage rates increase to 87.9\% and 93.5\%, respectively. With 100,000 observations, the coverage rates reach 89.9\% and 95.3\%, respectively, and a Kolmogorov-Smirnov test no longer rejects that p-values are precisely calibrated ($p = 0.54$). Conversely, fixed-trimming 90\% confidence interval coverage tends to zero, with a coverage rate of 56.2\% with 1,000 observations, and no observed coverage with 100,000 observations.

\begin{table}[!p]
\centering
\caption{\label{tab:mc_results}Monte Carlo results across specifications and sample sizes.}
\centering
\resizebox{\ifdim\width>\linewidth\linewidth\else\width\fi}{!}{
\fontsize{9}{11}\selectfont
\begin{threeparttable}
\begin{tabular}[t]{lrrrrrrrrrr}
\toprule
\multicolumn{4}{c}{ } & \multicolumn{4}{c}{Median} & \multicolumn{2}{c}{Coverage} & \multicolumn{1}{c}{ } \\
\cmidrule(l{3pt}r{3pt}){5-8} \cmidrule(l{3pt}r{3pt}){9-10}
Method & $n$ & Bias & SD & SE/SD & $b_n$ & Width & $t_{Crit}$ & 90\% & 95\% & KS $p$\\
\midrule
\addlinespace[0.3em]
\multicolumn{11}{l}{\textbf{Panel A: $\gamma_0 = 1.8$}}\\
\hspace{1em} & 1,000 & -0.097 & 0.082 & 0.82 & 0.100 & 0.267 & 1.991 & 56.2 & 66.0 & 0\\

\hspace{1em} & 10,000 & -0.096 & 0.027 & 0.81 & 0.100 & 0.085 & 1.965 & 1.2 & 2.4 & 0\\

\hspace{1em}\multirow[t]{-3}{*}{\raggedright\arraybackslash Fixed AIPW} & 100,000 & -0.099 & 0.009 & 0.76 & 0.100 & 0.027 & 1.961 & 0.0 & 0.0 & 0\\
\addlinespace[0.3em]
\hspace{1em} & 1,000 & -0.060 & 0.087 & 1.09 & 0.034 & 0.404 & 2.111 & 82.1 & 88.4 & $2\times 10^{-31}$\\

\hspace{1em} & 10,000 & -0.031 & 0.031 & 1.22 & 0.015 & 0.157 & 2.031 & 81.3 & 88.5 & $3\times 10^{-48}$\\

\hspace{1em}\multirow[t]{-3}{*}{\raggedright\arraybackslash Clipped IPW} & 100,000 & -0.016 & 0.011 & 1.31 & 0.007 & 0.058 & 1.991 & 74.5 & 84.9 & $4\times 10^{-249}$\\
\addlinespace[0.3em]
\hspace{1em} & 1,000 & 0.005 & 0.085 & 0.84 & 0.034 & 0.291 & 2.025 & 86.4 & 92.2 & $4\times 10^{-15}$\\

\hspace{1em} & 10,000 & 0.002 & 0.029 & 0.90 & 0.015 & 0.105 & 1.994 & 87.9 & 93.5 & 0.001\\

\hspace{1em}\multirow[t]{-3}{*}{\raggedright\arraybackslash Clipped AIPW} & 100,000 & 0.001 & 0.010 & 0.98 & 0.007 & 0.037 & 1.978 & 89.9 & 95.3 & 0.539\\
\addlinespace[0.3em]
\hspace{1em} & 1,000 & -0.007 & 0.098 &  & 0.015 & 0.462 &  & 88.3 & 93.3 & 0.032\\

\hspace{1em} & 10,000 & 0.007 & 0.041 &  & 0.003 & 0.182 &  & 86.2 & 92.0 & $4\times 10^{-10}$\\

\hspace{1em}\multirow[t]{-3}{*}{\raggedright\arraybackslash Ma and Wang IPW} & 100,000 & 0.013 & 0.019 &  & 0.001 & 0.077 &  & 72.3 & 80.8 & $3\times 10^{-167}$\\
\addlinespace[0.3em]
\hspace{1em} & 1,000 & 0.008 & 0.083 & 0.80 & 0.034 & 0.264 & 2.002 & 83.8 & 90.6 & $2\times 10^{-34}$\\

\hspace{1em} & 10,000 & 0.003 & 0.028 & 0.85 & 0.015 & 0.095 & 1.981 & 85.2 & 91.7 & $3\times 10^{-18}$\\

\hspace{1em}\multirow[t]{-3}{*}{\raggedright\arraybackslash Trimmed AIPW} & 100,000 & 0.002 & 0.009 & 0.95 & 0.007 & 0.034 & 1.970 & 88.0 & 93.6 & 0.004\\
\addlinespace[0.3em]
\hspace{1em} & 1,000 & 0.005 & 0.085 & 0.84 & 0.034 & 0.280 & 1.960 & 85.1 & 91.2 & $2\times 10^{-21}$\\

\hspace{1em} & 10,000 & 0.002 & 0.029 & 0.90 & 0.015 & 0.103 & 1.960 & 87.1 & 92.8 & $8\times 10^{-6}$\\

\hspace{1em}\multirow[t]{-3}{*}{\raggedright\arraybackslash Normal Approx.} & 100,000 & 0.001 & 0.010 & 0.98 & 0.007 & 0.037 & 1.960 & 89.6 & 94.8 & 0.833\\
\addlinespace[0.3em]
\multicolumn{11}{l}{\textbf{Panel B: $\gamma_0 = 1.4$}}\\
\hspace{1em} & 1,000 & -0.608 & 0.098 & 0.83 & 0.100 & 0.329 & 2.004 & 0.0 & 0.0 & 0\\

\hspace{1em} & 10,000 & -0.598 & 0.033 & 0.79 & 0.100 & 0.105 & 1.967 & 0.0 & 0.0 & 0\\

\hspace{1em}\multirow[t]{-3}{*}{\raggedright\arraybackslash Fixed AIPW} & 100,000 & -0.595 & 0.011 & 0.75 & 0.100 & 0.033 & 1.961 & 0.0 & 0.0 & 0\\
\addlinespace[0.3em]
\hspace{1em} & 1,000 & -0.235 & 0.099 & 1.00 & 0.032 & 0.435 & 2.183 & 35.3 & 45.6 & 0\\

\hspace{1em} & 10,000 & -0.158 & 0.047 & 1.12 & 0.010 & 0.219 & 2.089 & 15.2 & 24.4 & 0\\

\hspace{1em}\multirow[t]{-3}{*}{\raggedright\arraybackslash Clipped IPW} & 100,000 & -0.103 & 0.021 & 1.23 & 0.003 & 0.105 & 2.023 & 1.4 & 3.9 & 0\\
\addlinespace[0.3em]
\hspace{1em} & 1,000 & 0.021 & 0.155 & 0.55 & 0.032 & 0.356 & 2.075 & 69.5 & 78.8 & $2\times 10^{-238}$\\

\hspace{1em} & 10,000 & 0.010 & 0.060 & 0.63 & 0.010 & 0.154 & 2.033 & 73.0 & 81.9 & $1\times 10^{-179}$\\

\hspace{1em}\multirow[t]{-3}{*}{\raggedright\arraybackslash Clipped AIPW} & 100,000 & 0.007 & 0.023 & 0.74 & 0.003 & 0.068 & 2.006 & 78.1 & 86.7 & $4\times 10^{-105}$\\
\addlinespace[0.3em]
\hspace{1em} & 1,000 & -0.075 & 0.171 &  & 0.006 & 0.817 &  & 83.3 & 88.4 & $1\times 10^{-20}$\\

\hspace{1em} & 10,000 & -0.034 & 0.081 &  & 0.001 & 0.442 &  & 86.6 & 92.2 & $2\times 10^{-7}$\\

\hspace{1em}\multirow[t]{-3}{*}{\raggedright\arraybackslash Ma and Wang IPW} & 100,000 & -0.012 & 0.047 &  & $2\times 10^{-4}$ & 0.207 &  & 80.7 & 87.8 & $4\times 10^{-54}$\\
\addlinespace[0.3em]
\multicolumn{11}{l}{\textbf{Panel C: $\gamma_0 = 2.4$}}\\
\hspace{1em} & 1,000 & 0.209 & 0.079 & 0.76 & 0.100 & 0.239 & 1.982 & 6.9 & 11.7 & 0\\

\hspace{1em} & 10,000 & 0.195 & 0.025 & 0.77 & 0.100 & 0.076 & 1.964 & 0.0 & 0.0 & 0\\

\hspace{1em}\multirow[t]{-3}{*}{\raggedright\arraybackslash Fixed AIPW} & 100,000 & 0.190 & 0.009 & 0.71 & 0.100 & 0.024 & 1.960 & 0.0 & 0.0 & 0\\
\addlinespace[0.3em]
\hspace{1em} & 1,000 & -0.015 & 0.061 & 1.17 & 0.069 & 0.287 & 2.006 & 93.1 & 97.1 & $9\times 10^{-21}$\\

\hspace{1em} & 10,000 & -0.009 & 0.019 & 1.26 & 0.035 & 0.094 & 1.970 & 93.4 & 96.9 & $5\times 10^{-20}$\\

\hspace{1em}\multirow[t]{-3}{*}{\raggedright\arraybackslash Clipped IPW} & 100,000 & -0.006 & 0.006 & 1.30 & 0.017 & 0.031 & 1.962 & 88.1 & 94.3 & $2\times 10^{-18}$\\
\addlinespace[0.3em]
\hspace{1em} & 1,000 & 0.001 & 0.061 & 0.95 & 0.069 & 0.230 & 1.980 & 89.3 & 94.5 & 0.008\\

\hspace{1em} & 10,000 & 0.000 & 0.019 & 0.99 & 0.035 & 0.074 & 1.964 & 90.1 & 95.1 & 0.881\\

\hspace{1em}\multirow[t]{-3}{*}{\raggedright\arraybackslash Clipped AIPW} & 100,000 & 0.000 & 0.006 & 1.00 & 0.017 & 0.024 & 1.961 & 89.9 & 95.0 & 0.416\\
\addlinespace[0.3em]
\hspace{1em} & 1,000 & -0.011 & 0.062 &  & 0.072 & 0.285 &  & 92.5 & 96.5 & $8\times 10^{-14}$\\

\hspace{1em} & 10,000 & -0.008 & 0.019 &  & 0.045 & 0.091 &  & 93.7 & 97.2 & $3\times 10^{-17}$\\

\hspace{1em}\multirow[t]{-3}{*}{\raggedright\arraybackslash Ma and Wang IPW} & 100,000 & -0.005 & 0.006 &  & 0.030 & 0.029 &  & 87.3 & 93.5 & $2\times 10^{-33}$\\
\bottomrule
\end{tabular}
\begin{tablenotes}[para]
\noindent\textit{Notes:} \textit{Fixed AIPW} trims to $\hat e(X) \in [0.1, 0.9]$. \textit{Clipped AIPW} (proposed) and \textit{Trimmed AIPW} clip or trim weights at $b_n$ with the AIPW adjustment; \textit{Clipped IPW} clips without the AIPW adjustment; \textit{Normal Approx.} is clipped AIPW with a 1.96 critical value. \textit{Ma and Wang IPW} uses their threshold, bias correction, and subsampling intervals, holding first-stage nuisance estimates fixed. \textit{Bias} and \textit{SD} are the mean bias and standard deviation of the point estimate. \textit{Median} columns report the estimated standard error (as a ratio to SD), threshold, nominal 95\% interval width, and Satterthwaite critical value (1.96 for Normal Approx.). \textit{Coverage} columns give empirical coverage (\%). \textit{KS $p$} tests whether $p$-values under the true null are Uniform(0,1).
\end{tablenotes}
\end{threeparttable}}
\end{table}

In \Cref{tab:mc_results}, I find that clipped AIPW's performance compares favorably to \citeauthor{XinweiMaRobustIPW}'s purpose-built procedure at least until overlap is extremely weak, and the potential precision advantage of fixed trimming is limited and can even invert under somewhat weak overlap. I compare to a \cite{XinweiMaRobustIPW}-inspired IPW procedure with a different adaptive threshold, a local-linear-regression-on-$\hat{e}(X)$ debiasing procedure below the threshold, and subsampling inference (but keeping nuisance estimates constant in the interest of speed). 
%choosing a regression bandwidth $h$ to set $\frac{1}{n} \sum \mathbf{1}\{ \hat{e}(X) \leq h \} \approx 1 \big/ \big( n h^{5} \big)$, chooses a trimming threshold $b$ to solve $2 b \sum \mathbf{1}\{ \hat{e}(X) \leq b \} \approx  \hat{E}[Y^2 \mid e(X) = 0, D = 1] / \hat{E}[Y \mid e(X) = 0, D = 1]^2$, returns IPW above $b$ and a local linear regression below $b$, and subsamples (but keeping nuisance estimates constant in the interest of speed). 
The \citeauthor{XinweiMaRobustIPW} procedure returns wider confidence intervals from smaller thresholds, and improved small-sample coverage with some large-sample undercoverage potentially driven by first-stage nuisance estimation. Trimmed AIPW's p-values converge more slowly. A plain 1.96 critical value for clipped AIPW worsens small-sample performance but covers in 94.8\% of simulations at $n=100{,}000$. I capture the role of weak overlap in coverage by changing $e(X) = X^{1/(\gamma_0-1)}$, which moves the degree of overlap and of Hölder smoothness in opposite directions; I try to minimize any accidental benefit of weak overlap by keeping the $\mu(x)$ function constant. 
Reducing overlap to $\gamma_0 = 1.4$ leads to more severe undercoverage, with clipped AIPW standard errors understating sampling variability by 26--45\%. The \citeauthor{XinweiMaRobustIPW} approach covers best, though the gap narrows at $n=100{,}000$ despite clipped AIPW's much smaller confidence intervals. Going from the very weak overlap regime to the somewhat weak overlap $\gamma_0 = 2.4$ leads all thresholded estimators to perform better, though the IPW-based estimators can actually overcover at $n=1{,}000$ and $n=10{,}000$, and there is some hint that the lack of orthogonalization under IPW leads to some undercoverage at $n=100{,}000$. At $\gamma_0 = 2.4$, clipped AIPW's 95\%  intervals cover in 94.5--95.1\% of simulations at every sample size, and clipped AIPW's point estimates are always more precise than those of fixed-trimming AIPW. This is consistent with the calculation in Appendix \Cref{lemma:VarianceComparison}: under somewhat weak overlap, $b_n$ must be quite small before clipped AIPW's extreme weights outweigh the lost subpopulation. If overlap is not too weak --- if $\gamma_0 > 2.63$ in the lemma's tweaked specification that holds $Var(\mu(X))$ instead of outcome smoothness constant --- oracle fixed trimming has worse precision than even the unthresholded AIPW estimator.

\section{Application to Right Heart Catheterization}\label{sec:EmpiricalApplication}

I apply the clipped AIPW estimator to study the effect of right heart catheterization (RHC) on 30-day mortality. I find that a fixed-trimming rule discards 18\% of the sample, while a data-adaptive clipped AIPW estimator that restores the full sample only increases the estimated standard error from 1.7 to 1.8 percentage points.

This dataset was first analyzed by \cite{ConnorsEtAl1996}, and is a common benchmark in the weak overlap literature \citep{CrumpEtAlOptimizePrecision, ArmstrongKolesarBandwidthSnooping}. I analyze a version of the dataset from \cite{ArmstrongKolesarBandwidthSnooping}. The dataset is composed of 5{,}735 adult patients, and the treatment $D$ corresponds to receiving RHC within 24 hours of admission. The target is the average treatment effect of RHC on 30-day mortality. The data includes 52 covariates $X$ (72 covariates if counting factor levels separately).

This setting is known to exhibit weak overlap driven by low-propensity patients. I do not take a stance on how weak: the procedure I apply does not require knowing $\gamma_0$, and extreme propensity ranks exhibit nonlinearity, suggesting there is not a precise tail index. For example, applying  \cite{KARLSSON2024102912}'s estimator to this data with ranks 8--50 yields estimated tail parameters of $\gamma_0 = 3.466$ (SE = 0.065) below and $\gamma_0 = 2.874$ (SE = 0.047) above, which would put both tails in the somewhat weak overlap regime. However, the estimated tail indices are sensitive to the choice of rank window, and suggest weaker overlap in the tail that has apparently better overlap properties.

I compare the clipped AIPW estimator to \citeauthor{CrumpEtAlOptimizePrecision}'s rule of thumb that targets the  subpopulation with propensity scores between 10\% and 90\% in order to approximately minimize asymptotic variance. I estimate the nuisance functions $e(X)$ and $\mu(X)$ using five-fold cross-fitting. I estimate nuisance functions with logistic regression to align with \cite{CrumpEtAlOptimizePrecision}'s empirical application. I estimate standard errors from 1,000 bootstrap draws of the full procedure, reestimating nuisance functions and thresholds within each draw but keeping fold assignments fixed. The bootstrap is a practical implementation choice, since the theory above does not cover a misspecified nuisance model (\Cref{subsec:AnalysisOfParametricCase}).

I compare AIPW in the \citeauthor{CrumpEtAlOptimizePrecision} subpopulation to full-population AIPW with clipping at a threshold $b_n$. I choose $b_n$ through the no-specified-upper-bound version of Algorithm \ref{alg:errorUpperBound} because both nuisance functions are estimated parametrically. I estimate separate thresholds for the treated and controls based on targeting the average potential outcomes separately, yielding $b_n^{(1)} = 0.068$ and $b_n^{(0)} = 0.09$ (clipping $\hat{e}$ to $[0.068, 0.91]$), which clip 10.5\% and 1.4\% of the population, respectively. The larger $b_n^{(0)}$ reflects the scarcity of propensity scores near one, so weaker nuisance rates suffice for t-statistic validity in the upper tail. 

In contrast, the fixed 10\% trimming rule discards 16.3\% of the population with estimated propensity scores below 0.1 and 1.8\% of the population with estimated propensity scores above 0.9.

I find that targeting the full sample comes at minimal precision cost, but also has only a small effect on the point estimate. The estimated effect in the rule-of-thumb subpopulation is a 5.8 percentage point increase in 30-day mortality (bootstrap SE = 1.7). With clipped AIPW, the estimated effect in the full population is a 6.1 percentage point mortality increase (SE = 1.8), and the \citeauthor{SatterthwaiteAdjustment} estimated degrees of freedom, 451.7, only increases the t-statistic critical value to 1.965. The difference in point estimates of 0.3 is small but imprecise (SE=1.3), and the implied effect among trimmed observations of 7.4 percentage points is consistent with treatment selection but even more imprecise (SE = 6.4). The 10\% fixed trimming rule's implicit assumption therefore survives the robustness check: restoring the discarded 18\% of patients moves the estimate by five percent, and only costs 0.1 percentage points of standard error.

%Taken together, these results illustrate that under weak overlap, targeting the causal effect within the full population need not come at a large precision cost. In this application, clipped AIPW with a rule-of-thumb clipping rate yields similar estimates to estimators that target a fixed trimmed sample, while targeting a population that is often more relevant and adding only a small precision cost. 

\section{Conclusion}\label{sec:Conclusion}

This work shows that standard Wald confidence intervals for clipped AIPW can achieve nominal coverage for standard causal effects under plausible conditions. I provide sufficient conditions on nuisance regression rates for clipped (or trimmed) AIPW to be uniformly valid over distributions with even very weak overlap. I use these theoretical results to derive new rules of thumb for choosing a threshold. I find that Wald confidence intervals perform well in simulations, especially in large samples, and can achieve comparable precision to a fixed 10\% trimming rule in practice.

These results can be extended in many interesting directions. This work exploits Neyman orthogonality to achieve standard statistical inference in the presence of a small region of irregular identification.  \cite{SasakiUra2022} and \cite{ma2023doubly} propose estimators for ratio estimands beyond IPW; the arguments here are likely to extend to their more general framework. Issues of weak overlap also arise for more general importance sampling estimators in settings like difference-in-differences \citep{CallawaySantAnna}, choice-based sampling with extreme weights \citep{ChoiceBasedSampling}, or identification at infinity \citep{AndrewsSchafgans, KhanNekipelov}; the results and rules of thumb here can likely be adapted to those settings. \cite{semenova2024aggregatedintersectionboundsaggregated} applies thresholding strategies to intersection bounds, where at a high level a margin condition plays the role of the minimal overlap bound here. Further, this paper leaves open whether the standard product is sufficient when $\gamma_0 > 1 + \zeta$, the minimal product rate when $\gamma_0 < 1 + \zeta$, regression estimators that are robust to degenerate local polynomial Gram matrices, and heteroskedasticity-aware rules of thumb. 
%The regression analysis in \Cref{sec:RegressionRates} may also extend to estimating the effects of continuous treatments, and the partitioning argument may be useful for characterizing the behavior of nonparametric regression estimators on nonconvex sets. 

The results here suggest that thresholded AIPW is a viable alternative to fixed-trimming rules. Nothing in the paper should be read to argue against safer estimands like overlap-weighted estimands when those estimands are truly of interest. However, when researchers already apply a 10\% fixed-trimming rule of thumb as a simple way to ensure t-statistics are valid, thresholded AIPW at a minimum provides a simple way to confirm that the results are not driven by the data that is being thrown out.

\bibliography{bibliography}

\clearpage

\appendix

\section{Proofs}\label{app:Proofs}

The proofs are split into initial sections showing asymptotic properties of oracle clipped AIPW (\Cref{proofs:OracleNormality}), showing bias and variance guarantees to yield consistency (\Cref{subsec:BlackBoxBiasVar}), and degradation of black-box consistency rates (\Cref{proofs:ConsistencyRate}). Note that \Cref{proofs:ConsistencyRate} is out of order from the perspective of the main text and corresponds to claims in \Cref{subsec:SlowerConsistencyRate}, but \Cref{prop:ConsistencyRates} is used to show claims that appear earlier in the main text.  \Cref{proofs:AsymptoticNormality} presents proofs for the main asymptotic normality claims. Finally, \Cref{proofs:RateThings}, \Cref{proofs:Regression}, and \Cref{proofs:RulesOfThumb} prove claims about black-box nuisance rates, outcome regression rates, and rules of thumb and the simulation, respectively.

\textbf{Additional Notation}. In these proofs, I use $P_n$ to refer to an arbitrary sequence of distributions for the purposes of computing suprema; for such sequences, I use $\psi_n = \psi(P_n)$ to denote the sequence of average potential outcomes. I use $\operatorname{Leb}( A )$ to refer to the Lebesgue measure. I use $\mathbb{P}_n[ c_n ]$ to refer to the average of $c_n$ over $n$ draws from $P$ or $P_n$ (sometimes abusing notation and including nuisance functions in $c_n$), and I use $\mathbb{P}_{n_k}$ for the same in fold $k$. I write $\lim_{x \to z^+} f(x)$ and $\lim_{x \to z^-} f(x)$ for the right- and left-hand limits of $f$ at $z$. I write that $c_n \overset{P_n}{\rightsquigarrow} N(0, 1)$ if $\sup_{t \in \R} \left| P_n( c_n \leq t) - \Phi(t) \right| \to 0$, where $\Phi$ is the standard normal cumulative distribution function. I use local polynomial regression to refer to the estimator $\hat{\mu}_n(x) = U(0)^T \hat{\beta}_n(x)$, where $\hat{\beta}_n(x)$ solves the local regression problem $\argmin_{\beta \in \R^{\binom{d + \ell_{\mu}}{d}}} \sum_{i=1}^n D_i \big( Y_i - \beta^T U( (X_i - x) / h ) \big)^2 K( (X_i - x) / h )$, where $U(z)$ is the vector of all monomials in $z$ up to degree $\ell_{\mu}$, for simplicity $K(z) = \mathbf{1}\{ \| z \| \leq 1 \}$, and where when $\hat{\beta}_n(x)$ is not uniquely identified then $\hat{\mu}_n(x) = 0$.

\subsection{Oracle Normality}\label{proofs:OracleNormality}

\begin{lemma}[Some useful inequalities] \label{lemma:useful_inequalities}
    There exists a $\pi_{\min} = \pi_{\min}(C, \gamma_0) \in (0, 1)$ such that $P(e(X) \geq 2 \pi_{\min}) \geq 1/2$, $\inf_{P \in \mathscr{P}} P(D = 1) \geq \pi_{\min} > 0$, and $\inf_{P \in \mathscr{P}} P(e(X) > \pi_{\min}/2) \geq \pi_{\min}/2$.
    
    Further, for all $b_n \to 0$ and all $n$ large enough, the following inequalities hold throughout $\mathscr{P}$:
    \begin{enumerate}[label=(\roman*), itemsep=-0.5ex, topsep=-0.5ex]
        \item $\E[e(X)/\{ e(X) \vee b_n \}^2] \geq \pi_{\min}/2$ \label{item:propensity_ratio_lower_bound}
        \item $\E[|\phi_n - \E_{P_n}[\phi_n]|^q] \leq (4 M)^q \E[ e(X) / \{ e(X) \vee b_n \}^2] / b_n^{q - 2}$ \label{item:centered_qth_moment_bound}
        \item $\E[| \phi_n|^q] \leq (8M)^q \E[ e(X) / \{ e(X) \vee b_n \}^2]/b_n^{q - 2}$ \label{item:qth_moment_bound}
    \end{enumerate}
\end{lemma}

\begin{proof}[Proof of \Cref{lemma:useful_inequalities}]
For any $P \in \mathscr{P}$, $P( e(X) \leq 2 \pi_{\min} ) \leq C \left( (2 C)^{\frac{-1}{\gamma_0-1}} \right)^{\gamma_0-1} = 1/2$, so that $P\left( D = 1 \right) \geq P\left( e(X) \geq 2 \pi_{\min}, D = 1 \right) \geq (1/2) (2 \pi_{\min}) = \pi_{\min}$. Without loss of generality let $\pi_{\min} < 1$. 

For the third property of $\pi_{\min}$, 
\begin{align*}
    \pi_{\min} &\leq \E[ e(X)]\\
    &= \E[ e(X) \mathbf{1} \{ e(X) \leq \pi_{\min}/2 \}] + \E[ e(X) \mathbf{1} \{ e(X) > \pi_{\min}/2 \}]\\
    &\leq (\pi_{\min}/2) [1 - P(e(X) > \pi_{\min}/2)] + P(e(X) > \pi_{\min}/2)\\
    &< \pi_{\min}/2 + P(e(X) > \pi_{\min}/2)
\end{align*}
Subtracting $\pi_{\min}/2$ from the far left- and right-hand sides of this inequality gives the desired conclusion.

I take the $b_n \to 0$ proofs one at a time. 
\begin{enumerate}[label=(\roman*), itemsep=-0.5ex]
    \item If $b_n \leq \pi_{\min}/2$ (which happens for all large $n$), then:
    \begin{align*}
    \E\left[ \frac{e(X)}{\{ e(X) \vee b_n \}^2} \right] &\geq \E\left[ \frac{ \mathbf{1} \{ e(X) \geq b_n \}}{e(X)} \right] \geq P(e(X) \geq b_n) \geq P\left( e(X) \geq \frac{\pi_{\min}}{2} \right) \geq \frac{\pi_{\min}}{2}.
    \end{align*}
    
    \item By Jensen's inequality:
    \begin{align*}
        \E[ | \phi_n - \E_{P_n}[\phi_n]|^q] &\leq 2^{q - 1} (\E[|\mu(X) - \E_{P_n}[\mu(X)]|^q] + \E[|Y - \mu(X)|^q D / \{ e(X) \vee b_n \}^q])\\
        &\leq 2^{q - 1}( 2^q \E[| \mu(X)|^q] + \E[\E[| Y - \mu(X)|^q \mid X, D = 1] e(X) / \{ e(X) \vee b_n \}^q])\\
        &\leq 2^{q - 1} (2^q M^q + 2^q \E[ \E[|Y|^q \mid X, D = 1] e(X) / \{ e(X) \vee b_n \}^q])\\
        &\leq 2^{q - 1}(2^q M^q + 2^q M^q \E[ e(X) / \{ e(X) \vee b_n \}^2 \times 1/\{ e(X) \vee b_n \}^{q - 2}])\\
        &\leq 2^{q - 1} (2^q M^q + 2^q M^q \E[ e(X) / \{ e(X) \vee b_n \}^2] / b_n^{q - 2}). 
    \end{align*}
    Since $\E[ e(X) / \{ e(X) \vee b_n \}^2] / b_n^{q - 2} \geq \pi_{\min}/2 b_n^{q - 2} \to \infty$ by \Cref{item:propensity_ratio_lower_bound}, I may further bound the above quantity by $2^{2q} M^q \E[ e(X) / \{ e(X) \vee b_n \}^2]/b_n^{q - 2}$ once $n$ is large enough.
    
    \item By Jensen's inequality:
    \begin{align*}
        \E[|\phi_n|^q] &= \E[| \phi_n - \E_{P_n}[\phi_n] + \E_{P_n}[\phi_n]|^q]\\
        &\leq 2^{q - 1} (\E[| \phi_n - \E_{P_n}[\phi_n]|^q] + |\E_{P_n}[\phi_n]|^q])\\
        &\leq 2^{q - 1} (4M)^q \E[ e(X)/\{ e(X) \vee b_n \}^2]/b_n^{q - 2} + 2^{q - 1} |\E_{P_n}[\mu(X)]|^q \tag{\Cref{item:centered_qth_moment_bound}}\\
        &\leq 2^{q - 1}( 4M)^q \E[ e(X)/\{ e(X) \vee b_n \}^2]/b_n^{q - 2}] + 2^{q - 1}  \E[\E[|Y|^q \mid X, D = 1]] \tag{Jensen}\\
        &\leq 2^{q - 1} (4M)^q \E[ e(X)/\{ e(X) \vee b_n \}^2]/b_n^{q -2} + 2^{q - 1} M^q. 
    \end{align*}
    As before, since $\E[ e(X) / \{ e(X) \vee b_n \}^2] /b_n^{q - 2} \to \infty$, the first term in the upper bound is eventually larger than the second and I may bound the whole expression by $(8M)^q \E[ e(X) / \{ e(X) \vee b_n \}^2]/b_n^{q - 2}$ once $n$ is large enough. \qedhere 
    \end{enumerate}
\end{proof}

\begin{lemma}[Asymptotic variance order] \label{lemma:PhiLowerBound}
    For any sequence of $P_n  \in \mathscr{P}$ and $b_n \in (0, 1]$, $V_n = \Theta\big( W(P_n, b_n) \big)$ with constants that depend only on $M, \sigma_{\min}, C,$ and $\gamma_0$, and with a uniform lower bound $V_n \geq \sigma_{\min}^2 \pi_{\min} > 0$.
\end{lemma}

\begin{proof}[Proof of \Cref{lemma:PhiLowerBound}] 
    By iterated expectations, $W(P, b) = E_{P}[e / \max\{e, b\}^2 ]$. Because $(e \vee b_n) \leq 1$, $W(P_n, b_n) \geq E_{P_n}[e(X)] \geq \pi_{\min}$ by \Cref{lemma:useful_inequalities}. Then:
    \begin{align*}
        V_n & = \Var_{P_n}(\mu(X)) + E_{P_n}\left[ \Var_{P_n}\left( \phi_n \mid X \right) \right] = \Var_{P_n}(\mu(X)) +  E_{P_n}\left[ \frac{e(X) \Var(Y \mid X, D = 1)}{\max\{ e(X), b_n \}^2} \right].
    \end{align*}
    By \Cref{def:AllowedDistributions}, $0 \leq \Var_{P_n}(\mu(X)) \leq M$ and $\sigma_{\min}^2 \leq \Var(Y \mid X, D = 1) \leq M^2$. Therefore, $\sigma_{\min}^2 W(P_n, b_n) \leq V_n \leq M + M^2 W(P_n, b_n) \leq (M / \pi_{\min} + M^2) W(P_n, b_n)$, and $V_n \geq \sigma_{\min}^2 \pi_{\min}$.
\end{proof}

\begin{lemma}[Berry-Esseen condition verification]\label{lemma:BerryEsseenConditions}
    Define $\tilde{\phi}(Z \mid b, P) \equiv \phi(Z \mid b, \eta(P)) - \E_{P}[\mu(X)]$ for $P \in \mathscr{P}$. Further define $m(b, P, t) \equiv \E_{P}[  | \tilde{\phi}(Z \mid b, P) |^t ]$.

    Then the following hold:
    \begin{enumerate}
        \item $\E_{P}[\tilde{\phi}(Z \mid b, P)] = 0$ 
        \item $V_n > 0$ 
        \item $m(b, P, 3) < \infty$ (though it may be arbitrarily large) 
        \item Let $b_n > 0$ and let $P_n$ be a sequence of distributions in $\mathscr{P}$. If (i) $\gamma_0 > 2$ or (ii) $n^{-1/2} \ll b_n$, then there exists an $\epsilon > 0$ such that $\frac{m(b_n, P_n, 2 + \epsilon)}{V_n^{1+\epsilon/2} n^{\epsilon/2}} = o(1)$.\label{subitem:BEFinalRemainder}
    \end{enumerate}
\end{lemma}

\begin{proof}[Proof of \Cref{lemma:BerryEsseenConditions}]
    $\E_{P}[\tilde{\phi}(Z \mid b_n, P)] = 0$ is immediate.

    $V_n > 0$ follows by \Cref{lemma:PhiLowerBound}. 

    For the third moment being finite:
    \begin{align*}
        m(b, P, 3) & = \E_{P}[  | \tilde{\phi}(Z \mid b, P) |^3 ]  \leq 8 \E_{P}\left[ | \mu(X) - \E_{P}[\mu(X)] |^3 + b^{-3} e(X) |Y - \mu(X) |^3 \right] \\
        & \leq O(M^3) + 16 b^{-3} \E_{P}\left[ |Y|^3 \right]. 
    \end{align*}
    This is finite (and $O(b^{-3}) O(M^q)$) by assumption. 

    Finally, I have the claim for sequences. Under sufficient condition (i), take some $\epsilon \in (0, \min\{ \gamma_0 - 2, q - 2 \})$. Then $V_n = \Theta(1)$ and $m(b_n, P_n, 2 + \epsilon) = \Theta(1)$, so that the claim holds. 
    
    Under sufficient condition (ii), recall that by Lemmas \ref{lemma:PhiLowerBound} and \ref{lemma:useful_inequalities}, $\frac{1}{V_n^{3/2} \sqrt{n}} = o(1)$ and $\E_{P_n}\left[ \frac{D}{\max\{e(X), b_n\}^2} \right] \geq \pi_{\min} / 2$. Take $\epsilon = 1 > 0$. Then:
    \begin{align*}
        \frac{m(b_n, P_n, 3)}{V_n^{3/2} \sqrt{n}} & \leq 8 \frac{\E_{P_n}\left[ \frac{D |Y - \mu(X)|^3}{\max\{e(X), b_n\}^3} + | \mu(X) - \E_{P_n}[\mu(X)] |^3 \right]}{V_n^{3/2} \sqrt{n}} \\ 
        & \leq O(M^3) \frac{\E_{P_n}\left[ \frac{D}{\max\{e(X), b_n\}^2} \right]}{b_n V_n^{3/2} \sqrt{n}} + \frac{O(M^3)}{V_n^{3/2} \sqrt{n}} \\ 
        & = O(M^3 \sigma_{\min}^{-3}) \E_{P_n}\left[ \frac{D}{\max\{e(X), b_n\}^2} \right]^{-1/2} (b_n^2 n)^{-1/2} + o(1) = o\left( 1 \right). \qedhere
    \end{align*} 
\end{proof}

\subsection{Bias, Variance, and Consistency}\label{subsec:BlackBoxBiasVar}

\begin{lemma}[Integral order]\label{lemma:IntegralOrder}
    Suppose $0 < a_n \ll b_n$ and $s > 0$. Then if $s \neq 1$, $\int_{a_n}^{b_n} t^{-s} d t = \Theta\big( \max\big\{ a_n^{1-s}, b_n^{1-s} \big\} \big)$, and if $s = 1$, $\int_{a_n}^{b_n} t^{-s} d t = \Theta\big( \log(b_n / a_n ) \big)$, with a constant that only depends on $s$. 
\end{lemma}

\begin{proof}[Proof of \Cref{lemma:IntegralOrder}]
    If $s \neq 1$, $\int_{a_n}^{b_n} t^{-s} d t = \big( b_n^{1-s} - a_n^{1-s} \big) / (1-s)$. If $s > 1$, then this is $\Theta( a_n^{1-s} \big)$ with $a_n^{1-s} \geq b_n^{1-s}$. If $s < 1$, then this is $\Theta\big( b_n^{1-s} \big)$, with $b_n^{1-s} \geq a_n^{1-s}$. And if $s = 1$, then $\int_{a_n}^{b_n} t^{-s} d t = \log( b_n / a_n )$ exactly. 
\end{proof}

\begin{lemma}[Bound on product of errors on small events]\label{lemma:ConditionalBiasBound}
    Let $\mathscr{P}$ be as defined in \Cref{def:AllowedDistributions}, and let the uniform bound $B_{\mu}$ on $\hat{\mu}$ be as defined in \Cref{assum:NuisanceRates}. For any $P \in \mathscr{P}$, $\pi \in (0, 1]$, and $r_{\mu,n}, r_{e,n} > 0$, let $\mathcal{S}(P, \pi)$ be the set of events $A$ under $P$ with probability at most $\pi$, let $\mathcal{M}_{\mu}(P, r)$ be the set of functions $\bar{\mu}(x)$ bounded by $B_{\mu}$ satisfying $\| \bar{\mu}(X) - \mu(X) \|_{L^{p_{\mu}}(P)} \leq r$, and let $\mathcal{M}_{e}(P, r)$ be the set of functions $\bar{e}(x)$ satisfying $\| \bar{e}(X) - e(X) \|_{L^{p_{e}}(P)} \leq r$. For an event $S \in \mathcal{S}(P, \pi)$ and functions $\bar{\mu} \in \mathcal{M}_{\mu}(P, r_{\mu,n})$ and $\bar{e} \in \mathcal{M}_{e}(P,r_{e,n})$, define:
   \begin{align*}
       B( P, S,  \bar{\mu}, \bar{e}) & = E_{P}\left[ S | \bar{\mu}(X) - \mu(X) | \ | \bar{e}(X) - e(X) | \right].
   \end{align*}
   Then 
   \begin{align*}
       B( P, S,  \bar{\mu}, \bar{e}) & \leq r_{e,n} \pi^{1 / \zeta} \min\left\{ r_{\mu,n}, (B_{\mu} + M) \pi^{1 / p_{\mu}} \right\}. \qedhere 
   \end{align*}
\end{lemma}

\begin{proof}[Proof of \Cref{lemma:ConditionalBiasBound}]
    If $p_{\mu} = p_{e} = 2$, then the formula holds by replacing $\pi$ with one (which weakly increases the upper bound of $B$) and applying Cauchy-Schwarz. I therefore proceed assuming $\max\{ p_{\mu}, p_{e} \} > 2$.

    Next, let the effective outcome consistency rate be
    \begin{align*}
        \tilde{r}_{\mu,n} &  = \min\left\{  r_{\mu,n},  (B_{\mu}+M) \pi^{1 / p_{\mu}}
        \right\}, 
    \end{align*}
    which applies a tighter rate if the uniform bound on $\bar{\mu}$ implies that the outcome regression errors on $S$ alone cannot generate an outcome regression error as large as $r_{\mu,n}$.

    Next, assume without loss of generality that $(1-S) ( \bar{\mu} - \mu ) = 0$ almost surely; this change does not affect the value of $B$. Note that as a result, $\| \bar{\mu} - \mu \|_{L^{p_{\mu}}(P)} \leq \tilde{r}_{\mu,n}$.

    I apply Hölder's inequality (and the equality that $S^t = S$ for all $t > 0$) to derive the bound:
    \begin{align*}
        B( P, S,  \bar{\mu}, \bar{e}) & \leq r_{e,n} E_{P}\left[ S | \bar{\mu} - \mu |^{\frac{p_{e}}{p_{e} - 1}} \right]^{\frac{p_{e}-1}{p_{e}}}. 
    \end{align*}
    Take $\tilde{p} = \frac{p_{\mu} (p_e - 1)}{ p_{e} }$, so that $\frac{\tilde{p}}{\tilde{p}-1} = \zeta \frac{p_{e} - 1}{p_{e}}$. Then
    \begin{align*}
        B( P, S,  \bar{\mu}, \bar{e}) & \leq r_{e,n} \tilde{r}_{\mu,n} E_{P}[ S ]^{\frac{p_{e}}{p_{e}-1} \zeta^{-1}  \frac{p_{e} - 1}{p_{e}} } \leq r_{e,n} \pi^{1 / \zeta} \tilde{r}_{\mu,n},
    \end{align*}
    completing the proof.  
\end{proof}

\begin{lemma}[Upper bound on mass with incorrect clipping]\label{lemma:BoundClipping}
    Suppose Assumptions \ref{def:AllowedDistributions} and
    \ref{assum:NuisanceRates} hold. Write $\hat{A}(s) = e(X)^s \mathbf{1}\{ |\hat{e}(X) - e(X)| \geq (e(X) - b_n) / 2 \}$.  Then for any $s \in [0, p_e)$ and any sequence of $P_n \in \mathscr{P}$, $E_{P_n}\big[ \hat{A}(s) \big] = O\big( b_n^{s + \gamma_0 - 1} + r_{e,n}^{(s + \gamma_0 - 1) p_e / (p_e + \gamma_0 - 1)} \big)$.  $\rho_n$ is recovered as the case $s = 0$. 
\end{lemma}

\begin{proof}[Proof of \Cref{lemma:BoundClipping}]
    Write $d_e = |\hat e(X) - e(X)|$ and $\pi_n = b_n + r_{e,n}^{p_e / (p_e + \gamma_0 - 1)}$, so that the claim is $E_{P_n}\big[ \hat{A}(s) \big] = O\big( \pi_n^{s + \gamma_0 - 1} \big)$. The claim is clear if $p_e$ is infinite, so proceed assuming $p_e$ is finite. 

    Partition $\{ e(X) > 2 \pi_n \}$ into windows $W_k = \{ e(X) \in [2^k \pi_n, 2^{k+1} \pi_n) \}$ for $k \geq 1$. For any fixed $k$, if $e(X) \in W_k$ and $\hat{A}(s) > 0$, then $d_e \geq 2^{k-2} \pi_n$. Therefore, by Markov's inequality and the inequality $\| d_e \|_{L^{p_{e}}(P_n)} \leq r_{e,n}$, $P_n( e(X) \in W_k, \hat{A}(s) > 0 ) \leq \big( r_{e,n} 2^{2-k} \pi_n^{-1} \big)^{p_e}$.

    Then:
    \begin{align*} 
        E_{P_n}\big[ \hat{A}(s) \big] & \leq (2 \pi_n)^s P_n( e(X) \leq 2 \pi_n ) + \sum_{k \geq 1} (2^{k+1} \pi_n)^{s}  \big( r_{e,n} 2^{2-k} \pi_n^{-1} \big)^{p_e} \\
        & \leq C 2^{\gamma_0-1+s} \pi_n^{s+\gamma_0-1} + \pi_n^{s-p_e} r_{e,n}^{p_e} \sum_{k \geq 1} 2^{(k+1) s - (k-2) p_e}  = O\big( \pi_n^{s + \gamma_0 - 1} \big),
    \end{align*}
    with the final inequality holding because $s < p_e$ and $r_{e,n}^{p_e} \lesssim \pi_n^{p_e + \gamma_0 - 1}$. 
\end{proof}

\begin{lemma}[Inverse propensity moment bound]\label{lemma:InvPropMoment}
    Fix some $p \geq 1$. Then for any sequence of $P_n$ in $\mathscr{P}$, $\| 1 / (e(X) \vee b_n) \|_{L^{p}(P_n)} = O\big( b_n^{(\gamma_0-1-p)/p} + \log(1/b_n)^{\mathbf{1}\{ p = \gamma_0 - 1 \} / p} \big)$. If in addition $P_n$ satisfies $P_n( e(X) \leq \pi ) \geq C' \pi^{\gamma_0-1}$ for some $C' > 0$, then $\| 1 / (e(X) \vee b_n) \|_{L^{p}(P_n)} = \Theta\big( b_n^{(\gamma_0-1-p) / p} + \log(1/b_n)^{\mathbf{1}\{ p = \gamma_0 - 1 \} / p} \big)$, and in particular, the asymptotic order of $w_n$ is  $w_n = \Theta\big( b_n^{\min\{\gamma_0-1-\zeta, 0\}/\zeta} \log(1/b_n)^{\mathbf{1}\{ \zeta = \gamma_0 - 1 \} / \zeta} \big)$.
\end{lemma}

\begin{proof}
    If $p = \infty$, then both claims are immediate. Otherwise, if $p$ is finite, 
    \begin{align*}
        \left\| \frac{1}{e(X) \vee b_n} \right\|_{L^{p}(P_n)}^p & = E_{P_n}\left[ ( e(X) \vee b_n )^{-p} \right] =   \int_{0}^\infty P_n\left( e(X) \vee b_n < t^{-1/p} \right) d t\\
        & =   \int_{0}^{b_n^{-p}} P_n\left( e(X) < t^{-1/p} \right) d t   \leq 1 + C \int_1^{b_n^{-p}} t^{\frac{1-\gamma_0}{p}} d t.
    \end{align*}
    By \Cref{lemma:IntegralOrder} and noting $b_n^{-p} \gg 1$, $\big\| (e(X) \vee b_n)^{-1} \big\|_{L^{p}(P_n)} = O\big( b_n^{(\gamma_0-1-p) / p} + \log(1/b_n)^{\mathbf{1}\{ p = \gamma_0 - 1 \} / p} \big)$. 
    
    If $P_n$ satisfies $P_n( e(X) \leq \pi ) \geq C' \pi^{\gamma_0-1}$ for some $C' > 0$, then the claim holds by replacing the inequality $\leq C \int_0^{b_n^{-p}} t^{\frac{1-\gamma_0}{p}} d t $ with the flipped inequality $\geq C' \int_0^{b_n^{-p}} t^{\frac{1-\gamma_0}{p}} d t $. The result for $w_n$ is the special case $p = \zeta$. 
\end{proof}

\begin{lemma}[Generalized asymptotic variance bound]\label{lemma:AsympVarBound}
    Suppose \Cref{def:AllowedDistributions} holds. Define $R = \frac{D}{\max\{e, b_n\}^2}$ and $\tilde{R} = \frac{e}{\max\{e, b_n\}^2}$. Then if $t > 1$, $\| \tilde{R} \|_{L^t(P_n)} \leq E_{P_n}[R] \left( E_{P_n}[b_n R] \right)^{(1-t) / t}$. Also note that when $t = \frac{p}{p-2}$, then $\frac{1-t}{t} = \frac{-2}{p}$.
\end{lemma}

\begin{proof}[Proof of \Cref{lemma:AsympVarBound}]
     Note that $E[\tilde{R}] = E[R]$, so that for $t > 1$ finite:
    \begin{align*}
        \| \tilde{R} \|_{L^t(P_n)} \leq \| \tilde{R} \|_{L^1(P_n)}^{1/t} \| \tilde{R} \|_{L^{\infty}(P_n)}^{(t-1) / t} = \| \tilde{R} \|_{L^1(P_n)} \left( \frac{\| \tilde{R} \|_{L^{\infty}(P_n)}}{\| \tilde{R} \|_{L^1(P_n)}} \right)^{\frac{t-1}{t}} \leq E_{P_n}[R] \left( E_{P_n}[b_n R] \right)^{\frac{1-t}{t}},
    \end{align*}
    and for $t$ infinite,
    \begin{align*}
        \| \tilde{R} \|_{L^{t}(P_n)} & \leq b_n^{-1} = E_{P_n}[R] \left( E_{P_n}[b_n R] \right)^{-1} = E_{P_n}[R] \left( E_{P_n}[b_n R] \right)^{(1-t) / t}.
    \end{align*}
    The equality for $t = \frac{p}{p-2}$ is immediate. 
\end{proof}

\begin{lemma}[Variance bound]\label{lemma:VarBound}
    Suppose \Cref{assumption:ConsistencyRatesSufficient}\ref*{assum:StrictOverlapRates} and \Cref{assumption:ConsistencyRatesSufficient}\ref*{assum:ControlOfSquaredDiff} hold, and $n^{-1/2} \ll b_n \ll 1$. Then $m_n^2 \ll V_n$  and $Var_{P_n}\left( \hat{\psi}_{clip}^{AIPW}(b_n) - \tilde{\psi}_{(Oracle)}^{AIPW}(b_n)  \right) = O\left( n^{-1} E_{P_n}[ ( \hat{\phi}_n - \phi_n )^2 ]  \right)$.  
\end{lemma}

\begin{proof}[Proof of \Cref{lemma:VarBound}]
    Fix one fold, with the index left implicit. Write $M(n) = E_{P_n}[ ( \hat{\phi}_n - \phi_n )^2 ]$.  Also write $d_e = | \hat{e}(X) - e(X) |$. Then I decompose:
    \begin{align*}
        M(n) & \leq 2 E_{P_n}\left[ \left( \left( 1 - D / (e \vee b_n) \right) (\hat{\mu} - \mu ) \right)^2 +  \left( \left( D / (e \vee b_n) - D / ( \hat{e} \vee b_n ) \right) (Y - \hat{\mu} ) \right)^2 \right] \\ 
        & \leq 2 E_{P_n} \left[ \left( 1 + e / ( e \vee b_n)^2 \right) \left( \hat{\mu}(X) - \mu(X) \right)^2 \right] \tag{L\ref*{lemma:VarBound}(A)} \\ 
        & \ \ + 2 E_{P_n} \left[ \mathbf{1}\left\{ e(X) \leq b_n\right\} \left( \frac{D}{b_n + d_e} - \frac{D}{b_n} \right)^2 \left( Y - \hat{\mu}(X) \right)^2 \right] \tag{L\ref*{lemma:VarBound}(B)} \\
        & \ \ + 2 E_{P_n} \left[ \mathbf{1}\left\{ e(X) \geq b_n \right\} \left( \frac{D}{e(X)} - \frac{D}{\hat{e}(X) \vee b_n} \right)^2 \left( Y - \hat{\mu}(X) \right)^2 \right]. \tag{L\ref*{lemma:VarBound}(C)}
    \end{align*}
    
    Write $R = D / ( e \vee b_n )^2$ and $\tilde{R} = e / (e \vee b_n)^2$, so that $E_{P_n}[ R ] = E_{P_n}[ e / ( e \vee b_n )^2 ]$. By \Cref{lemma:PhiLowerBound}, $V_n = \Theta( E_{P_n}[ R ]  )$. 

     If $\| \hat{\mu} - \mu \|_{L^\infty(P_n)} = o(1)$ under \Cref{assumption:ConsistencyRatesSufficient}\ref*{assum:ControlOfSquaredDiff}, then bound (A) with $p_{\mu} = \infty$ and $r_{\mu,n} = \| \hat{\mu} - \mu \|_{L^\infty(P_n)}$. Otherwise, $r_{\mu,n} \ll (b_n V_n)^{1/p_{\mu}}$ by assumption. Then, 
\begin{align*}
    L\ref*{lemma:VarBound}(A) &  \lesssim r_{\mu,n}^2 + r_{\mu,n}^2 E_{P_n}[ R ] \left( b_n E_{P_n}[R] \right)^{-2 / p_{\mu}} = o( V_n ).  \tag{L\ref{lemma:AsympVarBound}, A\ref{assumption:ConsistencyRatesSufficient}\ref{assum:ControlOfSquaredDiff}} 
\end{align*}

For the second term, fix some $\epsilon \in (0, 1)$. Then:
\begin{align*}
    L\ref*{lemma:VarBound}(B)  & \lesssim E_{P_n}\left[ e \mathbf{1}\{ e(X) \leq b_n \} \frac{d_e^2}{b_n^2 (b_n + d_e)^2} \right] \leq E_{P_n}\left[ \frac{e}{b_n^2} \mathbf{1}\{ e(X) \leq b_n \} \min\left\{ 1, (d_e / b_n)^2 \right\} \right] \\
    & \leq E_{P_n}\left[ \frac{e}{b_n^2} \mathbf{1}\{ e(X) \leq b_n, d_e \geq \epsilon b_n \} \min\left\{ 1, d_e / b_n \right\}^2 \right] \\
    & \ \ + E_{P_n}\left[ \frac{e}{b_n^2} \mathbf{1}\{ e(X) \leq b_n, d_e < \epsilon b_n \} \min\left\{ 1, d_e / b_n \right\}^2 \right] \\
    & \leq b_n^{-1} P_n \left(  d_e  \geq \epsilon b_n   \right) + \epsilon^2 E_{P_n}\left[ e b_n^{-2} \mathbf{1}\{ e(X) \leq b_n, d_e < \epsilon b_n \} \right].
\end{align*}
If $p_e$ is infinite, the assumption $\sup | \hat{e} - e | \ll b_n$ implies the first term is eventually zero. Otherwise, applying Markov's inequality to $d_e^{p_e}$, and regardless controlling the second term, for any $\epsilon \in (0, 1)$ fixed,
\begin{align*}
    L\ref*{lemma:VarBound}(B)  &  \lesssim \epsilon^{-p_e} b_n^{-1} (r_{e,n} / b_n)^{p_e} + \epsilon^2 E_{P_n}\left[ e / (e \vee b_n)^2 \right]  = o(V_n) + \epsilon^2 O(V_n). \tag{\Cref{assumption:ConsistencyRatesSufficient}\ref{assum:ControlOfSquaredDiff}}
\end{align*}
Therefore, for any fixed $\epsilon > 0$, $\limsup_{n \to \infty}  L\ref*{lemma:VarBound}(B) / V_n \leq \epsilon^2$. But $\limsup_{n \to \infty}  L\ref*{lemma:VarBound}(B) / V_n$ does not depend on $\epsilon$. Therefore $L\ref*{lemma:VarBound}(B) = o(V_n)$. 

For the third term, note that $p_e / (p_e - 2) > 1$ for $p_e > 2$ (and should be taken as infinity for $p_e = 2$). Then 
\begin{align*}
    \left\| (e \vee b_n) \left/ (\hat{e} \vee b_n) \right. - 1  \right\|_{L^{p_e}\left( P_n \right)}  & \leq \left\| (e \vee b_n)  - (\hat{e} \vee b_n) \right\|_{L^{p_e}\left( P_n \right)} / b_n \leq r_{e,n} / b_n,
\end{align*} 
so that 
\begin{align*}
    L\ref*{lemma:VarBound}(C) & \lesssim E_{P_n}\left[ \frac{e(X) \mathbf{1}\{ e(X) \geq b_n \}}{e(X)^2} \left( 1 - \frac{\max\{ e(X), b_n \}}{\max\{ \hat{e}(X), b_n \}} \right)^2 \right] \\
    & \lesssim \left\| \frac{\mathbf{1}\{e \geq b_n\}}{e \vee b_n} \right\|_{L^{\frac{p_e}{p_e - 2}}(P)} \left( \frac{r_{e,n}}{b_n} \right)^2 \leq  E_{P_n}[R] \left( \frac{r_{e,n}}{b_n^{1 + 1 / p_{e}} E_{P_n}\left[ R \right]^{1 / p_{e}}} \right)^{2}. {\tag{L\ref{lemma:AsympVarBound}}}
\end{align*}
The final replacement holds by inspection. But then, by \Cref{assumption:ConsistencyRatesSufficient}\ref*{assum:ControlOfSquaredDiff}, L\ref*{lemma:VarBound}(C) = $o\big( E_{P_n}[ R ] \big) = o( V_n )$.

Then, for the variance:
\begin{align*}
	Var_{P_n}\left[ \hat{\psi}_{clip}^{AIPW}(b_n) - \tilde{\psi}_{(Oracle)}^{AIPW}(b_n) \right] & = n^{-1} \Var_{P_n}[ \hat{\phi}_n - \phi_n ]  \\
    & = \frac{E_{P_n} [ (\hat{\phi}_n - \phi_n)^2 ] -  E_{P_n} [ (\hat{\phi}_n - \phi_n) ]^2}{n} \leq \frac{m_n^2}{n}. 
\end{align*}

The remaining conclusion holds after averaging across folds. 
\end{proof}

\subsection{Degradation of Consistency Rate}\label{proofs:ConsistencyRate}

\begin{lemma}\label{lemma:InvSqBoundAbove}
    Suppose the conditions of \Cref{prop:OracleAsympNormal} hold and let $P_n$ be a sequence of distributions in $\mathscr{P}$. Then, $\E_{P_n}\left[ \frac{D}{\max\{e(X), b_n\}^2} \right] = O\left( 1 + b_n^{\gamma_0 - 2} \log(1/b_n)^{\mathbf{1}\{ \gamma_0 = 2 \}} \right)$, with a constant that only depends on $C$ and $\gamma_0$. If in addition \Cref{assum:TailLowerBound} holds, then $\E_{P_n}\left[ \frac{D}{\max\{e(X), b_n\}^2} \right] = \Theta\left( 1 + b_n^{\gamma_0 - 2} \log(1/b_n)^{\mathbf{1}\{ \gamma_0 = 2 \}} \right)$. 
\end{lemma}

\begin{proof}[Proof of \Cref{lemma:InvSqBoundAbove}]
Let $P_n$ be given. Then:
\begin{align*}
    \E_{P_n}\left[ \frac{D}{(e(X) \vee b_n)^2} \right] & \leq \E_{P_n}\left[ \frac{1}{e(X) \vee b_n} \right] = \int_0^{\infty} P_n\left( \frac{1}{e(X) \vee b_n} > t \right) d t   \\
    & = \int_0^{\infty} P_n\left( e(X) \vee b_n < 1 / t \right) d t  = 1 + \int_1^{\infty} P_n\left( e(X) \vee b_n < 1 / t \right) d t \\
    & = 1 + \int_1^{1 / b_n} P_n\left( e(X) \vee b_n < 1 / t \right) d t  \leq 1 + C \int_1^{b_n^{-1}} t^{1 - \gamma_0} d t.
\end{align*}
By \Cref{lemma:IntegralOrder}, the first conclusion holds. 

If  \Cref{assum:TailLowerBound} also holds, then the conclusion holds as follows. Take $b_n' = b_n (2 C / C')^{1 / (\gamma_0-1)}$, so that $\pi \geq b_n'$ and continuity of the distribution lower bound means $P( b_n \leq e(X) < \pi ) \geq C' \pi^{\gamma_0-1} - C b_n^{\gamma_0-1} \geq C'' \pi^{\gamma_0-1}$ where $C'' = C' / 2$. Then, letting $n$ be large enough for $b_n \leq 1$,
\begin{align*}
    \E_{P_n}\left[ \frac{D}{(e(X) \vee b_n)^2} \right] & \geq E_{P_n}\left[ \frac{\mathbf{1}\{ e(X) \geq b_n\}}{e(X)} \right] = \int_0^{\infty} P_n\left( b_n \leq e(X) < 1 / t \right) d t \\
    & \geq C'' + \int_{1}^{1 / b_n'} P_n\left( b_n \leq e(X) < 1 / t \right) d t \geq C'' + C'' \int_1^{1 / b_n'} t^{1-\gamma_0} d t,
\end{align*}
with the result following by observing $b_n' = \Theta(b_n)$. 
\end{proof}

\begin{proof}[Proof of \Cref{prop:ConsistencyRates}]
    By i.i.d. sampling, $\sigma_n^2 = n^{-1} V_n$. Choose some sequence of $P_n \in \mathscr{P}$ and $b_n \in (0, 1]$. By iterated expectations, $E_{P_n}[ D / \max\{ e(X), b_n \}^2 ] = W( P_n, b_n )$. The claim then follows by \Cref{lemma:PhiLowerBound} and dividing by $n$. 
\end{proof}

\begin{proof}[Proof of \Cref{cor:WorstCaseConsistency}]
    This holds by \Cref{prop:ConsistencyRates}, \Cref{lemma:InvSqBoundAbove}, and inspection of $n s_n(b_n)$. 
\end{proof}

\subsection{Sufficient Black-Box Rates and Need for Improved Rates}\label{proofs:AsymptoticNormality}

\begin{lemma}[Estimated variance consistency] \label{lemma:estimated_variance_consistency}
    Suppose the conditions of \Cref{lemma:VarBound} hold. Let $P_n$ be a sequence of distributions in $\mathscr{P}$. Then $\hat{\sigma}_n^2 / \sigma_n^2 - 1 = o_{P_n}(1)$. 
\end{lemma}

\begin{proof}[Proof of \Cref{lemma:estimated_variance_consistency}]
    Note that $W(P_n, b_n) = E_{P_n}[e(X) / \max\{e(X), b_n\}^2] = O\big( 1 / b_n \big)$. By \Cref{lemma:PhiLowerBound}, $V_n = \Theta(W(P_n, b_n))$ and there is a $\ubar{v} > 0$ such that $V_n \geq \ubar{v}$ uniformly. By \Cref{lemma:useful_inequalities}\ref{item:qth_moment_bound}, there exists a $q \in (3, 4]$ such that $E[ |\phi_n|^q ] \leq (8 M)^q W(P_n, b_n) b_n^{2-q}$. By \Cref{def:AllowedDistributions}, $|\psi(P_n)| \leq M$.  By noting $\rho_n \leq 1$ and $r_{\mu,n} r_{e,n} / b_n \ll 1$ by \Cref{assumption:ConsistencyRatesSufficient}\ref*{assum:StrictOverlapRates} and $b_n \gg n^{-1/2}$ and then applying \Cref{prop:Consistency}, $\hat{\psi}_{n} = \psi_n + o_{P_n}(1) = O_{P_n}(1)$.

    Next, I show oracle variance consistency. By Markov's inequality applied to exponent $q / 2$ and then applying the von Bahr-Esseen inequality and Jensen's inequality, for any $P \in \mathscr{P}$, 
    \begin{align*}
        P\left( \left| \frac{\mathbb{P}_n[\phi_n^2] - E_{P}[\phi_n^2]}{V_n} \right| > t \right) & \leq \frac{E_{P}\left[ \left| \mathbb{P}_n[\phi_n^2] - E_{P}[\phi_n^2] \right|^{\frac{q}{2}} \right] }{t^{\frac{q}{2}} V_n^{\frac{q}{2}}} \leq \frac{2^{\frac{q+2}{2}} E_{P}\big[ |\phi_n|^{q} \big]}{n^{\frac{q-2}{2}} t^{\frac{q}{2}} V_n^{\frac{q}{2}}} = O\left( n^{\frac{2-q}{2}} b_n^{2-q} \right) = o(1),
    \end{align*}
    with constants independent of $P$. As a result, and applying $V_n \geq \ubar{v} > 0$,
    \begin{align*}
        \frac{\tilde{\sigma}_n^2}{\sigma_n^2} & = \frac{\mathbb{P}_n[\phi_n^2] - \mathbb{P}_n[\phi_n]^2}{V_n} = 1 + \frac{\psi(P_n)^2 - \mathbb{P}_n[\phi_n]^2}{V_n} + o_{P_n}(1) = 1 + o_{P_n}(1).
    \end{align*}

    Next, I show estimated variance consistency. Write $\hat{m}_n^2 = \mathbb{P}_n\left[ (\hat{\phi}_n - \phi_n)^2 \right]$ for the empirical analog of $m_n^2$. By \Cref{lemma:VarBound} and iterated expectations conditional on training data, $E_{P_n}[\hat{m}^2] = o(V_n)$, so $\hat{m}_n^2 = o_{P_n}\big( V_n \big)$ by Markov's inequality.  The key step is to note that by the Cauchy-Schwarz inequality,
    \begin{align*}
        \frac{\left| \mathbb{P}_n\left[ \hat{\phi}_n^2 - \phi_n^2 \right] \right|}{V_n} & \leq \frac{\hat{m}_n^2 + 2 \sqrt{\hat{m}_n^2 \mathbb{P}_n[\phi_n^2]}}{V_n} = o_{P_n}(1).
    \end{align*} 
    Then: 
    \begin{align*}
        \frac{\hat{\sigma}_n^2}{\sigma_n^2} & = \frac{\tilde{\sigma}_n^2}{\sigma_n^2} + \frac{\mathbb{P}_n\left[ \hat{\phi}_n^2 \right] - \mathbb{P}_n\left[ \phi_n^2 \right]}{V_n} + \frac{\mathbb{P}_n\left[ \phi_n \right]^2 - \mathbb{P}_n\left[ \hat{\phi}_n \right]^2}{V_n} = 1 + o_{P_n}(1). \qedhere 
    \end{align*}
\end{proof}

\begin{proof}[Proof of \Cref{cor:ATE}]
    For this proof, write the shorthand $\phi_d = \phi_d(Z \mid b_n^{(d)}, \eta )$, $\hat{\phi}_d = \phi_d(Z \mid b_n^{(d)}, \hat{\eta} )$, $\hat{\tau}^* = \frac{1}{n} \sum \phi_1 - \phi_0$, $\hat{\tau} = \frac{1}{n} \sum \hat{\phi}_1 - \hat{\phi}_0$, $V_n^{(d)} = \Var( \phi_d(Z \mid b_n^{(d)}, \eta ) )$, and $\sigma_n^2 = n^{-1} \Var_{P_n}\big( \phi_1 - \phi_0 \big)$. 

    Because $\phi_1$ and $\phi_0$ are conditionally uncorrelated given $X$ and $\Var(\mu_1 - \mu_0) = O(1)$, the argument of \Cref{lemma:PhiLowerBound} yields $n \Var(\hat{\tau}^*) = \Theta\big( V_n^{(1)} + V_n^{(0)} \big)$.
    
    Oracle equivalence follows by \Cref{thm:SecondOrderNuisances} applied to each potential outcome and applying the triangle inequality: $\frac{1}{n} \sum_{i=1}^n \phi_d(Z \mid b_n^{(d)}, \hat{\eta} ) - \phi_d(Z \mid b_n^{(d)}, \eta ) = o_{P_n}\big( \sqrt{V_n^{(d)} / n} \big)$ with $V_n^{(d)} = O\big( \Var_{P_n} ( \phi_1 - \phi_0 ) \big)$.

    By the Berry-Esseen argument of \Cref{prop:OracleAsympNormal} applied to $\phi_1 - \phi_0$, with third moments bounded arm-by-arm by \Cref{lemma:useful_inequalities}\ref{item:centered_qth_moment_bound}, $\sigma_n^{-1} \big( \hat{\tau}^* - \tau \big) \to^d \mathcal{N}(0, 1)$. By the argument of \Cref{lemma:estimated_variance_consistency} applied to $\phi_1 - \phi_0$ with $\min\{ b_n^{(1)}, b_n^{(0)} \}$ in place of $b_n$ and noting $Var_{P_n}( \phi_1 - \phi_0 ) \gtrsim 1$, $\sigma_n / \hat{\sigma}_n - 1 = o_{P_n}(1)$.
    Then, the feasible t-statistics satisfy $\hat{\sigma}_n^{-1} \big( \hat{\tau} - \tau \big) = \big( 1 + o_{P_n}(1) \big) \sigma_n^{-1} \big( \hat{\tau}^* - \tau \big) + o_{P_n}(1) \to^d \mathcal{N}(0, 1)$.  Nominal coverage follows. 
\end{proof}

\begin{proof}[Proof of \Cref{cor:TtestsUnderSomewhatWeakOverlap}]
    For this proof, set $a_n = \big( n^{1/2} r_{\mu,n} \big)^{-p_{\mu} / ((\gamma_0-1)(p_{\mu} - 1))}$, and set $c_n = n^{-1/\gamma_0}$ under the additional condition and $c_n = 1$ otherwise. I claim that there is a $b_n$ such that 
    \begin{align*}
        \min\big\{ 1, a_n, c_n \big\} \gg b_n \gg \max\big\{ n^{-1/2}, r_{e,n}^{p_e / (p_e + \gamma_0 - 1)} \big\}.
    \end{align*}
    Both right-hand terms are $o(c_n)$: trivially if $c_n = 1$, and by $\gamma_0 > 2$ and the additional condition if $c_n = n^{-1/\gamma_0}$. It only remains to show that both are $o\big( a_n \big)$. By assumption, $r_{\mu,n} \ll n^{-1/(2 p_{\mu})} \ll n^{\big( (\gamma_0-1)(p_{\mu} - 1) - p_{\mu} \big) / (2 p_{\mu} )}$. Further, $a_n^{(\gamma_0-1)(p_{\mu}-1)/p_{\mu}} \gg r_{e,n}^{(\gamma_0-1) \frac{p_{\mu}-1}{p_{\mu}} \frac{p_e}{p_e + \gamma_0-1}}$, so $r_{e,n}^{p_e / (p_e + \gamma_0 - 1)} \ll a_n$. Choose such a $b_n$. Note further that $b_n \gg n^{-1/2} \gg r_{\mu,n}^{p_{\mu}}$. 
    
    I show the elements of \Cref{assumption:ConsistencyRatesSufficient} other than \ref*{assum:ControlOfUnclippedBias}. Condition \ref*{assum:StrictOverlapRates} holds by assumption, since the product exponent is at most one. Then \ref*{assum:ControlOfClippedBias} holds by $b_n \ll a_n$ and $P( e(X) \leq b_n ) \leq C b_n^{\gamma_0 - 1}$.  By construction of $b_n$ and that $V_n = \Theta(1)$ under $\gamma_0 > 2$, $r_{e,n} \ll b_n^{1 + 1 / p_e}$ and $b_n V_n = \Theta( b_n ) \gg  r_{\mu,n}^{p_{\mu}}$ \ref*{assum:ControlOfSquaredDiff}.

    Now I show \ref*{assum:ControlOfUnclippedBias}. By \Cref{lemma:InvPropMoment}, $w_n = O\big( b_n^{(\gamma_0-1)/\zeta - 1} + \log(1/b_n)^{\mathbf{1}\{ \zeta = \gamma_0 - 1 \}} \big)$. By \Cref{lemma:BoundClipping} and construction of $b_n$, $\rho_n^{1/\zeta} / b_n = O\big( b_n^{(\gamma_0-1)/\zeta - 1} \big)$. Therefore it remains to show that $r_{\mu,n} r_{e,n} b_n^{(\gamma_0-1)/\zeta - 1} \ll n^{-1/2}$ and $r_{\mu,n} r_{e,n} \log(1/b_n)^{\mathbf{1}\{ \gamma_0 = \zeta + 1 \}} \ll n^{-1/2}$. If $\gamma_0 > \zeta + 1$, then this holds by $r_{\mu,n} r_{e,n} \ll n^{-1/2}$. If $\gamma_0 \leq \zeta + 1$, write $r_n^+ = b_n^{(\gamma_0-1)/\zeta - 1} \log(b_n^{-1})^{\mathbf{1}\{ \gamma_0 = \zeta + 1\}}$. Then 
    $$r_{\mu,n} r_{e,n} r_n^+ \ll r_{\mu,n} r_{e,n}^{ \frac{(\gamma_0-1) p_e}{p_e + \gamma_0-1} \big( \frac{1}{p_e} + \frac{1}{\zeta} \big)} \log(b_n^{-1})^{\mathbf{1}\{ \gamma_0 = \zeta + 1\}} \lesssim r_{\mu,n} r_{e,n}^{\frac{(\gamma_0-1) p_e}{p_e + \gamma_0-1} \frac{p_{\mu} - 1}{p_{\mu}}} \log(r_{e,n}^{-1})^{\mathbf{1}\{ \gamma_0 = \zeta + 1\}} \ll n^{-1/2}.$$ Therefore, for this sequence of $b_n$, \Cref{assumption:ConsistencyRatesSufficient} holds.  

    Let $P_n$ be a sequence of distributions in $\mathscr{P}$. By \Cref{thm:SecondOrderNuisances} and $V_n = \Theta(1)$, $\hat{\psi}_{clip}^{AIPW}(b_n) - \tilde{\psi}_{(Oracle)}^{AIPW}(b_n) = o_{P_n}(\sigma_n)$, where $\sigma_n = \Theta(n^{-1/2})$. Further, $\tilde{\psi}_{(Oracle)}^{AIPW}(b_n) - \tilde{\psi}_{(Oracle)}^{AIPW}(0)$ is mean-zero, and by \Cref{def:AllowedDistributions} and \Cref{lemma:IntegralOrder},
    \begin{align*}
        n E_{P_n}\left[ \left( \tilde{\psi}_{(Oracle)}^{AIPW}(b_n) - \tilde{\psi}_{(Oracle)}^{AIPW}(0) \right)^2 \right] & \leq E_{P_n}\left[ \sigma^2(X) \mathbf{1}\{ e(X) \leq b_n \} / e(X) \right] \lesssim b_n^{\gamma_0-2},
    \end{align*}
    so that $\tilde{\psi}_{(Oracle)}^{AIPW}(b_n) = \tilde{\psi}_{(Oracle)}^{AIPW}(0) + o_{P_n}(n^{-1/2})$ by Chebyshev's inequality. Combining arguments and noting $\tilde{\psi}_{(Oracle)}^{AIPW}(0) = \mathbb{P}_n[ \phi(Z \mid 0, \eta) ]$, $\hat{\psi}_{clip}^{AIPW}(b_n) = \tilde{\psi}_{(Oracle)}^{AIPW}(b_n) + o_{P_n}(n^{-1/2})$.

    Now suppose the additional assumption $n^{1/\gamma_0} r_{e,n}^{p_e / (p_e + \gamma_0 - 1)} \to 0$ holds, so that $b_n \ll n^{-1/\gamma_0}$ by construction. $\hat{\psi}_{clip}^{AIPW}(b_n) = \hat{\psi}_{clip}^{AIPW}(0)$ and $\hat{\sigma}_n(b_n) = \hat{\sigma}_n(0)$ in finite samples if every observation with $\hat{e}(X) < b_n$ has $D = 0$. Note that $\{ D = 1, \hat{e} \leq b_n \} \subseteq \{ D = 1, |\hat{e} - e| \geq (e - b_n)/2 \}$. \Cref{lemma:BoundClipping} applied to $s = 1$, iterated expectations over $D$, and $r_{e,n}^{p_e / (p_e + \gamma_0 - 1)} \ll b_n$ yield $P_n( D = 1, \hat{e}(X) \leq b_n ) = O( b_n^{\gamma_0} ) = o(1/n)$, with a constant that is uniform in $P_n$. By a union bound, with probability tending to one, $\hat{\psi}_{clip}^{AIPW}(0) = \hat{\psi}_{clip}^{AIPW}(b_n)$, so that $\hat{\psi}_{clip}^{AIPW}(0) = \hat{\psi}_{clip}^{AIPW}(b_n) = \tilde{\psi}_{(Oracle)}^{AIPW}(b_n) + o_{P_n}(n^{-1/2})$, and Wald confidence intervals for $\hat{\sigma}_n(0) = \hat{\sigma}_n(b_n)$ follow from \Cref{thm:SecondOrderNuisances} applied to $b_n$. 
\end{proof}

\begin{proof}[Proof of \Cref{prop:SufficientRates}]
    In this special case of $p_{\mu} = p_{e} = p$, then $\zeta = p / (p-2)$ and $p / \zeta = (p - 2)$. Write $\tilde{t} = \big( 2 (\gamma_0-1) (p_{\mu} - 1) + p_{\mu} \max\{2 - \gamma_0, 0\} \big) / (2 p_{\mu})$, so that the statement's hypothesis is $r_{\mu,n} r_{e,n}^{\min\{ \frac{\tilde{t} p_e}{p_e + \gamma_0 - 1}, 1 \}} \ll n^{-1/2}$.

    Take some $\delta_n \to \infty$ such that $r_{\mu,n} r_{e,n}^{\frac{\tilde{t} p_e}{p_e + \gamma_0 - 1}} \delta_n^{\tilde{t}} \ll n^{-1/2}$. Take $b_n = \delta_n r_{e,n}^{\frac{p_e}{p_e + \gamma_0 - 1}} \gg r_{e,n}$.

    Recall the results for $V_n, w_n, \rho_n$ from \Cref{lemma:InvSqBoundAbove,lemma:InvPropMoment,lemma:BoundClipping} and note that   $P_n( e(X) \leq b_n ) = \Theta( b_n^{\gamma_0-1} )$ under \Cref{assum:TailLowerBound}. 

    \Cref{assumption:ConsistencyRatesSufficient}\ref{assum:StrictOverlapRates} holds by assumption.

    \Cref{assumption:ConsistencyRatesSufficient}\ref{assum:ControlOfClippedBias} holds by
    \begin{align*}
        r_{\mu,n} b_n^{(\gamma_0-1)(p_{\mu} - 1) / p_{\mu}} / \sqrt{V_n} & = \Theta\left( r_{\mu,n} b_n^{\tilde{t}} \log(1/b_n)^{-\mathbf{1}\{ \gamma_0 = 2\} / 2} \right) \lesssim r_{\mu,n} r_{e,n}^{ \frac{p_e \tilde{t}}{p_e + \gamma_0-1} } \delta_n^{\tilde{t}} \ll n^{-1/2}. 
    \end{align*}

    \Cref{assumption:ConsistencyRatesSufficient}\ref{assum:ControlOfUnclippedBias} holds by observing that $(\gamma_0-1) / \zeta - \min\{ \gamma_0/2, 1 \}  < \tilde{t}$. If $\gamma_0 > 1 + \zeta \geq 2$,  then $V_n = \Theta(1)$ by \Cref{lemma:InvSqBoundAbove,lemma:PhiLowerBound} and $\rho_n^{1/\zeta} / b_n, w_n = O(1)$ by \Cref{lemma:BoundClipping,lemma:InvPropMoment}, so that the hypothesis $r_{\mu,n} r_{e,n} \ll n^{-1/2}$ suffices. If $\gamma_0 < 1 + \zeta$, then the definition of $\delta_n$ above ensures:
    \begin{align*}
        \frac{r_{\mu,n} r_{e,n} \left( w_n + \rho_n^{1/\zeta} / b_n \right)}{\sqrt{V_n}} & = O\left( r_{\mu,n} r_{e,n} b_n^{\frac{\gamma_0-1}{\zeta} - \min\left\{ \frac{\gamma_0}{2}, 1 \right\}} \log(1/b_n)^{\frac{-\mathbf{1}\{ \gamma_0 = 2\}}{2}} \right) \lesssim r_{\mu,n} r_{e,n}^{\frac{\tilde{t} p_e}{p_e + \gamma_0-1}} \delta_n^{\tilde t} \ll n^{-1/2}.
    \end{align*}

    \Cref{assumption:ConsistencyRatesSufficient}\ref{assum:ControlOfSquaredDiff} holds as follows. For the outcome term, $r_{\mu,n} \left( b_n V_n \right)^{-1/p_{\mu}} \lesssim r_{\mu,n} b_n^{-1/p_{\mu}} \to 0$, because $b_n \gg r_{e,n}^{p / (p + \gamma_0 - 1)} \gtrsim n^{-1/2} \gg r_{\mu,n}^{p}$. For the propensity term, if $\gamma_0 > 2$, then $r_{e,n} b_n^{-1-1/p_e} V_n^{-1/p_e} = \Theta\big( r_{e,n} b_n^{-1-1/p_e}  \big) = o(1)$. Otherwise, 
        \begin{align*}
        r_{e,n} b_n^{-1-1/p_e} V_n^{-1/p_e} & \lesssim r_{e,n} b_n^{-1-1/p_e} b_n^{(2-\gamma_0) / p_e} = r_{e,n} \big/ b_n^{(p_e + \gamma_0 - 1)/ p_e} \to 0.
    \end{align*}

    Thus, Assumptions \ref{def:AllowedDistributions}, \ref{assum:NuisanceRates}, and \ref{assumption:ConsistencyRatesSufficient} hold.% By \Cref{thm:SecondOrderNuisances}, clipped AIPW t-statistics are asymptotically well-calibrated. 
\end{proof}

\subsection{Efficiency and Consistency Rates}\label{proofs:RateThings}

\begin{lemma}[Counterexample construction for \Cref{prop:NoBetterProduct}]\label{lemma:NBP_construction}
    Fix some $\gamma_0 > 1$, some $p_{\mu}, p_{e} \geq 2$, and some sequence of  $0 < b_n < \pi_n \leq 1$ and some sequence of $0 \leq r_{\mu,n} \leq \Delta_n^{1/p_{\mu}}$, where $\Delta_n = \pi_n^{\gamma_0-1} - b_n^{\gamma_0-1}$.  Let $P$ have $X\sim Unif([0,1])$, $D \mid X \sim Bern\big( e(X) \big)$ for $e(X) = X^{1/(\gamma_0-1)}$, $Y \mid D,X \sim \mathcal N(\mu(X, D), 1)$ for $\mu(x, d) = 0$. Write $d_n = r_{\mu,n} \Delta_n^{-1/p_{\mu}}$ and set the deterministic nuisance function estimates
    \[
        \hat{e}(x) = e(x) + (b_n - e(x)) \mathbf{1}\{ e(x) \in (b_n, \pi_n] \} \quad \text{and} \quad
        \hat\mu(x) = \mu(x) - d_n \mathbf{1}\{e(x) \in (b_n,\pi_n]\}. 
    \]
    Then the family $\mathscr{P} = \{ P \}$ satisfies \Cref{def:AllowedDistributions} and \Cref{assum:TailLowerBound} with equal tail parameters $C = C' = 1$, and the nuisance functions $\hat{e}$ and $\hat{\mu}$ satisfy \Cref{assum:NuisanceRates} with consistency rate guarantees $r_{e,n} = \| \hat{e} - e \|_{L^{p_e}(P)} = \Theta \big( \Delta_n^{1/p_e} (\pi_n-b_n) \big)$ and $\| \hat{\mu} - \mu \|_{L^{p_\mu}(P)} = \Theta\big( r_{\mu,n} \big)$. 
\end{lemma}

\begin{proof}
    For every $\pi \in (0, 1]$, $P( e(X) \leq \pi ) = P( X \leq \pi^{\gamma_0-1} ) = \pi^{\gamma_0-1}$,  so \Cref{def:AllowedDistributions}\ref{item:PropensityTail}
    and \Cref{assum:TailLowerBound} hold with $C=C'=1$.  The remaining conditions hold for appropriate constants because $Y\mid X,D\sim\mathcal N(0,1)$. Thus, \Cref{def:AllowedDistributions} and \Cref{assum:TailLowerBound} hold for $\mathscr{P}$. 

    Next, note that the nuisance function estimates are deterministic, so that they are cross-fit. To verify the nuisance rates, note that $P\big( e(X) \in (b_n, \pi_n] \big) = \pi_n^{\gamma_0-1} - b_n^{\gamma_0-1} = \Delta_n$.  Then 
    \begin{align*}
        \| \hat{e} - e \|_{L^{p_e}(P)} & \leq \left( \Delta_n (\pi_n - b_n)^{p_e}  \right)^{1/p_e} = \Delta_n^{1/p_e} (\pi_n - b_n) = r_{e,n} \\
         \| \hat{\mu} - \mu \|_{L^{p_\mu}(P)} & = d_n \Delta_n^{1/p_{\mu}} = r_{\mu,n}. 
    \end{align*}
    Further, the propensity rate guarantee is a lower bound up to some constant: 
    \begin{align*}
        \| \hat{e} - e \|_{L^{p_e}(P)} & \geq \| (\hat{e} - e) 1\{ | e - b_n | \geq (\pi_n - b_n) / 2 \} \|_{L^{p_e}(P)} = \Theta\big( r_{e,n} \big). 
    \end{align*}

    Noting that $r_{e,n} \leq \pi_n = O(1)$, $r_{\mu,n} \leq \Delta_n^{1/p_{\mu}} = O(1)$, and that $\| \hat{\mu} \|_{L^\infty(P_n)} = d_n \leq 1$ completes the remaining conditions of \Cref{assum:NuisanceRates}.
\end{proof}

\begin{lemma}[Useful results for \Cref{prop:NoBetterProduct}]\label{lemma:NBP_floor}
    Suppose $\gamma_0 - 1 < \zeta$ and $p_{\mu}, p_{e} \geq 2$. For $b_n \in (0, 1)$ and $\Delta_n > 0$, let $\pi_n = (b_n^{\gamma_0-1} + \Delta_n)^{1/(\gamma_0-1)}$. Then (i) $\frac{\zeta}{2} \min\{ \gamma_0, 2 \} > \gamma_0 - 1$ and (ii) if $\Delta_n \lesssim b_n^{\gamma_0-1}$, then $\pi_n - b_n = \Theta\big( \Delta_n b_n^{2-\gamma_0} \big)$. 
\end{lemma}

\begin{proof}
    Rewrite $u_n = \Delta_n b_n^{1-\gamma_0}$. Note that $\pi_n = b_n (1 + u_n)^{1/(\gamma_0-1)}$, so that $\pi_n - b_n = b_n \left(  (1+u_n)^{1/(\gamma_0-1)} - 1 \right)$.   

    For result (i), if $\gamma_0 < 2$, then the bound $\zeta \geq 1$ implies $\frac{\zeta}{2} \min\{ \gamma_0, 2 \} > \gamma_0 / 2 > \gamma_0 - 1$. If $\gamma_0 \geq 2$, then the bound $\zeta > \gamma_0 - 1$ yields $\frac{\zeta}{2} \min\{ \gamma_0, 2 \} > \gamma_0 - 1$. 

    For the result (ii), Taylor's theorem implies $(1+u_n)^{1/(\gamma_0-1)} - 1 = \Theta( u_n )$, so that $\pi_n - b_n = \Theta( b_n u_n )$ and the desired result holds. 
\end{proof}

\begin{lemma}[Existence of sequences for \Cref{prop:NoBetterProduct}]\label{lemma:NBP_sequences}
    Suppose $\gamma_0 < 1 + \zeta$ and $p_{\mu}, p_{e} \geq 2$. Write $\nu(t) = 2 t + 1 - (\gamma_0 - 1) - \min\{ \gamma_0 / 2, 1 \}$. Take some $t \in \big( \gamma_0-1, \zeta \min\{ \gamma_0 / 2, 1 \} \big)$ and some $b_n$ satisfying $n^{-1/2} \ll b_n \ll 1$ and $n^{-1/2} \log(n) \ll b_n^{\nu(t)}$. Set $\Delta_n = b_n^{t}$, $\pi_n = \big( \Delta_n + b_n^{\gamma_0-1} \big)^{1 / (\gamma_0 - 1)}$,  and $r_{e,n}$ as in \Cref{lemma:NBP_construction}. Then there is a sequence $d_n \to 0$ such that with $r_{\mu,n} = d_n \Delta_n^{1/p_{\mu}}$, the construction  in \Cref{lemma:NBP_construction} satisfies 
    \begin{enumerate}[label=(\roman*), itemsep=-0.5ex, topsep=-0.5ex]
        \item \emph{Both nuisances are consistent}. $r_{\mu,n}, r_{e,n} \to 0$.  \label{item:NBP:consistency}
        \item \emph{\Cref{assumption:ConsistencyRatesSufficient}\ref{assum:StrictOverlapRates},\ref{assum:ControlOfSquaredDiff} hold}. $r_{\mu,n} r_{e,n} \ll n^{-1/2}$, $r_{e,n} \ll b_n^{1+1/p_e} V_n^{1/p_e}$, and $\| \hat{\mu} - \mu \|_{L^\infty(P_n)} = d_n =  o(1)$. \label{item:NBP:productrate}
        \item \emph{Bias dominates standard error}. The bias satisfies $0 < \sigma_n^{-1} (B_U + B_C) \to \infty$.\label{item:NBP:biasStillBig}
    \end{enumerate}
\end{lemma}

\begin{proof}[Proof of \Cref{lemma:NBP_sequences}]
    First, note that by \Cref{lemma:NBP_floor}(i),  such a $t$ exists. 

    Next, I calculate a few orders as follows. Since $t > \gamma_0 - 1$, $\Delta_n \ll b_n^{\gamma_0-1}$, so \Cref{lemma:NBP_floor}(ii) implies $\pi_n - b_n = \Theta\big( b_n^{t + 2-\gamma_0} \big)$ and $\Delta_n (\pi_n - b_n) = \Theta\big( b_n^{2 t + 2-\gamma_0} \big)$. By definition in \Cref{lemma:NBP_construction}, $r_{e,n} = \Theta\big( \Delta_n^{1/p_e} (\pi_n-b_n) \big) = \Theta\big( b_n^{t (p_e + 1) / p_e + 2 - \gamma_0} \big)$.

    Now, I construct an interval for $r_{\mu,n}$. Note that by construction of $\nu(t)$, $$\sigma_n b_n = \Theta\big( n^{\frac{-1}{2}} b_n^{-\nu(t)} \log(b_n^{-1})^{\mathbf{1}\{\gamma_0=2\} / 2} \Delta_n (\pi_n - b_n) \big) = o\big( \Delta_n (\pi_n - b_n) \big).$$
    Since $\Delta_n (\pi_n - b_n ) = \Delta_n^{1/p_{\mu} + 1/\zeta} r_{e,n}$, I obtain $\Delta_n^{1/p_{\mu}} \gg \sigma_n b_n \Delta_n^{-1/\zeta} r_{e,n}^{-1}$. Since $t < \zeta \min\{\gamma_0/2, 1\}$, $\sigma_n b_n \Delta_n^{-1/\zeta} \ll n^{-1/2}$. 
    
    Now I define $r_{\mu,n}$ and $d_n$. There exists a sequence of $r_{\mu,n}$ satisfying $\sigma_n b_n \Delta_n^{-1/\zeta} r_{e,n}^{-1} \ll r_{\mu,n} \ll \min\big\{ \Delta_n^{1/p_{\mu}}, n^{-1/2} r_{e,n}^{-1} \big\}$.  Take such a sequence and set $d_n = r_{\mu,n} \Delta_n^{-1/p_{\mu}}$.  Because $r_{\mu,n} \ll \Delta_n^{1/p_{\mu}}$, $d_n \to 0$. But since $r_{\mu,n} \leq d_n$, $r_{\mu,n} \to 0$. In addition, $r_{e,n} \to 0$, so \ref{item:NBP:consistency} holds.

    By construction, $r_{\mu,n} r_{e,n} \ll n^{-1/2}$ and $\| \hat{\mu} - \mu \|_{L^\infty(P_n)} = o(1)$. Further, the construction in \Cref{lemma:NBP_construction} satisfies \Cref{assum:TailLowerBound}, so by \Cref{cor:WorstCaseConsistency},
    \begin{align*}
        b_n^{1+1/p_e} V_n^{1/p_e} & \gtrsim b_n^{1 + \min\{ \gamma_0-1, 1 \} / p_e} \log(1/b_n)^{\mathbf{1}\{\gamma_0 = 2\}/p_e}.
    \end{align*}
    $r_{e,n} = \Theta\big( b_n^{t (p_e + 1) / p_e + 2 - \gamma_0} \big) \to 0$. The requirement $t > \gamma_0 - 1$ implies $r_{e,n} \ll b_n^{1+1/p_e} V_n^{1/p_e}$. Therefore \ref{item:NBP:productrate} holds.

    For  \ref{item:NBP:biasStillBig},  the bias satisfies $B_C = 0$. Applying that $\pi_n = b_n + o(b_n)$,
    \begin{align*}
        B_U  & = d_n b_n^{-1} E_{P_n}\left[ (e(X) - b_n) \mathbf{1}\{ e(X) \in (b_n, \pi_n ] \}  \right] = \Theta\left( d_n b_n^{\gamma_0-3} (\gamma_0 - 1) \int_{b_n}^{\pi_n} (t - b_n) d t \right) \\
        & = \Theta\left( d_n b_n^{\gamma_0-3} (\pi_n - b_n)^2 \right) = \Theta\left( d_n b_n^{1-\gamma_0} \Delta_n^2  \right)  = \Theta\left( d_n \Delta_n (\pi_n - b_n) b_n^{-1} \right) \tag{L\ref{lemma:NBP_floor}(ii)} \\
        & = \Theta\left(  r_{\mu,n} r_{e,n} \Delta_n^{1/\zeta} b_n^{-1} \right), \tag{$d_n = r_{\mu,n} \Delta_n^{-1/p_{\mu}}, \ \Delta_n (\pi_n - b_n) = r_{e,n} \Delta_n^{1/\zeta + 1/p_{\mu}}$}
    \end{align*}
    with constants depending only on $\gamma_0$. Then, by construction of $r_{\mu,n}$, $0 < \sigma_n^{-1}(B_U + B_C) \to \infty$.
\end{proof}

\begin{proof}[Proof of \Cref{prop:NoBetterProduct}]
    Let the construction be as in \Cref{lemma:NBP_construction} with sequences as in \Cref{lemma:NBP_sequences}. Take $t$ as in \Cref{lemma:NBP_sequences} and take $b_n = n^{-a}$ for any $0 < a < \min\{1/2, 1/(2 \nu(t))\}$. By \Cref{lemma:NBP_construction}, \Cref{prop:NoBetterProduct}(\ref{item:PNBP:Overlap}) holds. By \Cref{lemma:NBP_construction} and \Cref{lemma:NBP_sequences}\ref{item:NBP:productrate}, \Cref{prop:NoBetterProduct}.\ref{item:PNBP:Crossfit} holds. It therefore only remains to show \Cref{prop:NoBetterProduct}.\ref{item:PNBP:Coverage}.

    I verify the conditions of \Cref{lemma:VarBound}. \Cref{assumption:ConsistencyRatesSufficient}\ref{assum:StrictOverlapRates},\ref{assum:ControlOfSquaredDiff} hold by \Cref{lemma:NBP_sequences}\ref{item:NBP:productrate}. By \Cref{lemma:NBP_sequences}, $n^{-1/2} \ll b_n \ll 1$ and $r_{\mu,n} \to 0$. By \Cref{lemma:NBP_sequences}, the conditions of \Cref{lemma:VarBound} hold, so by \Cref{lemma:estimated_variance_consistency}, $\hat{\sigma}_n^2 / \sigma_n^2 - 1 = o_{P_n}(1)$.

    Now, fix some $\alpha \in (0, 1)$, and note that $\hat{\sigma}_n \geq 0$ by construction. Also note that by \Cref{lemma:VarBound},   $Var\big( \hat{\psi}_{clip}^{AIPW}(b_n) \big) = O( \sigma_n^2 )$. Thus, and applying Chebyshev's inequality, 
    \begin{align*}
        P_n( \psi_n \in \hat{\mathcal{C}}_n(\alpha)) & = P_n\left(  \hat{\psi}_{clip}^{AIPW}(b_n) \leq \psi_n + \hat{\sigma}_n z_{1-\alpha/2} \right) - P_n\left(  \hat{\psi}_{clip}^{AIPW}(b_n) \leq \psi_n + \hat{\sigma}_n z_{\alpha/2} \right) \\
        & \leq P_n\left(  \hat{\psi}_{clip}^{AIPW}(b_n) \leq \psi_n + \hat{\sigma}_n z_{1-\alpha/2} \right)  \\
        & = P_n\left(  \frac{\hat{\psi}_{clip}^{AIPW}(b_n) - E\left[ \hat{\psi}_{clip}^{AIPW}(b_n) \right]}{\sigma_n} \leq \frac{-(B_U + B_C)}{\sigma_n} + z_{1-\alpha/2} + o_{P_n}(1)  \right) \\
        & = P_n\left( z_{1-\alpha/2} + O_{P_n}(1) \leq - \sigma_n^{-1} (B_U + B_C) \right)   \to 0. \tag{L\ref{lemma:NBP_sequences}\ref{item:NBP:biasStillBig}}
    \end{align*}
    Therefore, Wald confidence intervals fail asymptotically. 
\end{proof}

\begin{proof}[Proof of \Cref{prop:OptimalVWORate}]
    Result (i) is immediate by \Cref{prop:SufficientRates} for $p=\infty$ and \Cref{thm:SecondOrderNuisances}.

    For result (ii), without loss of generality suppose that $p_{\mu}, p_{e} = \infty$. Fix some sequence of $b_n \to 0$, and let $\mathscr{P} = \{ P \}$ be the distribution constructed in \Cref{lemma:NBP_construction}. Take $\hat{e} = \max\big\{ e - r_{e,n} 1\{ e \in (b_n, b_n + r_{e,n}) \}, 0 \big\}$, so that $\hat{e} \vee b_n = b_n$ for $e < b_n + r_{e,n} = ``\pi_n''$ and $\hat{e} \vee b_n = e$ otherwise. Take $\hat{\mu} = r_{\mu,n} \big( \mathbf{1}\{ e < b_n \} - \mathbf{1}\{ e \in [b_n, \pi_n) \} \big)$, where $\mu = 0$ under this construction.  
    
    Items \ref{item:PNBP:Overlap} and \ref{item:PNBP:Crossfit} (other than possibly the $b_n$-dependent \Cref{assumption:ConsistencyRatesSufficient}\ref{assum:ControlOfSquaredDiff}) hold by the arguments of \Cref{prop:NoBetterProduct}. It remains to show the zero-coverage property \ref{item:PNBP:Coverage}. 

    Recall that for every $\pi \in (0, 1]$, $P( e(X) \leq \pi ) = P( X \leq \pi^{\gamma_0-1} ) = \pi^{\gamma_0-1}$. 
    
    The bias of this estimator is
    \begin{align*}
        B_U + B_C & = E_{P}\left[ (\hat{\mu} - \mu) (1 - e / (\hat{e} \vee b_n)) \right] = r_{\mu,n} b_n^{-1} E_{P}\left[ | (\hat{e} \vee b_n) - e | \right] \\
        & = r_{\mu,n} b_n^{-1} \left( b_n^{\gamma_0} / \gamma_0  + (\gamma_0-1) \int_{b_n}^{\pi_n} (t - b_n) t^{\gamma_0-2} d t \right) \\
        B_U + B_C & = | B_U + B_C |  = \Theta\big( r_{\mu,n} b_n^{-1} \pi_n^{\gamma_0} \big),
    \end{align*}
    with constants that depend only on $\gamma_0$. (In this construction, $C = C' = 1$.)

    Let $\hat{\sigma}_n^*$ be the standard deviation of this feasible estimator. By the computation of \Cref{lemma:PhiLowerBound}; noting that $\hat{\mu} = o(1)$ contributes a negligible factor to the variance; and adding the extra variance for the event $(\hat{e} \vee b_n) \neq (e \vee b_n)$ which is equivalent to $e \in (b_n, \pi_n)$,
    \begin{align*}
        n (\hat{\sigma}_n^*)^2 & = \Theta\left( E_{P_n}[e / (e \vee b_n)^2]  + \int_{b_n}^{\pi_n} \big( t b_n^{-2} - t^{-1} \big) t^{\gamma_0-2} d t  \right) \\
        & = \Theta\left( b_n^{\gamma_0-2} + b_n^{-2} ( \pi_n^{\gamma_0} - b_n^{\gamma_0} ) \right) = \Theta\left( \pi_n^{\gamma_0} b_n^{-2} \right) \\
        \frac{|B_U + B_C|}{\hat{\sigma}_n^*} & = \Theta\left( n^{1/2} r_{\mu,n} \pi_n^{\gamma_0/2} \right) \gtrsim n^{1/2} r_{\mu,n} r_{e,n}^{\gamma_0/2}  \to \infty. 
    \end{align*}

    Since $n/(n-1) \hat{\sigma}_n^2$ is unbiased for $(\hat{\sigma}_n^*)^2$, $E_{P_n}\big[ \hat{\sigma}_n^2 \big] = (\hat{\sigma}_n^*)^2 (n-1) / n$. By Markov's inequality, $\hat{\sigma}_n^2 = O_{P_n}\big( (\hat{\sigma}_n^*)^2 \big)$, even if the estimated standard error is inconsistent.   Then: 
    \begin{align*}
        P_n( \psi_n \in \hat{\mathcal{C}}_n(\alpha)) & = P_n\left(  \hat{\psi}_{clip}^{AIPW}(b_n) \leq \psi_n + \hat{\sigma}_n z_{1-\alpha/2} \right) - P_n\left(  \hat{\psi}_{clip}^{AIPW}(b_n) \leq \psi_n + \hat{\sigma}_n z_{\alpha/2} \right) \\
        & \leq P_n\left(  \hat{\psi}_{clip}^{AIPW}(b_n) \leq \psi_n + \hat{\sigma}_n z_{1-\alpha/2} \right)  \\
        & = P_n\left(  \frac{\hat{\psi}_{clip}^{AIPW}(b_n) - E\left[ \hat{\psi}_{clip}^{AIPW}(b_n) \right]}{\hat{\sigma}_n^*} \leq \frac{-|B_U + B_C|}{\hat{\sigma}_n^*} + O_{P_n}(1)  \right) \\
        & = P_n\left( O_{P_n}(1) \leq -|B_U + B_C| / \hat{\sigma}_n^*  \right)   \to 0. \qedhere
    \end{align*}
\end{proof}

\subsection{Sufficient Smoothness}\label{proofs:Regression}

\begin{lemma}[Minimal expected treated observations]\label{lemma:MinimalExpectedObs}
    Suppose the conditions of \Cref{prop:GlobalRateAchievable} hold. Then there is a $c_1(C, \gamma_0, d) = ``c_1'' > 0$ such that 
    \begin{align*}
        \inf_{P \in \mathscr{P}} \inf_{x \in [-1, 1]^d} \inf_{h \in (0, 2 \sqrt{d}]} h^{-d \gamma_0 / (\gamma_0 - 1)} E_{P}\left[ D \mathbf{1}\{ \| X - x \| \leq h \} \right] \geq c_1. 
    \end{align*}
\end{lemma}

\begin{proof}[Proof of \Cref{lemma:MinimalExpectedObs}]
    Since $X \sim Unif( [-1, 1]^d )$, there exists a $\hat{c}_1^{Init}(C, \gamma_0, d) > 0$ such that $\inf_{x \in [-1, 1]^d} P( X \in A(x \mid h)) \geq \hat{c}_1^{Init}(C, \gamma_0, d) h^d$ for all $h \in (0, 1]$. Scaling to $\hat{c}_1 = \hat{c}_1^{Init} (2 \sqrt{d} )^{-d}$ ensures that $\inf_{x \in [-1, 1]^d} P( X \in A(x \mid h)) \geq \hat{c}_1 h^d$ for all $h \in (0, 2 \sqrt{d}]$. 
    
    Now let $\pi(C, d, h, \gamma_0)$ solve $C \pi^{\gamma_0-1} = (\hat{c}_1(C, \gamma_0, d) / 2) h^d$, i.e. $\pi(C, d, h, \gamma_0) = (\hat{c}_1(C, \gamma_0, d) / (2C))^{1/(\gamma_0-1)} h^{d/(\gamma_0-1)}$. 
    
    Then, for any $x_0 \in [-1, 1]^d$,
    \begin{align*}
        P\left( X \in A(x_0 \mid h), e(X) \leq \pi(C, d, h, \gamma_0) \right)  & \leq P(e(X) \leq \pi(C, d, h, \gamma_0)) \leq C \pi(C, d, h, \gamma_0)^{\gamma_0-1} \\
        & = \hat{c}_1(C, \gamma_0, d) h^d / 2 \leq P( X \in A(x_0 \mid h) ) / 2.
    \end{align*}
    Therefore 
    \begin{align*}
        E_{P}\left[ D \mid X \in A(x \mid h) \right] & \geq E_{P}\left[ D 1\{ e(X) \geq \pi(C, d, h, \gamma_0) \} \mid X \in A(x \mid h) \right] \\
        & \geq \pi(C, d, h, \gamma_0) P\left( e(X) \geq \pi(C, d, h, \gamma_0)  \mid X \in A(x \mid h) \right) \\
        & \geq \pi(C, d, h, \gamma_0) / 2 =  (\hat{c}_1(C, \gamma_0, d) / (2C))^{1/(\gamma_0-1)} / 2 h^{d/(\gamma_0-1)}.
    \end{align*}
    Taking $c_1( C, \gamma_0, d ) = (\hat{c}_1 / 2)^{\gamma_0/(\gamma_0-1)} C^{-1/(\gamma_0-1)}$ and noting that $P( X \in A(x \mid h ) ) \geq \hat{c}_1 h^d$ completes the claim. 
\end{proof}

\begin{lemma}[Nondegenerate design in expectation]\label{lemma:UniformMinEigenvalue}
    Suppose the conditions of \Cref{prop:GlobalRateAchievable} hold. For any $h > 0$ and $x \in [-1, 1]^d$, define $M_P(h, x) = E_{P}\big[ U\big( (X-x)/h \big) U\big( (X-x)/h \big)^T \mid \| X - x \| \leq h, D = 1 \big]$. For a square matrix $M$, define $\lambda_{\min}(M)$ to be the smallest eigenvalue of $M$. Then there exists a $\lambda^* > 0$ and $\bar{h} > 0$ such that $\inf_{P \in \mathscr{P}} \inf_{x \in [-1, 1]^d} \inf_{h \in (0, \bar{h}]} \lambda_{\min}( M_{P}(h, x) ) \geq \lambda^*$. 
\end{lemma}

\begin{proof}[Proof of \Cref{lemma:UniformMinEigenvalue}]
    Let $dim(U) = p$, let $B_1 = \{ v : \| v \| \leq 1 \}$, and write $\omega_d = \operatorname{Leb}(B_1)$. Also let $\bar{h}$ be half of the largest $h$ small enough for \Cref{assum:NonTrivialConcentration}.

    For any $a \in \R^p$, write $S(a, \epsilon) = \operatorname{Leb}\big( \big\{ v : \| v \| \leq 1, |a^T U(v)| \leq \epsilon \big\} \big)$. Then I claim that  $\lim_{\epsilon \to 0} \sup_{\| a \| = 1} S(a, \epsilon) = 0$. Otherwise, there would be a sequence of $\epsilon_m \to 0$ and a sequence of $a_m$ with $\liminf S(a_m, \epsilon_m) > 0$. But then compactness of $\| a \| = 1$ would imply there is a subsequence with a limiting $a_\infty$. Then, because $\sup_{\|v\| = 1} | a_m^T U(v) - a_\infty^T U(v) | = O\big(  \| a_m - a_{\infty} \| \big) \to 0$, $S(a_\infty, 0) > 0$. But that is impossible because $a_\infty^T U(v)$ is a nonzero polynomial in $v$ by linear independence of $U$.

    Pick some $P \in \mathscr{P}$, $x \in [-1, 1]^d$, $a$ with $\| a \| = 1$, $h \in (0, \bar{h}]$, and $\epsilon > 0$. Let $\mathcal{V} = \big\{ v \in B_1 : x + h v \in [-1, 1]^d \big\}$. Let $s_i = 1$ if $x_i \leq 0$ and $s_i = -1$ otherwise, so that $x_i + s_i h \in [-1, 1]$  for all $h \in [0, 1]$, and the corresponding orthant of $x$ intersected with $B_1$ is in $\mathcal{V}$. Therefore $\operatorname{Leb}(\mathcal{V}) \geq 2^{-d} \omega_d$.

    Now, write $\bar{e} = \sup_{\| x - x' \| \leq h} e(x')$. By \Cref{assum:NonTrivialConcentration}, $E_{P}\big[ e(X) \mid \| X - x \| \leq h \big] \geq \rho \bar{e} P\big( e(X) \geq \rho \bar{e} \mid \| X - x \| \leq h \big)$. Therefore, for any Lebesgue-measurable $\mathcal{S} \subseteq \mathcal{V}$ 
    \begin{align*}
        P\left( \frac{X-x}{h} \in \mathcal{S} \mid D = 1, \| X - x \| \leq h \right) & = \frac{E_{P}\left[ e(X) \mathbf{1}\{ (X-x) / h \in \mathcal{S}, \| X - x \| \leq h \} \right]}{E_{P}\left[ e(X) \mathbf{1}\{  \| X - x \| \leq h \} \right]} \leq \frac{\operatorname{Leb}( \mathcal{S} )}{\rho \nu \operatorname{Leb}( \mathcal{V} )}. 
    \end{align*}

    In particular, there is some $\epsilon > 0$ such that $\sup_{\| a \| = 1} S(a, \epsilon) \leq  \frac{\rho \nu 2^{-d} \omega_d}{2}$. For that $\epsilon$, I obtain $P\left( \left| a^T U\left( \frac{X-x}{h} \right) \right| < \epsilon \mid D=1, \|X - x\| \leq h \right)  \leq 1/2$. Then, taking $\lambda^* = \epsilon^2 / 2$, and 
    \begin{align*}
        a^T M_P(h,x) a & = E_{P}[ (a^T U((X-x)/h))^2 \mid D=1, \|X-x\|\leq h ] \geq \epsilon^2 / 2 = \lambda^*.
    \end{align*}
    Since $a$, $P$, and $x$ were arbitrary, $\lambda_{\min}(M_P(h,x)) \geq \lambda^*$.
\end{proof}

\begin{lemma}[Bandwidth and eigenvalue bounds]\label{lemma:LPUniformBounds}
    Suppose the conditions of \Cref{prop:GlobalRateAchievable} hold, let $c_1$ and $\lambda^*$ be as in \Cref{lemma:MinimalExpectedObs,lemma:UniformMinEigenvalue}, and let $\{ x_{n,j}, k(j) \}$ be constructed as in the proof of \Cref{prop:GlobalRateAchievable}. Let $\ubar{k}_n$ be the largest element of $0 \leq k \leq \log_2(n)$ with $2 \sqrt{d} 2^{-k} \geq \big( 4 c_1^{-1} \log(n)^2 n^{-1}   \big)^{1 / \big(2 \beta_{\mu} + d \gamma_0 / (\gamma_0 - 1) \big)}$. For all $x, h$, define 
    \begin{align*}
        \hat{M}(h, x) = \sum_{i=1}^{n} D_i \mathbf{1}\{ \| X_i - x \| \leq h \} U\left( \frac{X_i - x}{h} \right) U\left( \frac{X_i - x}{h} \right)^T  \left/  \sum_{i=1}^{n} D_i \mathbf{1}\{ \| X_i - x \| \leq h \}. \right. 
    \end{align*}
    Then for any sequence of $P_n \in \mathscr{P}$, with probability tending to one, (i) $\max_j h_{n,k(j)} \leq \bar{h}_n \lesssim \big( \log(n)^2 / n \big)^{1 / (2 \beta_{\mu} + d \gamma_0 / (\gamma_0-1))}$ and (ii) for all $j \geq 1$, $\lambda_{\min}\big( \hat{M}( h_{n,k(j)}, x_{n,j} ) \big) \geq \lambda^* / 2$. 
\end{lemma}

\begin{proof}[Proof of \Cref{lemma:LPUniformBounds}]
    There is at least one such $k$ for all $n$ large enough because $2^{-\log_2(n)} \ll n^{-a} \operatorname{polylog}(n) \ll 2^{0}$ for all $a < 1$. 

    Let $P_n$ be a sequence of distributions in $\mathscr{P}$. Write $\bar{h}_n = h_{n,\ubar{k}_n} = 2 \sqrt{d} 2^{-\ubar{k}_n}$.
    
    By \Cref{lemma:MinimalExpectedObs}, for all $j$, $$E_{P_n}\big[ N_{j,\ubar{k}_n} \big] \geq c_1 n \bar{h}_n^{\frac{d \gamma_0}{\gamma_0 - 1}} \geq 4 \log(n)^2 \big( 4 c_1^{-1} \log(n)^2 n^{-1}   \big)^{\frac{- 2 \beta_{\mu}}{2 \beta_{\mu} + d \gamma_0 / (\gamma_0 - 1) }} \geq 4 \log(n)^2 \bar{h}_n^{-2 \beta_{\mu}}.$$ 

    Note that $N_{j, \ubar{k}_n}$ is a  $Binom\big(n,  P_n( D = 1, \| X - x \| \leq h_{n,\ubar{k}_n} ) \big)$ random variable. Then, by the multiplicative Chernoff bound for binomial random variables, for every $j$, and writing $K(z) = \mathbf{1}\{ \| z \| \leq 1 \}$,
    \begin{align*}
        P_n\left( \frac{\sum D_i K\left( \frac{\| X_i - x_j \|}{h_{n,\ubar{k}_n}} \right)}{n E\left[ D K\left( \frac{\| X - x_j \|}{h_{n,\ubar{k}_n}} \right) \right]} \leq 1/4  \right) & \leq \exp\left( -9 E_{P_n}\big[ N_{j,\ubar{k}_n} \big] / 32 \right) \leq \exp\left( - 9 \log(n)^2 \bar{h}_n^{-2 \beta_{\mu}} / 8 \right).
    \end{align*}

    Then by construction of $\bar{h}_n$, I bound the probability of any $k(j) < \ubar{k}_n$ as
    \begin{align*}
        P_n\left( \min_{j} N_{j,\ubar{k}_n} < \log(n)^2 \bar{h}_n^{-2 \beta_{\mu}}  \right) & \leq P_n\left( \min_{j} N_{j,\ubar{k}_n} - E_{P}[N_{j, \ubar{k}_n}] < -3 \log(n)^2 \bar{h}_n^{-2 \beta_{\mu}}  \right)  \\
        & \leq (2 n + 1)^d  \exp\left( - 9 \log(n)^2 \bar{h}_n^{-2 \beta_{\mu}} / 8 \right) \to 0. 
    \end{align*}
    Then, with probability tending to one, every $j$ has $k(j) \geq \ubar{k}_n$, so that $\max_j h_{n,k(j)} \leq \bar{h}_n = 2 h_{\ubar{k}_n+1} \lesssim \big( 4 c_1^{-1} \log(n)^2 n^{-1}   \big)^{1 / \big(2 \beta_{\mu} + d \gamma_0 / (\gamma_0 - 1) \big)}$.  Thus, claim (i) holds.

    For claim (ii), fix a $(j,k)$ and condition on the event $N_{j,k} = N$. Let $n$ be large enough that $\bar{h}_n \leq \bar{h}$ from \Cref{lemma:UniformMinEigenvalue}. $\hat{M}(h, x)$ is formed via $N$ random draws from the distribution of $X \mid D = 1, \| X - x_{n,j} \| \leq h_k$. The entries of $\big( \hat{M}(h_{k}, x_{n,j}) - M(h_{k}, x_{n,j}) \big) / 4$ are averages of terms between $-1/2$ and $1/2$. Then, applying Hoeffding's inequality to a union bound over the $p^2$ matrix elements,
    \begin{align*}
        P_n\left( \left\| \hat{M}( h_k, x_{n,j} ) - M( h_k, x_{n,j} )  \right\|_{(op)} \geq \lambda^* / 2  \mid N_{j,k} = N  \right) \leq 2 p^2 exp\left( - N (\lambda^*)^2 / (32 p^2) \right).
    \end{align*}
    Proceed on the high-probability event of result (i), so that for all $\ubar{k}_n \leq k \leq k(j)$, $N_n( h_k \mid x_{n,j} ) \geq N_n( h_{n,k(j)} \mid x_{n,j} ) \geq \log(n)^2 h_{n,k(j)}^{-2 \beta_{\mu}} \geq \log(n)^2 \bar{h}_n^{-2 \beta_{\mu}}$ for all $\ubar{k}_n \leq k \leq k(j)$. Then, taking a union bound over a stronger set of $(j, k)$ combinations, on an event $E$ with probability tending to one, I obtain
    \begin{align*}
        & P_n\left( \max_{1 \leq j \leq (2 n + 1)^d} \max_{\ubar{k}_n \leq k \leq k(j)} \left\| \hat{M}( h_k, x_{n,j} ) - M( h_k, x_{n,j} )  \right\|_{(op)} \geq \lambda^* / 2, E  \right) \\
        \leq  & 2 p^2 (2 n + 1)^d \big( 1 + \log_2(n) \big)  \exp\left( - \log(n)^2 \bar{h}_n^{-2 \beta_{\mu}} (\lambda^*)^2 / (32 p^2) \right).
    \end{align*}
    This term is $o(1)$ by the lower bound on $\bar{h}_n$ and $P_n( E ) = 1 - o(1)$ by construction, so that claim (ii) holds by Weyl's inequality.  
\end{proof}

\begin{lemma}[Local polynomial bias-variance bounds]\label{lemma:LPBiasVar}
    Suppose the conditions of \Cref{prop:GlobalRateAchievable} hold and let $\{ \hat{\mu}_{n,j}, x_{n,j}, k(j) \}$ be constructed as in the proof of \Cref{prop:GlobalRateAchievable}. Then there exists finite positive constants $c_{\mu}( d, \beta_{\mu}, L, \lambda^* ) = ``c_{\mu}''$ and $c_{s}( d, \beta_{\mu}, L, \lambda^*, \sigma_{\min} ) = ``c_{s}''$ such that with probability tending to one, (i) for each $n,j$, $\hat{\mu}_{n,j} \mid \{ X_i, D_i \} \sim \mathcal{N}\big( \bar{\mu}_{n,j}, s_{n,j}^2 \big)$ for some $\bar{\mu}_{n,j} \in \R, s_{n,j} > 0$ and (ii) for all $j$,  $| \bar{\mu}_{n,j} - \mu(x_{n,j}) | \leq c_{\mu} h_{n,k(j)}^{\beta_{\mu}}$ and $s_{n,j} \leq c_{s} h_{n,k(j)}^{\beta_{\mu}} / \log(n)$. 
\end{lemma}

\begin{proof}[Proof of \Cref{lemma:LPBiasVar}]
    Proceed under the event from \Cref{lemma:LPUniformBounds}. On this event, local polynomial regression corresponds to a linear smoother, so each $\hat{\mu}_{n,j}$ is a weighted average of $Y_{n,j}$, where the weights only depend on $\{ X_i, D_i \}$ and the $Y_{n,j}$ are independent normal with variance $\sigma_{\min}^2$. Therefore (i) holds for some $\bar{\mu}_{n,j}$, $s_{n,j}$.

    Then, by standard Hölder arguments \citep{Stone1982,TsybakovBook2009} and \Cref{lemma:LPUniformBounds}, $\max_{j} | \bar{\mu}_{n,j} - \mu(x_{n,j}) | h_{n,k(j)}^{-\beta_{\mu}} \leq c_{\mu}$ for some constant $c_{\mu}$ that depends on $d$, $\beta_{\mu}$, $L$, and $\lambda^*$.

    Further, by standard local polynomial arguments and \Cref{lemma:LPUniformBounds}, there exists a $c_s(d, \beta_{\mu}, L, \lambda^*, \sigma_{\min})$ such that $s_{n,j}^2 \leq c_s^2 / N_n( h_{n,k(j)} \mid x_{n,j} ) \leq c_s^2 h_{n,k(j)}^{2 \beta_{\mu}} / \log(n)^2$.
\end{proof}

\begin{proof}[Proof of \Cref{cor:MinimalSmoothnessConditions}]
    Write $\alpha_{\mu} = \frac{\beta_{\mu}}{2 \beta_{\mu} + d \gamma_0 / (\gamma_0 - 1)}$ and $\alpha_{e} = \frac{\beta_e}{2 \beta_e + d}$. Pick some $\gamma_0 > \bar{\gamma}$. 

    By standard arguments \citep{Stone1982} and \Cref{prop:GlobalRateAchievable}, there are cross-fit estimators $\hat{\mu}(X)$ (regressing $Y$ on $X$ among the treated sample) and $\hat{e}(X)$ (regressing $D$ on $X$) that are constructed without knowledge of $\gamma_0$, and polylog functions $\operatorname{polylog}_{\mu}, \operatorname{polylog}_{e}$ such that, with probability tending to one, $\| \hat{\mu} - \mu \|_{L^{\infty}(P)} \leq r_{\mu,n} = n^{-\alpha_{\mu}} \operatorname{polylog}_{\mu}(n)$ and $\| \hat{e} - e \|_{L^{\infty}(P)} \leq r_{e,n} = n^{-\alpha_{e}} \operatorname{polylog}_{e}(n)$. As a result, \Cref{assum:NuisanceRates} holds on an event of probability tending to one. 
    
    Take $b_n = r_{e,n} \log(n)$. Because $\alpha_{e} > 0$, $1 \gg b_n \gg \max\{ r_{e,n}, n^{-1/2} \}$. By \Cref{thm:SecondOrderNuisances}, it only remains to show that the conditions of \Cref{assumption:ConsistencyRatesSufficient} hold.

    Since $p_{\mu} = p_{e} = \infty$, $\zeta = 1$. Recall the results for $\sigma_n = \sqrt{V_n / n}, w_n, \rho_n$ from \Cref{lemma:InvSqBoundAbove,lemma:InvPropMoment,lemma:BoundClipping} and note that  $P_n( e(X) \leq b_n ) = \Theta( b_n^{\gamma_0-1} )$ under \Cref{assum:TailLowerBound}.   Then: 
    \begin{enumerate}[label=(\alph*), itemsep=-0.5ex, topsep=-0.5ex]
        \item \emph{Usual conditions}. By construction, $r_{\mu,n}, r_{e,n} \to 0$. For the product, $r_{\mu,n} r_{e,n} = O\big( n^{-\alpha_{\mu} - \alpha_{e} \min\{ \gamma_0/2, 1\}} \operatorname{polylog}(n) \big) \ll n^{-1/2}$.  
        \item \emph{Control of $B_C$}. $r_{\mu,n} P_n\big( e(X) \leq b_n \big)^{1-1/p_{\mu}} = \Theta\big( r_{\mu,n} b_n^{\gamma_0-1} \big) \lesssim  r_{\mu,n} r_{e,n}^{\min\{\gamma_0, 2\}/2} \sigma_n n^{1/2} \log(n)^2 \ll \sigma_n$ 
        \item \emph{Control of $B_U$}. $r_{\mu,n} r_{e,n} w_n = \Theta\big(  r_{\mu,n} r_{e,n} V_n \big)$ and by inspection of $\gamma_0 < 2$ and $\gamma_0 \geq 2$, $r_{\mu,n} r_{e,n} \rho_n^{1/\zeta} / b_n = O\big( r_{\mu,n} r_{e,n} V_n  \big)$. Then, $r_{\mu,n} r_{e,n} \big( w_n + \rho_n^{1/\zeta} / b_n \big)  =     \Theta\big( n^{1/2} r_{\mu,n} r_{e,n} V_n^{1/2} \sigma_n \big)$. I claim this term is $o(\sigma_n)$. If $\gamma_0 \geq 2$, then $V_n^{1/2} = O( \log(n)^{1/2} )$ and the claim holds by $r_{\mu,n} r_{e,n} \ll n^{-\alpha}$ for some $\alpha > 1/2$. If $\gamma_0 < 2$, then $V_n^{1/2} = \Theta\big( b_n^{\gamma_0/2 - 1} \big) \lesssim r_{e,n}^{\gamma_0/2-1}$ and the claim holds by $r_{\mu,n} r_{e,n}^{\gamma_0/2} \ll n^{-1/2}$.  
        \item \emph{Control of $m_n$}. Since $p_e = \infty$, $r_{e,n} \ll b_n^{1+1/p_e} V_n^{1/p_e} = b_n$. Since $p_{\mu} = \infty$ and $r_{\mu,n} \to 0$, $\| \hat{\mu} - \mu \|_{L^\infty(P_n)} = o(1)$.  
    \end{enumerate}
    As a result, by \Cref{thm:SecondOrderNuisances}, the result for Wald confidence interval validity holds. 
\end{proof}

\subsection{Choice of Threshold and Simulation}\label{proofs:RulesOfThumb}

\begin{proof}[Proof of \Cref{lemma:WellDefinedRuleOfThumb}]
    As noted in the text, $n^{-1/2} \sqrt{\frac{1}{n} \sum D / \max\{ \hat{e}, b \}^2}$ is decreasing in $b$ and $\frac{1}{n} \sum \big( D / \max\{ \hat{e} / b, 1 \}^2 + b \mathbf{1}\{ \hat{e} \leq b \} \big)$ is increasing in $b$, so that this set is convex and must be an interval.

    It remains to show the set is bounded and nonempty. Write $r_n(b) \equiv n^{-1/2} \sqrt{\frac{1}{n} \sum D / \max\{ \hat{e}, b \}^2} - \frac{1}{n} \sum \big( D / \max\{ \hat{e} / b, 1 \}^2 + b \mathbf{1}\{ \hat{e} \leq b \} \big)$.  For all $b > 1$, $r_n(b) \leq n^{-1/2} / b - b < 0$, so the interval is bounded. For $b < \inf \hat{e}$ small enough, $r_n(b) \geq  n^{-1/2} \sqrt{\frac{1}{n} \sum D / \hat{e}^2} -  \big( \sum D \big) \big/ \big( n \inf (\hat{e} / b)^2 \big) > 0$. Therefore the interval is non-empty. 
\end{proof}

\begin{proof}[Proof of \Cref{prop:CoverWithEstimatedThresh}]
    Let $P_n$ be arbitrary. By the arguments of \Cref{cor:MinimalSmoothnessConditions}, with probability tending to one, $\hat{\mu}$ and $\hat{e}$ satisfy \Cref{assum:NuisanceRates} with $p_{\mu} = p_e = \infty$ (so $\zeta = 1$) and $\max\{ r_{\mu,n}, r_{e,n} \} \lesssim n^{-a}$ for some $a > 1 / (2 + \min\{\gamma_0, 2\})$, with $\hat{e}$ depending only on $\{X_i, D_i\}$. It suffices to show that conditional on this event (with the event implicit), $P_n( T_n \leq t ) = \Phi(t) + o(1)$, where $T_n = \big( \hat{\psi}_{clip}^{AIPW}(\hat{b}_n) - \psi_n \big) / \hat{\sigma}_n(\hat{b}_n)$.  Write $\ell(b) = \log(1/b)^{\mathbf{1}\{\gamma_0 = 2\}}$.

    First, I define a target threshold and characterize its order. For any $b \in (0, 1]$, write $S_n^*(b) = E_{P_n}[ D / (e \vee b)^2]$, $g_n^*(b) = P_n( e(X) \leq b )$, and $f_n^*(b) = n^{-1/2} \sqrt{S_n^*(b)} - b^2 S_n^*(b) - b g_n^*(b)$. Let $\hat{S}_n(b)$, $\hat{g}_n(b)$, and $\hat{f}_n(b)$ be the sample analogs at $\hat{e}$, so that Algorithm \ref{alg:errorUpperBound} returns $\hat{b}_n = \sup\{ b : \hat{f}_n(b) \geq 0 \}$.  Define $b_n^* = \sup\{ b : f_n^*(b) \geq 0 \}$. By \Cref{lemma:InvSqBoundAbove}, \Cref{lemma:useful_inequalities}, and \Cref{assum:TailLowerBound}, for any sequence of $b \to 0$, $S_n^*(b) \asymp \ell(b) b^{\min\{ \gamma_0 - 2, 0 \}}  \gtrsim \max\{ 1, b^{\gamma_0 - 2} \}$ and $b g_n^*(b) \asymp b^{\gamma_0} \lesssim b^2 S_n^*(b)$. By monotonicity, right-continuity, and $S_n^*(b/2) \leq 4 S_n^*(b)$,  $f^*(2 b_n^*) < 0 \leq f_n^*(b_n^* / 2)$ and $f_n^*(b_n^*) \leq 0$, so
    \begin{align*}
        n^{-1/2} \sqrt{S_n^*(b_n^*)} & = \Theta\left( (b_n^*)^2 S_n^*(b_n^*) + b_n^* g_n^*(b_n^*) \right) = \Theta\left( (b_n^*)^2 S_n^*(b_n^*) \right) \\
        n^{-1/2} & = \Theta\left( (b_n^*)^2 \sqrt{S_n^*(b_n^*)}  \right) = \Theta\left( (b_n^*)^{\min\{ \gamma_0 + 2, 4 \} / 2} \sqrt{\ell(b_n^*)}  \right),
    \end{align*}
    i.e. $b_n^* = \Theta\big( n^{-1 / (2 + \min\{\gamma_0, 2\})} \big)$ up to polylog factors at $\gamma_0 = 2$.  In particular, $n^{-1/2} \ll b_n^* \ll 1$ and $\max\{ r_{\mu,n}, r_{e,n} \} \ll b_n^*$. Write $\sigma_n = \sigma_n(b_n^*)$. $\sigma_n = \Theta( n^{-1/2} \sqrt{S_n^*(b_n^*)} )$ by \Cref{lemma:PhiLowerBound}.  

    Next, I show that $\hat{b}_n = b_n^* + o_{P_n}(b_n^*)$.  Fix some $\epsilon \in (0, 1/2]$, and let $n$ be large enough that $r_{e,n} \leq \epsilon b_n^* / 4$. Take $b_n^+ = b_n^* (1 + \epsilon)$ and $b_n^- = b_n^* (1 - \epsilon)$. Write $\tilde{g}_n(b) =  \mathbb{P}_n[ \mathbf{1}\{ e \leq b \} ]$.  For $b \in \{ b_n^-, b_n^+ \}$, $|\hat{e} \vee b - e \vee b| \leq r_{e,n} = o(b)$, so $\hat{S}_n(b) = \mathbb{P}_n[ D / (e \vee b)^2 ]  (1+o(1))$. By Chebyshev's inequality, for any such $b$, $n^{-1/2} \sqrt{\hat{S}_n(b)}$ and $b^2 \hat{S}_n(b)$ equal their population analogs up to $o_{P_n}(\sigma_n)$. Note also  that $\hat{g}_n(b_n^+) \geq \tilde{g}_n(b_n^*)$ and $\hat{g}_n(b_n^-) \leq \tilde{g}_n\big( b_n^*(1-3\epsilon/4) \big)$.  By Chebyshev's inequality, for $b \leq b_n^*$, $b \tilde{g}_n(b) = b g_n^*(b) + O_{P_n}(b n^{-1/2}) = b g_n^*(b) + o_{P_n}(\sigma_n)$. Note that by the lower-bound argument of \Cref{lemma:InvSqBoundAbove}, there is a $c > 0$ such that for all $\epsilon \in (0, 1]$ and all $b > 0$ small enough, 
    \begin{align*}
         (1+\epsilon)^2 b^2 S_n^*\left( (1+\epsilon) b \right) - b^2 S_n^*( b ) & \geq \left( (1+\epsilon)^2 - 1 \right) b^2 E_{P_n}[\mathbf{1}\{ e \geq 2 b  \} / e] \geq c \epsilon b^2 S_n^*( b ).
    \end{align*}
    Then:
    \begin{align*}
        \hat{f}_n(b_n^+) %& = n^{-1/2} \sqrt{\hat{S}_n(b_n^+)} - (b_n^+)^2 \hat{S}_n(b_n^+) - b_n^+ \hat{g}_n(b_n^+) \\
        & \leq n^{-1/2} \sqrt{S_n^*(b_n^*)} - \left( 1 + c \epsilon \right) ( b_n^* )^2 S_n^*(b_n^*) - b_n^* g_n^*(b_n^*) + o_{P_n}(\sigma_n) \leq - c \epsilon (b_n^*)^2 S_n^*(b_n^*) + o_{P_n}(\sigma_n).
    \end{align*}
    $c \epsilon (b_n^*)^2 S_n^*(b_n^*) = \Theta\big( \sigma_n \big)$, so $\hat{f}_n(b_n^+) < 0$ with probability tending to one. On the other hand, 
    \begin{align*}
        \hat{f}_n(b_n^-) %& =  n^{-1/2} \sqrt{\hat{S}_n(b_n^-)} - (b_n^-)^2 \hat{S}_n(b_n^-) - b_n^- \hat{g}_n(b_n^-) \\
        & \geq %n^{-1/2} \sqrt{S_n^*\big( b_n^*(1-3\epsilon/4) \big)} - (b_n^-)^2 S_n^*(b_n^-) - b_n^*(1-3\epsilon/4) g_n^*\big( b_n^*(1-3\epsilon/4) \big) - o_{P_n}(\sigma_n) \\ 
        %& = 
        f_n^*\big( (1-3\epsilon/4) b_n^* \big) + \big( (1-3\epsilon/4) b_n^* \big)^2 S_n^*\big( (1-3\epsilon/4) b_n^* \big) - (b_n^-)^2 S_n^*(b_n^-) - o_{P_n}(\sigma_n)  \\
        & \geq f_n^*\left( b_n^*(1-3\epsilon/4) \right)  + \frac{c \epsilon}{16} (b_n^*)^2 S_n^*(b_n^*) - o_{P_n}(\sigma_n).
    \end{align*}
    Therefore $\hat{f}_n(b_n^-) \geq \frac{c \epsilon}{16} (b_n^*)^2 S_n^*(b_n^*) - o_{P_n}(\sigma_n)$ is strictly positive with probability tending to one, so that $|\hat{b}_n / b_n^* - 1| \leq \epsilon$ with probability tending to one. But $\epsilon$ was arbitrary. Therefore $\hat{b}_n = b_n^* + o_{P_n}(b_n^*)$.

    Now I show oracle equivalence at $b_n^*$. Fix a fold $k$, suppress the index. \Cref{assumption:ConsistencyRatesSufficient} holds for $p_{\mu} = p_{e} = \infty$ because $g_n^*(b_n^*)/b_n^* + w_n + \rho_n / b_n^* = O\big( S_n^*(b_n^*) \big)$ by the $g_n^*(b_n^*)$ orders above and applying \Cref{lemma:InvPropMoment,lemma:BoundClipping,lemma:InvSqBoundAbove}, and noting $\max\{ r_{\mu,n}, r_{e,n} \} \ll b_n^* \lesssim n^{-1/4}$. By  \Cref{thm:SecondOrderNuisances} and \Cref{lemma:estimated_variance_consistency}, $\hat{\sigma}_n(b_n^*) / \sigma_n \to_{P_n} 1$ and $\hat{\sigma}_n(b_n^*)^{-1} \big( \hat{\psi}_{clip}^{AIPW}(b_n^*) - \psi(P_n) \big) \to^d \mathcal{N}(0, 1)$. 
    
    Now I show oracle equivalence of $\hat{b}_n$ relative to $b_n^*$. Write $D_n = \hat{\psi}_{clip}^{AIPW}(\hat{b}_n) - \hat{\psi}_{clip}^{AIPW}(b_n^*)$ and $\Delta = D / (\hat{e} \vee \hat{b}_n) - D / (\hat{e} \vee b_n^*)$. On the high-probability-of-$\mathcal{G}$ event that $\hat{b}_n / b_n^* \in [1/2, 3/2]$ and $r_{e,n} \leq b_n^* / 2$, $\Delta$ can only be non-zero for $e \leq 2 b_n^*$, so that by Markov's inequality, $\mathbb{P}_{n_k}\big[ D \mathbf{1}\{ e \leq 2 b_n^* \} \big] = O_{P_n}\big( (b_n^*)^{\gamma_0} \big)$. Then 
    \begin{align*}
        \left| \Delta \right| & \leq D \frac{|b_n^* - \hat{b}_n|}{\hat{b}_n b_n^*} \leq 2 D \mathbf{1}\{ e \leq 2 b_n^* \}  \left. \left( |\hat{b}_n - b_n^*| \right) \right/ (b_n^*)^2  \\
        \left| E_{P_n}\left[ D_n \mid \mathcal{G} \right] \right| &
          \leq r_{\mu,n} \mathbb{P}_{n_k}[ |\Delta| ] \leq \frac{2 r_{\mu,n} |\hat{b}_n - b_n^*|}{(b_n^*)^2} \mathbb{P}_{n_k}[ D \mathbf{1}\{ e \leq 2 b_n^* \} ] = o_{P_n}\left( r_{\mu,n} (b_n^*)^{\gamma_0-1} \right) = o_{P_n}(\sigma_n) \\
        Var_{P_n}( D_n \mid \mathcal{G} ) & \leq \frac{M^2}{n_k}  \mathbb{P}_{n_k}\left[ \Delta^2 \right] \leq \frac{4 M^2 (\hat{b}_n - b_n^*)^2}{n_k (b_n^*)^2} \mathbb{P}_{n_k}\left[ D / (b_n^*)^2 \mathbf{1}\{ e \leq 2 b_n^* \} \right] = o_{P_n}(\sigma_n^{2}). 
    \end{align*}
    By Chebyshev conditional on $\mathcal{G}$ and bounded convergence, $D_n = o_{P_n}(\sigma_n)$ for each fold. By averaging across folds, $\hat{\psi}_{clip}^{AIPW}(\hat{b}_n) = \hat{\psi}_{clip}^{AIPW}(b_n^*) + o_{P_n}(\sigma_n)$.

    Now I show consistent standard errors. Write $\Gamma = | \Delta (Y - \hat{\mu}) |$. 
    By the bound on $\mathbb{P}_{n_k}[\Delta^2]$ above, $E_{P_n}\big[ \mathbb{P}_{n_k}[\Gamma^2] \mid \mathcal{G} \big] \leq (M^2 + r_{\mu,n}^2) \mathbb{P}_{n_k}[\Delta^2] = o_{P_n}\big( (b_n^*)^{\gamma_0-2} \big) = o_{P_n}(V_n)$. By Markov's inequality and bounded convergence, $\mathbb{P}_{n_k}[ \Gamma^2 ] = o_{P_n}(V_n)$. Then,
    \begin{align*}
        \left| \mathbb{P}_{n_k}[ \phi(Z \mid \hat{b}_n, \hat{\eta})^2 - \phi(Z \mid b_n^*, \hat{\eta})^2 ]   \right| & \leq \mathbb{P}_{n_k}[ \Gamma^2 ] + 2 \mathbb{P}_{n_k} [ \Gamma \ | \phi(Z \mid b_n^*, \hat{\eta}) | ] = o_{P_n}(V_n)
    \end{align*}
    by Cauchy-Schwarz and observing $\mathbb{P}_{n_k}[ \phi(Z \mid b_n^*, \hat{\eta})^2 ] = O_{P_n}(V_n)$ under the above arguments.  Averaging over folds and using $\hat{\psi}_{clip}^{AIPW}(\hat{b}_n) - \hat{\psi}_{clip}^{AIPW}(b_n^*) = o_{P_n}(1)$ and $V_n \gtrsim 1$ yields $\hat{\sigma}_n(\hat{b}_n) = \sigma_n + o_{P_n}(\sigma_n)$.

    Coverage now follows. For any fixed $a \in (0, 1)$,
    \begin{align*}
        P_n\left( \hat{\psi}_{clip}^{AIPW}(\hat{b}_n) - \psi(P_n) \leq  z_{a} \hat{\sigma}_n(\hat{b}_n)  \right) = P_n\left( \sigma_n^{-1} \left( \hat{\psi}_{clip}^{AIPW}(b_n^*) - \psi(P_n) \right) \leq  z_{a} + o_{P_n}(1)  \right) = a + o(1).
    \end{align*}
    In particular, $P_n( \psi(P_n) \in \hat{\mathcal{C}}_n(\alpha)) = 1-\alpha + o(1)$. 
\end{proof}

\begin{lemma}[Simulation variance comparison]\label{lemma:VarianceComparison}
    Consider a tweaked version of the simulation specification of \Cref{sec:Sims} as a function of $\gamma_0$ where $\mu(X) = \kappa(\gamma_0) ( E[e(X)] - e(X) )$ for $\kappa = 1 / \sqrt{Var( e(X) )}$ changes with $\gamma_0$. Fix some $b > 0$ and $\gamma_0 \neq 2$. Let $V_{clip}$ be the asymptotic variance of the clipped-at-$b$ AIPW estimator with oracle nuisances and let $V_{fixed}$ be the asymptotic variance of the fixed trimming estimator with oracle nuisances. Write $d_m = 0.9^m - 0.1^m$. Then $V_{clip} / V_{fixed} = \big( c_1 + c_2 b^{\gamma_0-2} \big) \big/ c_3$, where $c_1 = (3 - 2 \gamma_0) d_{\gamma_0-1}^2 / (2 - \gamma_0)$, $c_2 = 2 (\gamma_0-1) d_{\gamma_0-1}^2 \big/ \big( \gamma_0 (2 - \gamma_0) \big)$, and $c_3 = \gamma_0^2  d_{\gamma_0+1} - (\gamma_0^2 - 1) d_{\gamma_0}^2 / d_{\gamma_0-1} - (\gamma_0-1) d_{\gamma_0-2} / (2 - \gamma_0)$. 
\end{lemma}

\begin{proof}[Proof of \Cref{lemma:VarianceComparison}]
    By the proof of \Cref{lemma:IntegralOrder}, the clipped AIPW asymptotic variance is 
    \begin{align*}
        V_{clip} & = Var(\mu(X)) + E\left[ \frac{e(X)}{(e(X) \vee b)^2} Var( Y \mid X, D = 1 ) \right]  \\
        & = 1 + b^{-2} E\left[ e(X) \mathbf{1} \{ e(X) \leq b \} \right] + E\left[ e(X)^{-1} \mathbf{1} \{ e(X) \geq b \} \right] \\
        & = \frac{3 - 2 \gamma_0}{2-\gamma_0} + \frac{2 (\gamma_0-1)}{\gamma_0 (2 - \gamma_0)} b^{\gamma_0-2}. 
    \end{align*}
    Note that $P(e(X) \in [0.1, 0.9]) = d_{\gamma_0-1}$. Also note that $\kappa^2 = 1 / Var( e(X) ) = \gamma_0^2 (\gamma_0 + 1) / (\gamma_0 - 1)$. Write $W_k = e(X)^k \mathbf{1}\big\{ e(X) \in [0.1, 0.9] \big\}$, so that $E[W_k] = (\gamma_0-1) d_{\gamma_0+k-1} / (\gamma_0 + k - 1)$. 
    Then
    \begin{align*}
        Var\left( e(X) \mid e(X) \in [0.1, 0.9] \right) & = \left. \left( d_{\gamma_0-1} E[W_2] - E[W_1]^2 \right) \right/ d_{\gamma_0-1}^2.
    \end{align*}
    Then fixed trimming asymptotic variance is 
    \begin{align*}
        V_{fixed} & = \left. Var\left( \mu(X) + D (Y - \mu(X)) / e(X) \mid e(X) \in [0.1, 0.9] \right) \right/ P\left( e(X) \in [0.1, 0.9] \right) \\
        & = \left( \kappa^2 d_{\gamma_0-1} Var\left( e(X) \mid e(X) \in [0.1, 0.9] \right) +  E\left[ \mathbf{1}\{ e(X) \in [0.1, 0.9] \} / e(X) \right] \right) / d_{\gamma_0-1}^2  \\
        & = \left. \left( \gamma_0^2 d_{\gamma_0+1} - (\gamma_0^2 - 1) d_{\gamma_0}^2 / d_{\gamma_0-1} - \frac{(\gamma_0-1) d_{\gamma_0-2}}{2 - \gamma_0} \right) \right/ d_{\gamma_0-1}^2. 
    \end{align*}
    The claim is then true by inspection. 
\end{proof}

\end{document}